\documentclass[12pt,a4paper]{article}
\usepackage{setspace}
\usepackage[notablist, nofiglist,tablesfirst, heads, markers]{endfloat}

\usepackage[margin=1in]{geometry}

\usepackage[T1]{fontenc}
\usepackage{amsmath}
\usepackage{amssymb}
\usepackage{bm,lscape}
\usepackage{longtable}
\usepackage{mathrsfs,dsfont}
\usepackage{graphicx}
\usepackage{eurosym}
\usepackage{tablefootnote}
\usepackage{footmisc}
\usepackage{subfig}
\usepackage{listings}
\usepackage{adjustbox}
\usepackage{booktabs}
 \usepackage[table,xcdraw]{xcolor}
\RequirePackage[colorlinks,citecolor=blue,urlcolor=blue]{hyperref}
\RequirePackage{natbib}
\usepackage[english]{babel}
\usepackage{amsmath,amsthm}
\usepackage{amssymb}
\usepackage{multirow}
\usepackage{bbm}
\usepackage{bm}
\usepackage{hyperref}
\usepackage{mathrsfs}
\usepackage{algorithm}
\usepackage{dsfont}
\usepackage{color}

\usepackage[noend]{algpseudocode}

\makeatletter
\def\BState{\State\hskip-\ALG@thistlm}
\makeatother

\theoremstyle{definition}
\newtheorem{example}{{\rm \bf Example}}
\newtheorem{proposition}{{\rm \bf Proposition}}
\newtheorem{corollary}{Corollary}
\newtheorem{theorem}{Theorem}

\def\bSig\mathbf{\Sigma}

\begin{document}
%

\begin{center}
\textmd{\LARGE{\bfseries{{A Novel  Multi-Period and Multilateral Price Index}}}}
\end{center}
\medskip

\begin{center}
\large{Consuelo R. Nava$^1$, \enspace
Maria Grazia Zoia$^{2*}$ }\\ \vspace{7mm}
\begin{small}$^1$ Department of Economics and Statistics ``Cognetti de Martiis'', Univerist\`a degli Studi di Torino, Lungo Dora Siena 100A, Torino, Italy. Email: \href{consuelorubina.nava@unito.it}{consuelorubina.nava@unito.it}\vspace{3mm}
\\$^{2}$ Department of Economic Policy, Universit\`a Cattolica del Sacro Cuore, Largo Gemelli 1, 20123, Milano, Italy.  Corresponding Author. Tel:+390272342948; fax: +390272342324. Email: \href{maria.zoia@unicatt.it}{maria.zoia@unicatt.it} 
\vspace{3mm}
\\$^{*}$ Corresponding author\end{small}
\end{center}



\smallskip

\begin{quote}
\begin{center}
\noindent {\sl Abstract} \end{center}
 A novel approach to price indices, leading to an innovative solution in both a multi-period or a multilateral framework, is presented. The index turns out to be the generalized least squares solution of a regression model linking values and quantities of the commodities. The index reference basket, which is the union of the intersections of the baskets of all country/period taken in pair, has a coverage broader than extant indices.  The properties of the index are investigated and updating formulas established. Applications to both real and simulated data provide evidence of the better index performance in comparison with extant alternatives. 
\vspace{10pt}
\par 
\end{quote}


\newpage

\section{Introduction}{\label{sec: intro}}
Multi-period and multilateral price indices, used to compare sets of commodities over time and across countries respectively, are of prominent interest for statisticians  \citep[see, e.g.,][]{Biggeri2010}. Several approaches to the problem have been carried out in the literature. 

One of these is the axiomatic approach \citep[see, e.g.,][and the references quoted therein]{balk1995axiomatic}, which rests on the availability of both quantities and prices and aims at obtaining price indices enjoying suitable properties \citep[][]{fisher1921best, fisher1922making}.

A second approach hinges on the economic theory\footnote{This approach is also known as preference field approach or functional approach \citep{divisia1926indice}.} \citep[see, among others,][for a review]{diewert1979economic, caves1982economic} and rests on the idea that consumption choices come from the optimization of a utility function under budget constraints. Here, prices play the role of independent variables, while  quantities arise as solutions to an optimization problem in accordance with the preference scheme OF decision makers.

A third approach is the stochastic one \citep[see][for a  review] {clements2006stochastic, diewert2010stochastic}, which can be traced back to the works of \citet{jevons1863serious, jevons1869depreciation} and  \citet{edgeworth1887report, edgeworth1925memorandum}. Thanks to \citet{balk1980method} and \citet{ clements1987measurement}, this approach has been recently  reappraised, and its role in inflation measurements duly acknowledged \citep[see, e.g.,][and references quoted therein]{zahid2010measuring}. In this framework, prices are assumed to be affected by  measurement errors whose bias effect must be duly minimized. 

The stochastic approach (hereafter, SA) turns out to be somewhat different from other approaches, insofar as it is closely related to regression theory \citep{theil1960best, clements1987measurement}. In fact, the SA enables the construction of tests and confidence intervals for price indices, which provide useful pieces of information \citep{clements2006stochastic}. Furthermore, the SA has less limits than other approaches\footnote{The SA, differently from the index number theory  does not need to account for the economic importance of single prices.} and clears the way to further extensions, as shown in \cite{diewert2004stochastic, diewert2005weighted, silver2009hedonic, rao2004country}. 

In this paper, we devise a multi-period/multilateral price index, MPL index henceforth, within the stochastic framework. The derivation of the MPL index, which is the solution to an optimization problem, calls for quantities and values of the commodities (not prices), like \citet{walsh1901measurement}. In fact, the MPL index is obtained by applying generalized least squares (GLS) methods to a regression model linking values and quantities of the commodities. In the two-period (country) case the index turns out to be the ratio of weighted sums with harmonic means of the squared quantities as weights. Depending on the choice of the objective function to optimize, the commonly used indices, namely Laspeyeres, Paasche, Marshall-Edgeworth, Walsh and Geary Khamis, arise as special cases. 

The reference basket of the MPL index, namely the set of commodities for all periods/countries, is made up of the union of the intersections of all the couples of year/country baskets in pairs. This implies that the price index of a commodity can be always computed once the latter is present in at least two periods/countries. Thus, the reference basket turns  out to be more representative than the ones commonly used by the majority of statistical agencies, which either align the reference basket to that of the first period, or make it tally with the intersection of the commodity sets of all periods/countries. Eventually, such a reference basket is likely to be scarcely representative of the commodities present in each period/country.	 
In this sense, just like hedonic \citep{pakes2003reconsideration}, GESKS \citep{balk2012price} and country/time-product-dummy (CPD/TPD) approaches with incomplete price tableau \citep{rao2004country}, the MPL index does not drop any observation on the account of having no counterpart in the reference basket. Indeed, unlike the aforesaid approaches, the MPL index is built on quantities and values, not on prices. Accordingly, the lack of a commodity in a period/country  implies setting its quantity and value equal to zero in that period/country, not its price which, being not observed, must be considered unknown but not necessarily null. 
Neither any preliminary computation of binary price indices, as in the GESKS approach \citep{ivancic2011scanner}, nor the use of any type of weighting matrix for dealing with missing values or quantities, as in the case of CPD/TPD indices, are needed.

The updating of the MPL index is easy to accomplish and suitable formulas, tailored to the multi-period or multilateral nature of the data, are provided. 
In fact, while the inclusion of fresh values and quantities, of a set of commodities corresponding to an extra period, does not affect the previous values of the MPL index, the inclusion of a new country affects all former MPL indices. Hence, two updating formulas have been proposed for the MPL index: one for the multi-period case and another for the multilateral case. Closed-form expressions for the standard errors of the MPL estimates are provided and the properties of the estimators are duly investigated. 
An empirical comparison of the MPL index to CPD/TPD index -- a multilateral/multi-period index that, like the MPL one, can be read as a solution to an optimization problem -- provides evidence of an easier implementation and greater efficiency of the former index.

To sum up, a threefold  novelty characterizes the paper. First, it proposes a price index, which proves effective either for the multi-period or the multilateral case.  
Second, updating formulas tailored to the multilateral and the multi-period version of the index are provided.
Third, the grater simplicity of use and efficiency of the said index is highlighted in comparison with well-known standard multilateral/multi-period indices. 
Furthermore, the approach employed to build the MPL index yields the so called reference prices,  which are the prices expected to be paid for the commodities in the base time/country. The latter together with the values of MPL index allow to  determine the prices of those commodities that, for whatever reason, can not be observed in a given period or country. \\
The MPL index proves to be particularly useful when i) there is the need of a reference basket more representative than the mere intersection of the baskets involved in all periods/countries or in presence of ii) historical data; iii) when the prices of some commodities in some periods/countries are unknown. In a multi period-perspective, this may occur, for instance, when some commodities enter or leave the basket as a consequence of a technological change. The knowledge of prices is lacking in the periods that are antecedent the inclusion of commodities in the basket or in the periods that are subsequent their exclusion from the basket. In a multilateral perspective, this happens when a commodity is not dealt in yet on the market of a given country.

The paper is organized as follows.  In Section~\ref{sec:mpl}, within the SA, we devise the MPL index according to a minimum-norm criterion as well as its updating formulas for the multi-period and the multilateral cases, respectively. 
Section~\ref{subsec:prop} is devoted to the properties of the MPL index. Section~\ref{sec:empirical} provides an application of the MPL index to the Italian cultural supply data to shed light on its potential as both a multi-period and a multilateral index. To gain a better insight into the performance of the MPL index, a comparison with the CPD/TPD indices is made   by using both real and perturbed data. Section~\ref{sec:simulation} enriches previous empirical evidences with a simulation example, while  Section~\ref{sec:conclusion} completes the paper with some concluding remarks and hints. For the sake of easier readability, an Appendix has been added with proofs and technicalities.

\section{The MPL index as solution to an optimization problem }\label{sec:mpl}
In this section, taking the SA as the reference frame, we derive a multi-period/multilateral price index  whose reference basket -- over a set of periods or across a set of countries -- is  the union of the intersections of the commodity baskets of various periods/countries, taken in pairs. Such a reference basket proves to be an effective solution for several reasons. First, it is broader and more representative than the ones built on the intersection of commodities which are present in all periods/countries. Second, the price index is well defined, provided each commodity is present in at least two baskets.\footnote{Some similarities arise with the chaining rule \citep{forsyth1981theory, von2001chain} where the price index is a measure of the cumulative effect of adjacent periods from 0 to 1, 1 to 2, $\dots$, $t-1$ to $t$. Thus, chain indices compare the current and the previous periods in order to evaluate the evolution over many periods (for a comparison of this approach with the fixed base one see \cite{diewert2001consumer}). However, chain indices, unlike the MPL index, leave unresolved the reference basket updating and are not applicable in a multilateral perspective. } The computation of the MPL price index hinges on quantities and values of the commodities, not prices. The lack of a commodity in a given period/country $t$ entails that both its quantity and value vanish in that period/country.

Figure~\ref{fig:1} shows the reference basket corresponding to the usual approach as compared with that devised in the paper for the case of two and three periods/countries.
\begin{figure}[htbp]
\begin{center}
\includegraphics[scale=0.3]{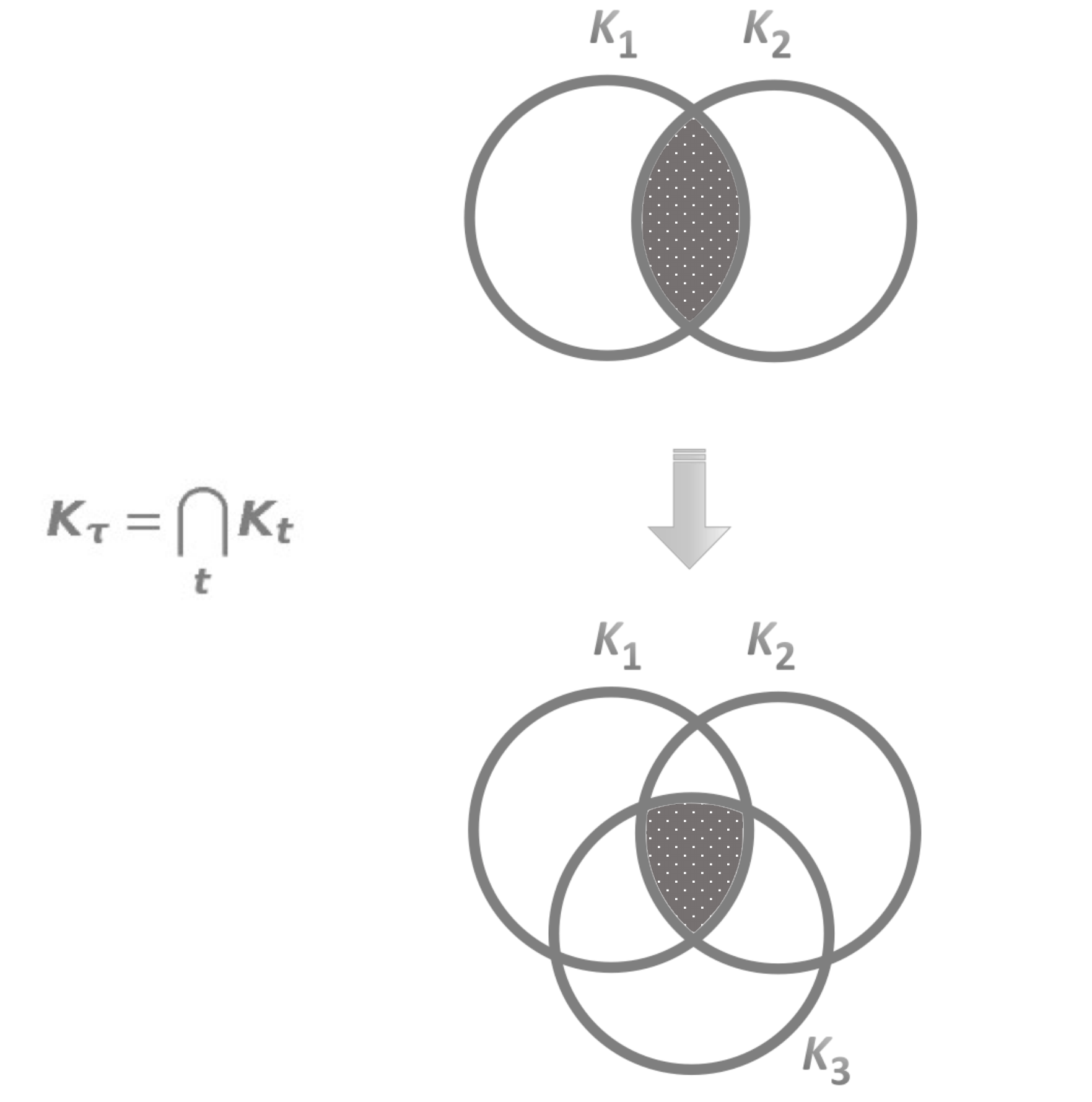} 
\includegraphics[scale=0.3]{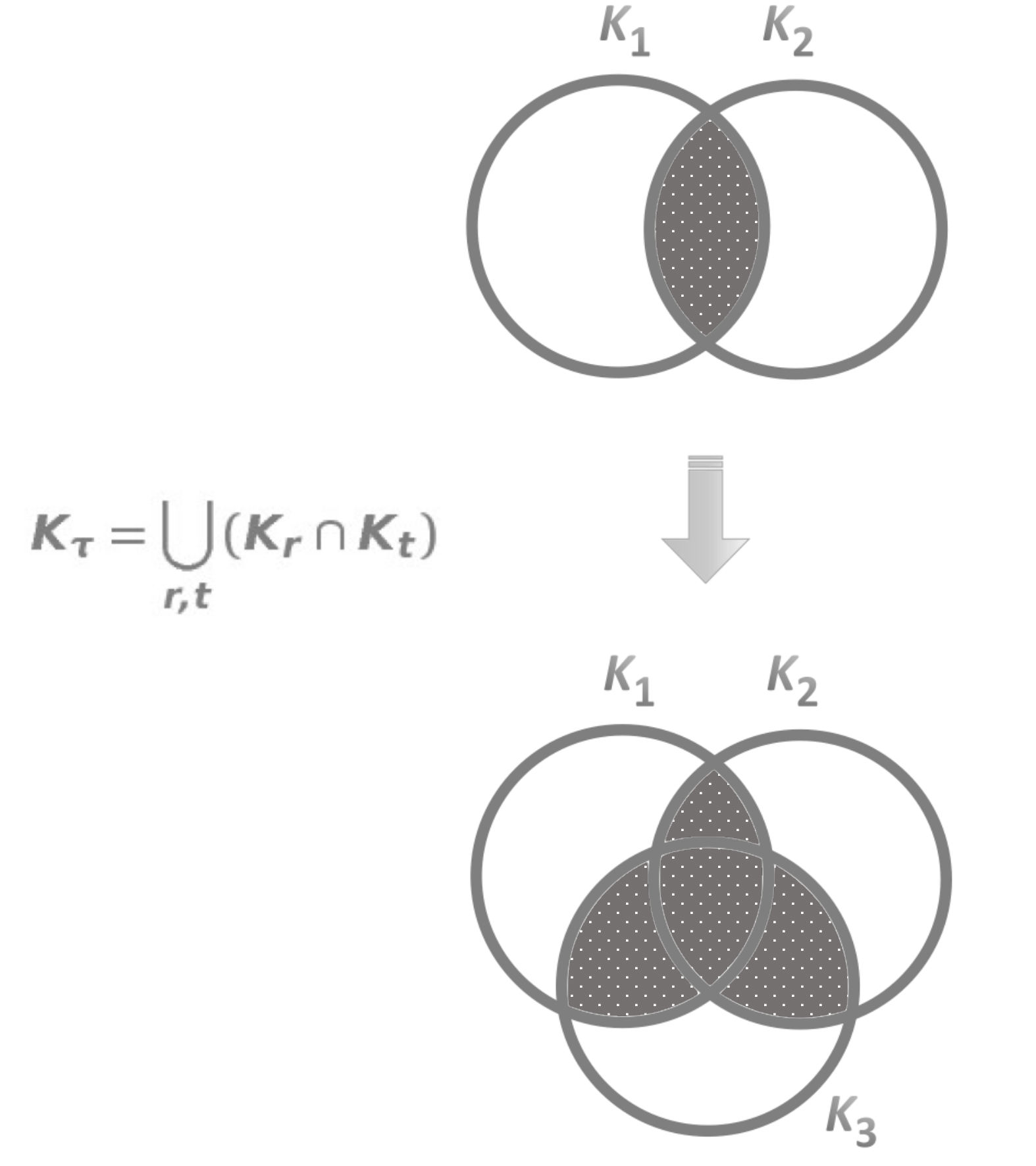}
\caption{
The left-hand side panel shows the reference basket corresponding to the ``traditional'' approach. The right-hand side panel shows the reference basket corresponding to the MPL index. }
\label{fig:1}
\end{center}
\end{figure} 

 According to the SA, the MPL index is worked out as solution of an optimization problem consisting in finding an hyperplane lying as close as possible to the points whose coordinates are the $N$ commodities prices in $T$ periods. In fact, it hinges on the idea that in each time/country $t$, the $N$ commodity prices  move proportionally to a  set of $N$ reference  prices to within  ``small'' discrepancies, that is
\begin{equation}\label{eq:multi1}
\underset{(N,1)}{\boldsymbol{p}_t}\,\approx\, \lambda_t  \underset{(N,1)}{\boldsymbol{\tilde{p}}}\,\,\,\,\,\,\,\,\,\,\,\,\forall\,t=1,2,\dots,T
\end{equation}
or
\begin{equation}\label{eq:multi2}
\underset{(N,1)}{\boldsymbol{p}_t}\,=\, \lambda_t  \underset{(N,1)}{\boldsymbol{\tilde{p}}}+\underset{(N,1)}{\boldsymbol{\varsigma}_t}\,\,\,\,\,\,\forall\,t=1,2,\dots,T.
\end{equation}
Here $\boldsymbol{p}_t$ is the actual price vector of  the $N$ commodities at period/country $t$, $\tilde{\boldsymbol{p}}$ is the vector of the unknown (time invariant) reference prices, $\lambda_t$  is a scalar factor acting as price index at period/country $t$ and $\boldsymbol{\varsigma}_t$ is the discrepancy vector, that is a vector of error terms. 
As per Eq.~\eqref{eq:multi2}, in each period/country $t$, the $N$ prices $\boldsymbol{p}_t$  can be represented by a point in a $N$-dimensional space. Accordingly, the $N$ prices in $T$ periods/countries, namely $\boldsymbol{P}=[\boldsymbol{p}_{1},\dots,\boldsymbol{p}_{t},\dots,\boldsymbol{p}_{T}]$, can be represented by $T$ points in a $N$-dimensional space. If all prices move proportionally, these points would lie on a hyperplane, $\tilde{\boldsymbol{p}}$, and, in particular, on a straight line crossing the origin for
$T=2$. 
In general, this is only approximately true and a price ``line'' crossing the origin is chosen with the property of fitting the observed price points, by minimizing the deviations of the data from the ``line''. 
In compact notation, Eq.~\eqref{eq:multi1} can be more conveniently reformulated as follows
\begin{equation}\label{eq:multi3}
\underset{(N,T)}{\boldsymbol{P}_{}}\approx \,\underset{(N,T)}{\boldsymbol{\Pi}_{}}=\underset{(N,1)}{\tilde{\boldsymbol{p}}}\,\underset{(1,T)}{\boldsymbol{\lambda}_{}^{'}}
\end{equation}
where $\boldsymbol{\lambda}$ is the vector of the $T$ price indices and  $\tilde{\boldsymbol{p}}$ is the vector of the $N$ (unknown) reference prices. According to Eq.~\eqref{eq:multi3}, the problem of determining a set of $T$ price indices can be read as the problem of approximating the price matrix,  $\boldsymbol{P}$, with a matrix of unit rank, $\boldsymbol{\Pi}$, defined as the outer product of a vector of price indices by a virtual price vector. Moving from prices to values, the matrix $\boldsymbol{V}$ of the values of  $N$ commodities in $T$ periods/countries has the following representation
\begin{equation}\label{eq:multi33}
\underset{(N,T)}{\boldsymbol{V}_{}}=\boldsymbol{P}\ast\boldsymbol{Q}\approx\boldsymbol{\Pi}\ast\boldsymbol{Q} \approx (\underset{(N,1)}{\tilde{\boldsymbol{p}}}\underset{(1,T)}{\boldsymbol{\lambda}'_{}})\ast\underset{(N,T)}{\boldsymbol{Q}_{}} 
\end{equation}
where $\boldsymbol{Q}$ is the matrix  of the quantities of $N$ commodities in $T$ periods/countries and $\ast$ stands for the Hadamard (element-wise) product.
According to Eq.~\eqref{eq:multi33}, the values of $N$ commodities at time $t$ or for the $t$-th country can be represented as
\begin{equation}\label{eq:kjk}
{\boldsymbol{v}_t}=\lambda_t {\tilde{\boldsymbol{p}}}\ast{\boldsymbol{q}_t}+{\bm{\varepsilon}_t}\,\,\,\,\,\forall\, t=1,2,\dots,T
\end{equation}
where  $\boldsymbol{v}_{t}$ and $\boldsymbol{q}_{t}$ are the $t$-th columns of $\boldsymbol{V}$ and $\boldsymbol{Q}$ respectively, and ${\boldsymbol{\varepsilon}_t}$ is added to embody the error term inherent in the model specification of Eq.~\eqref{eq:multi33}. 
The above formula, taking into account the identity
\[\underset{(N,1)}{\tilde{\boldsymbol{p}}}*\underset{(N,1)}{\boldsymbol{q}_t}=\underset{(N,N)}{\boldsymbol{D}_{\tilde{\boldsymbol{p}}}} \underset{(N,1)}{\boldsymbol{q}_t}, \]
can be re-written as
\begin{equation}\label{eq:mod1}
\delta_t \underset{(N,1)}{\boldsymbol{v}_t} =\underset{(N,N)}{\boldsymbol{D}_{\tilde{\boldsymbol{p}}}}\,\underset{(N,1)}{\boldsymbol{q}_{t_{}}}+\underset{(N,1)}{\boldsymbol{\varepsilon}_t}\,\,\,\,\,\forall\, t=1,2,\dots,T                                            \end{equation}
where $\delta_t=(\lambda_t)^{-1}$ takes the role of the deflator, and  $\boldsymbol{D}_{\tilde{\boldsymbol{p}}}$ denotes a diagonal matrix with diagonal entries equal to the elements of  $\tilde{\boldsymbol{p}}$.\footnote{The matrix $\boldsymbol{D}_{\tilde{\boldsymbol{p}}}$ is defined as follows:
\[\boldsymbol{D}_{\tilde{\boldsymbol{p}}}=\begin{bmatrix} \tilde{p}_{1} & 0 & \dots &0 \\ 0 & \tilde{p}_{2} & \dots &0
\\ \vdots & \vdots & \ddots &\vdots\\ 0 & 0 & \dots &\tilde{p}_{N} \end{bmatrix}.
\]\label{ftn:matrix}} 
Eq.~\eqref{eq:mod1} expresses the value, $v_{it}$, of each commodity $i$ at time $t$ (discounted by a factor  $\delta_t$) as the product between the (time invariant) reference  price, $\tilde{p}_i$, and the corresponding quantity, $q_{it}$, plus an error term, $\varepsilon_{it}$. By assuming $\delta_t\neq 0$, its inverse $\lambda_t=\delta_t^{-1}$ tallies with the price index in Eq.~\eqref{eq:kjk}.
Over $T$ periods/countries, the model can be written as
\begin{equation}\label{eq:1}
\underset{(N, T)}{\boldsymbol{V}_{}}\, \, \underset{(T, T)}{\boldsymbol{D}_{\boldsymbol{\delta}}}=  \underset{(N, N)}{\boldsymbol{D}_{\tilde{\boldsymbol{p}}}}\, \, \underset{(N, T)}{\boldsymbol{Q}_{}} + \underset{(N, T)}{\boldsymbol{E}_{}}
\end{equation}
where $\boldsymbol{D}_{\boldsymbol{\delta}}$ is a $T\times T$ diagonal matrix with diagonal entries equal to the elements of $\boldsymbol{\delta}$.

 With no lack of generality, we assume that the first period is the base period (that is  $\delta_1$=$\lambda_1=1$), and write the first equation separately from the others $T-1$. The system takes the form
\begin{equation}\label{eq:1b}
\left[\underset{(N,1)}{\boldsymbol{v}_{1},}\,\, \underset{(N,T-1)}{\boldsymbol{V}_1}\right] \begin{bmatrix} 1 &  \underset{(1,T-1)}{\boldsymbol{0}'}\\ \underset{(T-1,1)}{\boldsymbol{0}} & \underset{(T-1,T-1)}{\boldsymbol{\tilde{D}_\delta}}
\end{bmatrix} 
= \underset{}{\boldsymbol{D}_{\tilde{\boldsymbol{p}}}} \left[\underset{(N,1)}{\boldsymbol{q}_{1},}\,\, \underset{(N,T-1)}{\boldsymbol{Q}_1}\right]+\left[\underset{(N,1)}{\boldsymbol{\varepsilon}_{1},}\,\, \underset{(N,T-1)}{\boldsymbol{E}_1}\right]
\end{equation}
or equivalently, the form
\begin{equation}
\label{eq:17}
\begin{cases}
& \boldsymbol{v}_1=\boldsymbol{D}_{\tilde{\boldsymbol{p}}}\boldsymbol{q}_1+\boldsymbol{\varepsilon}_1
\\& \boldsymbol{V}_1 \boldsymbol{\tilde{D}_\delta}=\boldsymbol{D_{\tilde{\boldsymbol{p}}}}\boldsymbol{Q}_1+\boldsymbol{E}_1
\end{cases}.
\end{equation}
After some computations\footnote{Use has been made of the relationships
\[vec(\boldsymbol{ABC})=(\boldsymbol{C}'\otimes \boldsymbol{A})\,vec(\boldsymbol{A}) \,\,\,\,\,\,\mbox{and}\,\,\,     \,\,\,                  
\underset{(N^2,1)}{vec(\boldsymbol{D}_{\boldsymbol{a}})}=\boldsymbol{R}'_N\,\underset{(N,1)}{\boldsymbol{a}}\]                                              
where $\boldsymbol{D}_{\boldsymbol{a}}$ is a diagonal matrix whose diagonal entries are the elements of the vector $\boldsymbol{a}$ and $\boldsymbol{R}_N$ is the transition matrix from the Kronecker to the Hadamard product \citep{faliva1996hadamard}.
},  the system can be more conveniently rewritten in the form
\begin{equation}
\label{eq:sist1}
\begin{cases}
&  \boldsymbol{v}_1=(\boldsymbol{q}_1' \otimes \boldsymbol{I}_N) \boldsymbol{R}_N' \tilde{\boldsymbol{p}}+\boldsymbol{\varepsilon}_1
\\
& \underset{(N(T-1),1)}{\boldsymbol{0}}=(\boldsymbol{I}_{T-1}\otimes (-\boldsymbol{V}_1))\boldsymbol{R}_{T-1}'\boldsymbol{\delta}+(\boldsymbol{Q}_1'\otimes \boldsymbol{I}_N)\boldsymbol{R}_N'\tilde{\boldsymbol{p}}+\boldsymbol{\eta}\
\end{cases}.
\end{equation}
Here, $\boldsymbol{\eta}= vec(\boldsymbol{E}_1)$, $\boldsymbol{\delta}$ is a vector whose elements are the diagonal entries of $\tilde{\boldsymbol{D}}_{\boldsymbol{\delta}}$ as specified in Eq.~\eqref{eq:1b}, and $\boldsymbol{R}_j$ denotes the transition matrix from the Kronecker to the Hadamard product.\footnote{The matrix $\boldsymbol{R}_{j}'$ is defined as follows \[\underset{(j^2\times j)}{\boldsymbol{R}_j'} =
\begin{bmatrix} \underset{(1,j)}{\boldsymbol{e}_{1}}\otimes\underset{(1,j)}{\boldsymbol{e}_{1}} &
\underset{(j,1)}{\boldsymbol{e}_{2}}\otimes\underset{(j,1)}{\boldsymbol{e}_{2}}
& \dots  & \underset{j,1)}{\boldsymbol{e}_{j}}\otimes\underset{(1,j)}{\boldsymbol{e}_{j}} \end{bmatrix},
\]
where \noindent  $\boldsymbol{e}_i$  represents the $N$ dimensional $i$-th elementary vector. } 
An estimate of the vector $\boldsymbol{\delta}$ can be obtained by applying generalized least squares (GLS), by taking
\begin{equation}\label{eq:11}
\mathds{E}(\boldsymbol{\mu}\boldsymbol{\mu}')=\underset{(NT, NT)}{\boldsymbol{\Omega}}=\mbox{diag}\left[\bm{\vartheta}\right]=
\begin{bmatrix} \underset{(N,N)}{\boldsymbol{\Omega}_{11}} &  \underset{(N,N(T-1))}{\boldsymbol{0}}\\ 
\underset{(N,N)}{\boldsymbol{0}} & \underset{(N(T-1),N(T-1))}{\boldsymbol{\Omega}^{*}}
\end{bmatrix}
\end{equation}
where $\boldsymbol{\mu}'=[\varepsilon_{1}, \boldsymbol{\eta}']$ and  $\bm{\Omega}^{*}$ is a block diagonal matrix
\begin{equation}
\boldsymbol{\Omega}^{*}=\begin{bmatrix} \underset{(N,N)}{\boldsymbol{\Omega}_{22}} &  \underset{(N,N)}{\boldsymbol{0}}& ... \underset{(N,N)}{\boldsymbol{0}}\\
\underset{(N,N)}{\boldsymbol{0}} &\underset{(N,N)}{\boldsymbol{\Omega}_{33}}& ... \underset{(N,N)}{\boldsymbol{0}}\\
\vdots\\
\underset{(N,N)}{\boldsymbol{0}} &\underset{(N,N)}{\boldsymbol{0}}& ... \underset{(N,N)}{\boldsymbol{\Omega}_{TT}}\\
\end{bmatrix}.
\end{equation}
In this connection, we have the following result. 
\begin{theorem}\label{th:2gls}
\textit{The GLS estimate of the deflator vector $\boldsymbol{\delta}$ is given by 
\begin{equation}\label{eq:deltagg12}
\begin{split}\boldsymbol{\widehat{\delta}}_{GLS}=& 
\left\{\boldsymbol{I}_{T-1}\ast\boldsymbol{\widetilde{V}}_1'\boldsymbol{V}_1-(\boldsymbol{Q}'_1\ast\boldsymbol{\widetilde{V}}_1')\left[\sum_{j=1}^{T}\boldsymbol{q}_j\boldsymbol{q}_j'\ast{\boldsymbol{\Omega}^{-1}_{j,j}}\right]^{-1}(\boldsymbol{Q}_1\ast\boldsymbol{\widetilde{V}}_1)\right\}^{-1}
\\
&(\boldsymbol{Q}'_1\ast\boldsymbol{\widetilde{V}}_1')\left[\sum_{j=1}^{T}\boldsymbol{q}_j\boldsymbol{q}_j'\ast{\boldsymbol{\Omega}^{-1}_{j,j}}\right]^{-1}(\boldsymbol{q}_1\ast \boldsymbol{\widetilde{v}}_1)
\end{split}
\end{equation}
where $\boldsymbol{\widetilde{V}_1}$ is the matrix whose $j$-th columns is ${\boldsymbol{\Omega}^{-1}_{j,j}}\boldsymbol{v}_j$ and $\boldsymbol{\widetilde{v}}_1=\boldsymbol{\Omega^{-1}}_{1,1}{\boldsymbol{v}}_1$.
The vector $\bm{\widehat{\lambda}}$, whose entries are the reciprocals of non-null elements of the GLS estimate of $\boldsymbol{\delta}$ and zero otherwise,
is the MPL index  of $N$ commodities over $T$ periods or between $T$ countries}.\footnote{
The vector $\bm{\lambda}$ is defined as follows
\begin{equation}
\label{eq:666}
\boldsymbol{\lambda}=[\lambda_t]\,\,\,\mbox{with}\,\,\,\lambda_t=\begin{cases} \delta_{t}^{-1}\,\,\,\,\,\,\mbox{if}\,\,\,\delta_{t}\neq 0
\\
0 \,\,\,\,\,\,\,\,\,\,\,\,\mbox{otherwise}
\end{cases}.
\end{equation}
}
\end{theorem}
\begin{proof}
See Appendix~\ref{app:proofs2gls}.
\end{proof}
\noindent  The following corollaries provide estimates of the deflator vector for special cases of interest.  
\begin{corollary}\label{cor:glsstat1}
\textit{Let us assume that the error terms of the system in Eq.~\eqref{eq:sist1} are stationary. Accordingly, the diagonal blocks, ${\boldsymbol{\Omega}_{j,j}}$ of the matrix $\boldsymbol{\Omega}$ are the same, i.e. ${\boldsymbol{\Omega}_{j,j}}=\widetilde{\boldsymbol{\Omega}}$ for all $j$. Then, the GLS estimate of the deflator vector, $\boldsymbol{\delta}$, is }
\begin{equation}\label{eq:cc1}
\begin{split}\boldsymbol{\widehat{\delta}}_{GLS}=&\left\{
(\boldsymbol{I}_{T-1}*\boldsymbol{\widehat{V}}_1'\boldsymbol{V}_1)-(\boldsymbol{Q}_1'*\boldsymbol{\widehat{V}}_1')\left[
(\boldsymbol{q}_1\boldsymbol{q}_1'+\boldsymbol{Q}_1\boldsymbol{Q}_1')*\widetilde{\boldsymbol{\Omega}}^{-1}
\right]^{-1}(\boldsymbol{Q}_1*\boldsymbol{\widehat{V}}_1)
\right\}^{-1}\cdot
\\
&(\boldsymbol{Q}_1'*\boldsymbol{\widehat{V}}_1')\left[
(\boldsymbol{q}_1\boldsymbol{q}_1'+\boldsymbol{Q}_1\boldsymbol{Q}_1')*\widetilde{\boldsymbol{\Omega}}^{-1}
\right]^{-1}(\boldsymbol{q}_1\ast \boldsymbol{\widehat{v}}_1)
\end{split}
\end{equation}
\end{corollary}
\noindent\textit{where $\boldsymbol{\widehat{V}_1}$ is the matrix whose $j-th$ column is $\boldsymbol{\widetilde{\Omega}}^{-1}\boldsymbol{v}_j$ and $\boldsymbol{\widehat{v}}_1=\boldsymbol{\widetilde{\Omega}}^{-1}\boldsymbol{v}_1$ }.
\begin{proof}
The proof is given in Appendix~\ref{app:proofs2g}.
\end{proof}

\noindent The case $T=2$ is worth considering because it sheds light on the index structure. When assuming stationary and uncorrelated error terms, the matrix $\widetilde{\boldsymbol{\Omega}}$ reduces to a diagonal matrix, i.e. $\widetilde{\boldsymbol{\Omega}}=diag(\boldsymbol{\theta})=[\theta_i]$, and the price index turns out to be simply the ratio of weighted price averages, as stated in the following
corollary. 
\begin{corollary}\label{cor:1g}
\textit{Let the matrix $\widetilde{\boldsymbol{\Omega}}$ be diagonal and $T=2$. Then the MPL index $\widehat{\lambda}_{GLS}$ for a set of $N$ commodities is
\begin{equation}
\label{eq:ffrt0g}
\widehat{{\lambda}}_{GLS}=\frac{\sum_{i=1}^Np_{i2}\overline{\pi}_i}{\sum_{i=1}^Np_{i1}\overline{\pi}_i}=\frac{\boldsymbol{p}_{2}'\boldsymbol{\overline{\pi}} }
{\boldsymbol{ p}_{1}'\boldsymbol{\overline{\pi}}}.
\end{equation}
Here $\boldsymbol{p}_{t}$ is the price vector of the $N$ commodities at time $t$ and $\boldsymbol{\overline{\pi}}$ is the weight vector whose $i$-th entry is 
\begin{equation}\label{eq:gf}
\overline{\pi}_i=\frac{1}{\vartheta_i }p_{i2}\frac{q_{i1}^2q_{i2}^2}{q_{i1}^2+q_{i2}^2}.
\end{equation}
 } 
\end{corollary}
\begin{proof}
See Appendix~\ref{app:proofs3g}.
\end{proof}
\noindent 
 In case the diagonal entries $\vartheta_i$ of $\widetilde{\boldsymbol{\Omega}}$ are all equal, say  equal to 1,  then Eq.~\eqref{eq:ffrt0g} would tally with the OLS version of the MPL index. In this case the weights of the index would be function of the harmonic mean of the squared quantities
\begin{equation}\label{eq:ght}
\overline{\pi}_i=p_{i2}\frac{q_{i1}^2q_{i2}^2}{q_{i1}^2+q_{i2}^2}.
\end{equation}
\noindent
Two considerations are worth making about Eq.~\eqref{eq:ffrt0g}. The first is that most well-known price indices can be viewed as particular cases of the MPL index. In fact, they can be obtained from Eq.~\eqref{eq:ffrt0g} for particular choices of the scalars $\vartheta_i$, as proved in the following corollary.
\begin{corollary}\label{cor:81g}
\textit{By taking
\begin{equation}\label{eq:gf1}
\vartheta_i=p_{i2}\frac{q_{i1}q_{i2}^2}{q_{i1}^2+q_{i2}^2},
\end{equation}
then $\overline{\pi}_{i}=q_{i1}$ and the MPL index tallies with the Laspeyeres index, $\widehat{{\lambda}}_{L}$, i.e.
\begin{equation}\label{eq:gt1}
\widehat{{\lambda}}_{GLS}=\widehat{{\lambda}}_{L}=\frac{\sum_{i=1}^N p_{i2}q_{i1}}{\sum_{i=1}^N p_{i1}q_{i1}}=\frac{\boldsymbol{p}_2'\boldsymbol{q}_1}{\boldsymbol{p}_1'\boldsymbol{q}_1}.
\end{equation}
 By taking
\begin{equation}\label{eq:gf2}
\vartheta_i=p_{i2}\frac{q_{i1}^2q_{i2}}{q_{i1}^2+q_{i2}^2},
\end{equation}
then $\overline{\pi}_{i}=q_{i2}$ and the MPL index turns out to  tally with the Paasche index, $\widehat{{\lambda}}_{P}$, i.e. 
\begin{equation}\label{eq:gt2}
\widehat{{\lambda}}_{GLS}=\widehat{{\lambda}}_{P}=\frac{\sum_{i=1}^N p_{i2}q_{i2}}{\sum_{i=1}^N p_{i1}q_{i2}}=\frac{\boldsymbol{p}_2'\boldsymbol{q}_2}{\boldsymbol{p}_1'\boldsymbol{q}_2}.
\end{equation}
By taking
\begin{equation}\label{eq:gf3}
\vartheta_i=p_{i2}\frac{q_{i1}^2q_{i2}^2}{(q_{i1}^2+q_{i2}^2)(q_{i1}+q_{i2})},
\end{equation}
 then  $\overline{\pi}_{i}=(q_{i1}+ q_{i2})$ and the MPL index turns out to tally with the Marshal-Edgeworth index, $\widehat{{\lambda}}_{ME}$ i.e. 
\begin{equation}\label{eq:gt3}
\widehat{{\lambda}}_{GLS}=\widehat{{\lambda}}_{ME}=\frac{\sum_{i=1}^N p_{i2}(q_{i1}+q_{i2})}{\sum_{i=1}^N p_{i1}(q_{i1}+q_{i2})}=\frac{\boldsymbol{p}_2'(\boldsymbol{q}_{1}+\boldsymbol{q}_{2})}{\boldsymbol{p}_1'(\boldsymbol{q}_{1}+\boldsymbol{q}_{2})}.
\end{equation}
By taking
\begin{equation}\label{eq:gf4}
\vartheta_i=p_{i2}\frac{\left(q_{i1}q_{i2}\right)^{3/2}}{q_{i1}^2+q_{i2}^2},
\end{equation}
 then $\overline{\pi}_{i}=(q_{i1}q_{i2})^{1/2} $ and the MPL index turns out to tally with the Walsh index, $\widehat{{\lambda}}_{W}$, i.e. 
\begin{equation}\label{eq:gt4}
\widehat{{\lambda}}_{GLS}=\widehat{{\lambda}}_{W}=\frac{\sum_{i=1}^N p_{i2}(q_{i1}q_{i2})^{1/2}}{\sum_{i=1}^N p_{i1}(q_{i1}q_{i2})^{1/2}}=\frac{\boldsymbol{p}_2'\boldsymbol{\widetilde{q}}}{\boldsymbol{p}_1'\boldsymbol{\widetilde{q}}},
\end{equation}
where $\boldsymbol{\widetilde{q}}$ is a vector whose entries are the square roots of those of the vector $(\boldsymbol{q}_{1} \ast \boldsymbol{q}_{2})$.\\
By taking
\begin{equation}
\vartheta_i=\frac{1}{p_{i2}}
\end{equation}
and considering the square roots of the quantities in the index computation, then $\overline{\pi}_{i}=\frac{q_{i1}q_{i2}}{q_{i1}+q_{i2}} $ and the MPL index turns out to tally with the Geary-Khamis index \citep{drechsler1973weighting}, $\widehat{{\lambda}}_{GK}$, i.e. 
\begin{equation}\label{eq:gt4b}
\widehat{{\lambda}}_{GLS}=\widehat{{\lambda}}_{GK}=\frac{\sum_{i=1}^N p_{i2}\frac{q_{i1}q_{i2}}{q_{i1}+q_{i2}}}{\sum_{i=1}^N p_{i1}\frac{q_{i1}q_{i2}}{q_{i1}+q_{i2}}}=\frac{\boldsymbol{p}_2'\boldsymbol{\widehat{q}}}{\boldsymbol{p}_1'\boldsymbol{\widehat{q}}},
\end{equation}
where $\boldsymbol{\widehat{q}}$ is a vector whose $i$-th entry is $\frac{q_{i1}q_{i2}}{q_{i1}+q_{i2}}$.
}
\end{corollary} 
\begin{proof}
The proof is simple and follows straight forward. Therefore,  it is omitted.
\end{proof}
\noindent The other consideration is that the index $\widehat{{\lambda}}_{GLS}$ can be obtained as solution of an optimization problem specified as in the following corollary. 
\begin{corollary}\label{cor:51g}
\textit{ With reference to the following model
\begin{equation}\label{eq:mlu}
\boldsymbol{p}_{2}=\lambda \boldsymbol{p}_{1} + \boldsymbol{\eta}
\end{equation}
where $\boldsymbol{\eta}$ is a vector of random terms, the index $\widehat{{\lambda}}_{GLS}$ is solution of the optimization problem
\begin{equation}
\min_{\lambda}\left\| \boldsymbol{p}_{2}-\lambda \boldsymbol{p}_{1}\right\|=\min_{\lambda}\left| \boldsymbol{e}\right\|
\end{equation} 
where $\left\| \boldsymbol{e}\right\|=(\boldsymbol{e}' \boldsymbol{A}\boldsymbol{e})^{1/2}$ is a (semi)norm of $\boldsymbol{e}$
and $\boldsymbol{A}$ is given by
\begin{equation}
\boldsymbol{A}=\overline{\boldsymbol{\pi}}\overline{\boldsymbol{\pi}}'
\end{equation}
with $\boldsymbol{\overline{\pi}}$ is  a vector whose the $i$-th entry, $\overline{\pi}_i$, is specified as in Eq.~\eqref{eq:gf}.
}
\end{corollary} 
\begin{proof}
See Appendix~\ref{app:proofs1}.
\end{proof}

\noindent Eq.~\eqref{eq:mlu}, together with Figure~\ref{fig:1}, is useful to show that the index basket is the union of the intersection of the commodities that are present in at least two periods/countries and, accordingly, it does not include commodities which are not present in at least two periods/countries. In fact, simple computations prove that weights pertaining commodities not fulfilling this condition turn out to be null. Consequently, the value of the index does not change if computed leaving out them, as shown by the following example.     

\begin{example}\label{ex:2}
Let us consider the case of three commodities in two periods and assume that the first commodity is missing in both periods (i.e.  $q_{11}=q_{12}=0$). 
Then, let us assume that $\vartheta_i=1$ for each $i$ for simplicity. Simple computations prove that the weight, as defined in Eq.~\eqref{eq:ght}, associated to the first coefficient turns out to be null 
\[\begin{split}
&\overline{\pi}_1=p_{12}\frac{q_{11}^2q_{12}^2}{q_{11}^2+q_{12}^2}=p_{i2}\left(\frac{1}{\frac{1}{q_{i1}^2}}+\frac{1}{\frac{1}{q_{i2}^2}} \right)=p_{12}\frac{1}{\frac{1}{0}+\frac{1}{0} }=0; \\
\end{split}\]
The same happens if the first commodity is missing in just one period, for instance the second one ($q_{11}=0$)
As in the previous case, the coefficient $\overline{\pi}_1$ vanishes
\[\begin{split}
&\overline{\pi}_1=p_{12}\frac{q_{11}^2q_{12}^2}{q_{11}^2+q_{12}^2}=0;\\
\end{split}\]
Accordingly, in both cases
the first commodity can be left out from the computation of the index. In fact, the value of the latter does not change if computed by considering only the second and the third commodities.
\end{example}

\noindent As a by-product of Theorem~\ref{th:2gls}, we state the following result.

\begin{corollary}\label{cor:3}
\textit{The variance-covariance matrix of the deflator vector, $\widehat{\boldsymbol{\delta}}$, given in Theorem~\ref{th:2gls} is 
\begin{equation}
Var(\widehat{\boldsymbol{\delta}})=\sigma^2\,\boldsymbol{\kappa}^{-1}
\end{equation}
\noindent where 
\[\boldsymbol{\kappa}=\boldsymbol{I}_{T-1}\ast\boldsymbol{\widetilde{V}}_1'\boldsymbol{V}_1-(\boldsymbol{Q}'_1\ast\boldsymbol{\widetilde{V}}_1')\left[\sum_{j=1}^{T}\boldsymbol{q}_j\boldsymbol{q}_j'\ast{\boldsymbol{\Omega}^{-1}_{j,j}}\right]^{-1}(\boldsymbol{Q}_1\ast\boldsymbol{\widetilde{V}}_1).\]
The $t$-th diagonal entry of the above matrix provides the variance of the deflator in the $t$-th period/country, given by
\begin{equation}
var({\widehat{\delta}_{t}})=\sigma^2\left\{
\boldsymbol{\widetilde{v}}_t'\boldsymbol{v}_t - (\boldsymbol{q}'_t\ast\boldsymbol{\widetilde{v}}_t')\left[\sum_{j=1}^{T}\boldsymbol{q}_j\boldsymbol{q}_j'\ast{\boldsymbol{\Omega}^{-1}_{j,j}}\right]^{-1}(\boldsymbol{q}_t\ast\boldsymbol{\widetilde{v}}_t)
\right\}^{-1}= \sigma^2 \kappa^{-1}_t
\end{equation}
where $\widetilde{\boldsymbol{v}}_t$, $\boldsymbol{v}_t$ and  $\boldsymbol{q}_t$ denote the $t$-th column of the matrix $\boldsymbol{\widetilde{V}}_1$, $\boldsymbol{V}_1$ and $\boldsymbol{Q}_1$, respectively.
}
\end{corollary}
\begin{proof}
See Appendix~\ref{app:proofs6}.
\end{proof}

\noindent As for the deflator vector $\widehat{\boldsymbol{\delta}}$, its moments and confidence intervals can be easily obtained within the theory of linear regression models, from the result given in Corollary~\ref{cor:3}. As the price index vector $\widehat{\boldsymbol{\lambda}}$ turns out to be the reciprocal of the said deflator (see Eq.~\ref{eq:deltagg12}), its  statistical behavior can be derived from the former, following the arguments put forward, for example, in \citet{geary1930frequency, curtiss1941distribution} and \citet{marsaglia1965ratios}, merely to quote a few, on ratios (in particular reciprocals) of random variables. 
The following corollary provides an approximation of the variance of $\widehat{\lambda}_t$, obtained by using the first Taylor expansion of the variance of a ratio of two random variables. 
\begin{corollary}
\textit{The variance of the MPL index ${\widehat{\lambda}_{t}}$ is
\begin{equation}
var({\widehat{\lambda}_{t}}) \approx \frac{var({\widehat{\delta}_{t}})}{\mathds{E}({\widehat{\delta}_{t}}^4)}
=\frac{\widehat{\sigma}^2}{\kappa_t\,\mathds{E}({\widehat{\delta}_{t}}^4)} \,\,\,\,\,\,\forall\,t=1,\dots,T.
\end{equation}
\noindent In the above equation $\widehat{\sigma}^2=\frac{\widehat{\boldsymbol\mu}' \widehat{\boldsymbol{\Omega}}^{-1}\widehat{\boldsymbol \mu}}{NT-(N+T-1)}$ where $\widehat{\boldsymbol\mu}$ are the GLS residuals of Eq.~\eqref{eq:sist1} and $\widehat{\boldsymbol{\Omega}}$ is an estimate of $\boldsymbol{\Omega}$.}
\end{corollary}
\begin{proof}
See \citet[p. 351]{stuard1994kendall} and \citet[p. 69]{elandt1980survival}. 
\end{proof}


\subsection{The MPL update}
\noindent  The following two theorems provide updating formulas for the  price index $\boldsymbol{\widehat{\lambda}}$. The former proves suitable when the index is used as a multilateral price index,while the latter is appropriate when it is employed as a multi-period index. In the former case,  values and quantities of the commodities included in the reference basket are assumed available for an additional $T+1$ country. In the latter case, it is supposed that values and quantities of the commodities included in the reference basket become available at time $T+1$. It is worth noting that the approach used to update the index guarantees the temporal fixity issue requiring that its historical values must not be affected by the inclusion of values and quantities pertaining a new period. Spatial fixity, demanding that results for a core set of countries must be unaffected by the inclusion of new countries, is not preserved by the updating method here proposed. This property can be easily fulfilled by updating the index with the same approach used for the multi-period index.
\begin{theorem}\label{th:4}
\textit{
Should the values and quantities of $N$ commodities of a reference basket become available for a new additional country, say the $T+1$-th, then, the updated multilateral version of the MPL index, $\boldsymbol{\widehat{\lambda}}$, turns out to be the vector of the reciprocals, as defined in Eq.~\eqref{eq:deltagg12}, of the following deflator vector
\begin{equation}\label{eq:th4}
\begin{split}
\underset{(T,1)}{\widehat{\boldsymbol{\delta}}}
=&\left\{
\begin{bmatrix} 
\boldsymbol{I}_{T-1}\ast\boldsymbol{\widetilde{V}}'_1\boldsymbol{V}_1 & \boldsymbol{0}  \\ 
\boldsymbol{0} & \boldsymbol{\widetilde{v}}'_{T+1}\boldsymbol{v}_{T+1}  \end{bmatrix}-
\begin{bmatrix} 
\boldsymbol{Q}'_1\ast\boldsymbol{\widetilde{V}}'_1  \\ 
\boldsymbol{\widetilde{v}}'_{T+1}\ast\boldsymbol{q}'_{T+1}  \end{bmatrix}
\left[\sum_{j=1}^{T+1}\boldsymbol{q}_j\boldsymbol{q}'_j \ast\boldsymbol{\Omega}_{jj}^{-1} \right]^{-1} \right.
\\&
\left.\begin{bmatrix} 
\boldsymbol{Q}_1\ast\boldsymbol{\widetilde{V}}_1  & 
\boldsymbol{\widetilde{v}}_{T+1}\ast\boldsymbol{q}_{T+1}  \end{bmatrix}
\right\}^{-1}
\begin{bmatrix} 
\boldsymbol{Q}'_1\ast\boldsymbol{\widetilde{V}}'_1  \\ 
\boldsymbol{\widetilde{v}}'_{T+1}\ast\boldsymbol{q}'_{T+1}  \end{bmatrix}
\left[
\sum_{j=1}^{T+1}\boldsymbol{q}_j\boldsymbol{q}'_j \ast\boldsymbol{\Omega}_{jj}^{-1} \right]^{-1}
\left(\boldsymbol{q}_1\ast\boldsymbol{\widetilde{v}}_1\right).
\end{split}
\end{equation}
Here the symbols are defined as in Theorem~\ref{th:2gls}. The terms $\boldsymbol{v}_{T+1}$, $\boldsymbol{q}_{T+1}$ denote the vector of values and quantities of  $N$ commodities of the new $T+1$-th country, respectively and  $\boldsymbol{\widetilde{v}}_{T+1}=\boldsymbol{\Omega}_{T+1,T+1}^{-1}\boldsymbol{v}_{T+1}$ with $\boldsymbol{\Omega}_{T+1,T+1}$ variance-covariance matrix of the disturbances at time $T+1$.
}
\end{theorem}
\begin{proof}
See Appendix~\ref{app:proofs7}.
\end{proof}
\begin{theorem}\label{th:5}
\textit{
Should the values and quantities of $N$ commodities of a reference basket become available for  time $T+1$,  then, the updated value $\widehat{\lambda}_{T+1}$ of the multi-period version of the MPL index at time $T+1$ turns out to be the reciprocal of the deflator value at time $T+1$
}
\begin{equation}\label{eq:update1}
\begin{split}
\underset{(1,1)}{\widehat{\delta}_{T+1}}=&\left\{
\boldsymbol{\widetilde{v}}_{T+1}'\boldsymbol{v}_{T+1}-\left(\boldsymbol{q}_{T+1}'\ast\boldsymbol{\widetilde{v}}_{T+1}'\right)\left[
\sum_{j=1}^{T+1}\boldsymbol{q}_{j}\boldsymbol{q}_{j}'\ast \boldsymbol{\Omega}_{j,j}
\right]^{-1}\left(\boldsymbol{q}_{T+1}\ast\boldsymbol{\widetilde{v}}_{T+1}\right)
\right\}^{-1} \\&\left(\boldsymbol{q}_{T+1}'\ast\boldsymbol{\widetilde{v}}_{T+1}'\right)\left[
\sum_{j=1}^{T+1}\boldsymbol{q}_{j}\boldsymbol{q}_{j}'\ast \boldsymbol{\Omega}_{j,j}
\right]^{-1}\left(\boldsymbol{Q}\ast\boldsymbol{\widetilde{V}}\right)\tilde{{\boldsymbol{\delta}}}
\end{split}
\end{equation}
\noindent \textit{where $\tilde{{\boldsymbol{\delta}}}'=[1, \widehat{{\boldsymbol{\delta}}}]'$ and $\widehat{{\boldsymbol{\delta}}}$ is defined as in Eq.~\eqref{eq:deltagg12}.}
\end{theorem}
\begin{proof}
See Appendix~\ref{app:proofs8}.
\end{proof}
\noindent  Figure~\ref{fig:2} highlights the difference between the updating process of the deflator, and thus of the price index, depending on whether it is used in the multilateral or in the multi-period case.    

\begin{figure}[htbp]
\begin{center}
\includegraphics[scale=0.3]{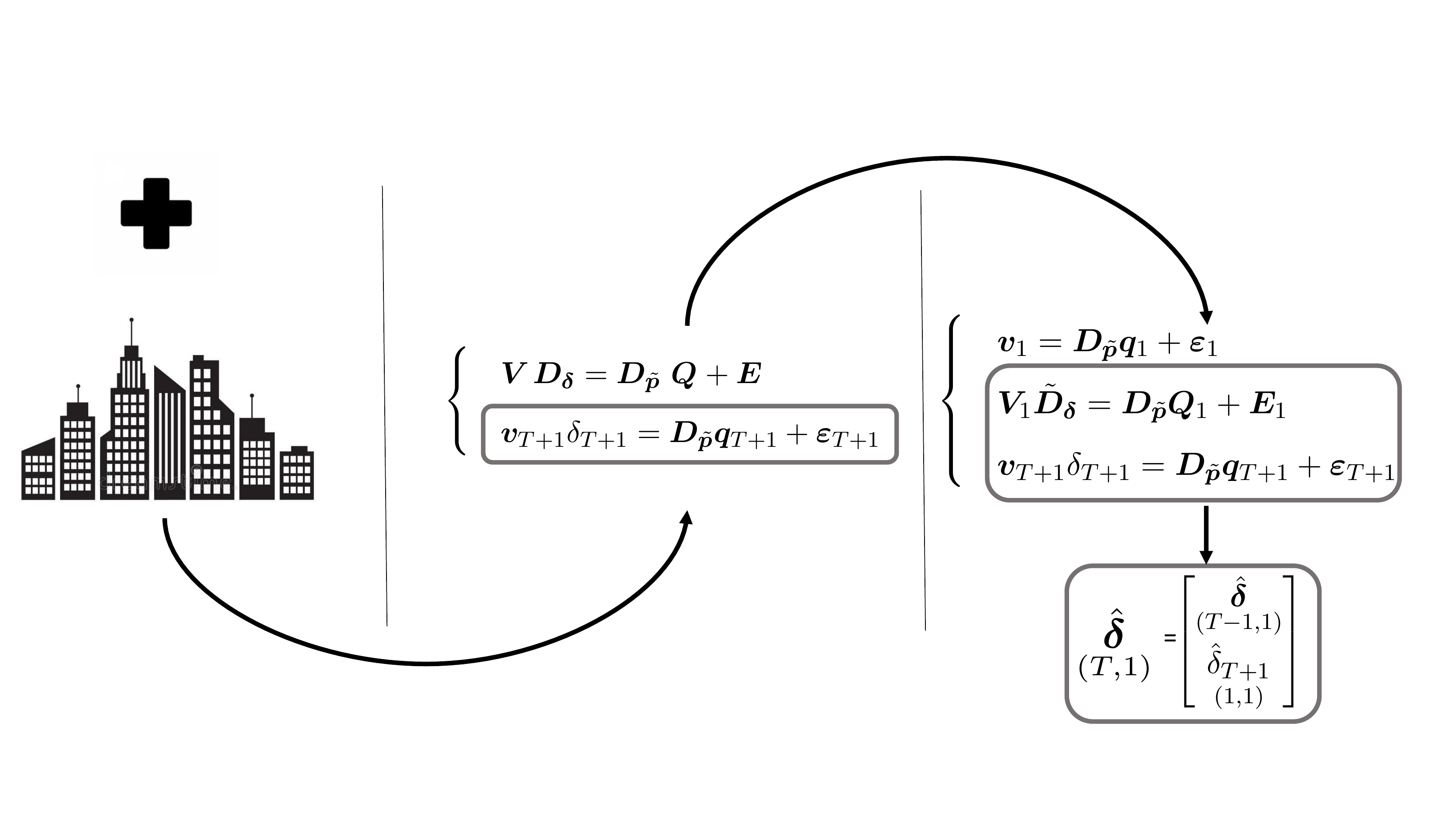}
\includegraphics[scale=0.3]{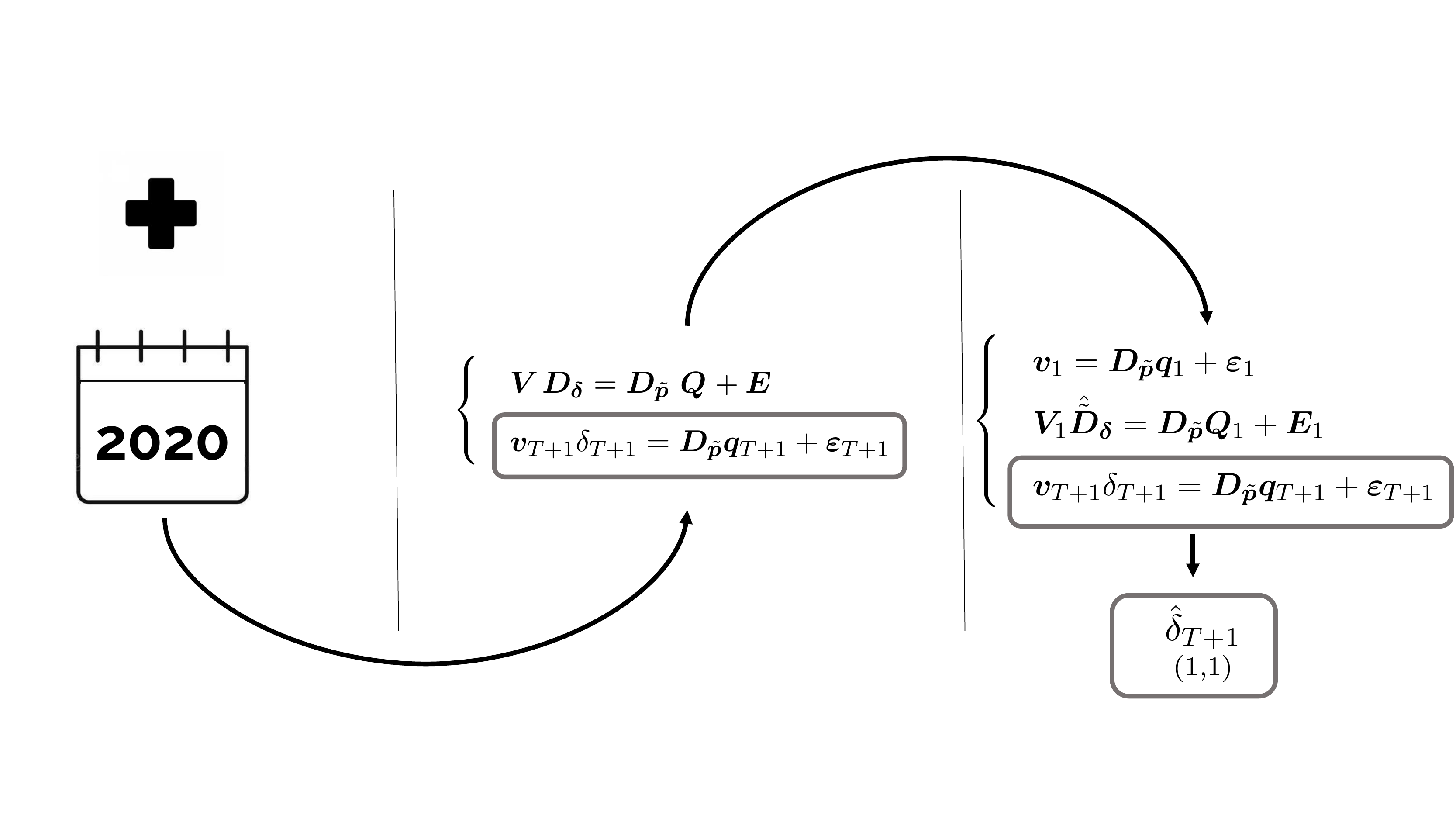} 
\caption{
The top panel shows the ratio of the updating formula for the multilateral version of the MPL index (see Theorem~\ref{th:4}); the bottom panel shows the ratio of the updating formula of the the multi-period version of the MPL index (see Theorem~\ref{th:5}).  
} 
\label{fig:2}
\end{center}
\end{figure}

\section{Properties of the MPL index}\label{subsec:prop}
Let us assume for simplicity $T=2$ and denote with $\widehat{\lambda}(\boldsymbol{p}_{1},\boldsymbol{p}_{2},\boldsymbol{q}_{1},\boldsymbol{q}_{2})$ a generic index number where $\boldsymbol{p}_{t}$ and $\boldsymbol{q}_{t}$ are prices and quantities at time $t$. Without lack of generality, $t=1$ is assumed to be the base period. 
Following \citet{predetti2006numeri}, \citet{martini1992numeri} and \citet{fattore2010axiomatic}, the main properties of an index number can be summarized as follows:
\begin{enumerate}
\item[P.1] \textit{Strong identity}: $\widehat{\lambda}(\boldsymbol{p}_{2},\boldsymbol{p}_{2},\boldsymbol{q}_{1},\boldsymbol{q}_{2})=1$.
\item[P.2] \textit{Commensurability}: $\widehat{\lambda}(\boldsymbol{\gamma}\ast\boldsymbol{p}_{1},\boldsymbol{\gamma}\ast\boldsymbol{p}_{2},\boldsymbol{\gamma}^{-1}\ast\boldsymbol{q}_{1},\boldsymbol{\gamma}^{-1}\ast\boldsymbol{q}_{2})=\widehat{\lambda}(\boldsymbol{p}_{1},\boldsymbol{p}_{2},\boldsymbol{q}_{1},\boldsymbol{q}_{2})$ where $\underset{(N,1)}{\boldsymbol{\gamma}}$ is a vector with non-null entries and $\boldsymbol{\gamma}^{-1}$ is the vector of reciprocals of the entries of $\boldsymbol{\gamma}$. 
\item[P.3] \textit{Proportionality}:  $\widehat{\lambda}(\boldsymbol{p}_{1},\alpha\boldsymbol{p}_{2},\boldsymbol{q}_{1},\boldsymbol{q}_{2})=\alpha\,\widehat{\lambda}(\boldsymbol{p}_{1},\boldsymbol{p}_{2},\boldsymbol{q}_{1},\boldsymbol{q}_{2})$ with $\alpha>0$.
\item[P.4] \textit{Dimensionality}:   $\widehat{\lambda}(\alpha\boldsymbol{p}_{1},\alpha\boldsymbol{p}_{2},\boldsymbol{q}_{1},\boldsymbol{q}_{2})=\widehat{\lambda}(\boldsymbol{p}_{1},\boldsymbol{p}_{2},\boldsymbol{q}_{1},\boldsymbol{q}_{2})$ with $\alpha>0$.
\item[P.5] \textit{Monotonicity}: $\widehat{\lambda}(\boldsymbol{p}_{1},\boldsymbol{k}\ast\boldsymbol{p}_{2},\boldsymbol{q}_{1},\boldsymbol{q}_{2})>\widehat{\lambda}(\boldsymbol{p}_{1},\boldsymbol{p}_{2},\boldsymbol{q}_{1},\boldsymbol{q}_{2})$ and $\widehat{\lambda}(\boldsymbol{k}\ast\boldsymbol{p}_{1},\boldsymbol{p}_{2},\boldsymbol{q}_{1},\boldsymbol{q}_{2})<\widehat{\lambda}(\boldsymbol{p}_{1},\boldsymbol{p}_{2},\boldsymbol{q}_{1},\boldsymbol{q}_{2})$ with $\underset{(N,1)}{\boldsymbol{k}}>\boldsymbol{u}$ where $\boldsymbol{u}$ is the unit vector.
\end{enumerate}
Moreover, also the following properties are worth mentioning:
\begin{itemize}
\item[P.6] \textit{Positivity}: $\widehat{\lambda}(\alpha\boldsymbol{p}_{1},\boldsymbol{p}_{2},\boldsymbol{q}_{1},\boldsymbol{q}_{2})\geq0$, with $\alpha$ positive constant;
\item[P.7] \textit{Inverse proportionality in the base period}: $\widehat{\lambda}(\alpha\boldsymbol{p}_{1},\boldsymbol{p}_{2},\boldsymbol{q}_{1},\boldsymbol{q}_{2})=\frac{1}{\alpha} \widehat{\lambda}(\boldsymbol{p}_{1},\boldsymbol{p}_{2},\boldsymbol{q}_{1},\boldsymbol{q}_{2})$;
\item[P.8] \textit{Commodity reversal property}: invariance of the index with respect to any commodity permutation;
\item[P.9] \textit{Quantity reversal test}:  a change in the quantity order does not affects  $\pi_i$ that remains invariant $\forall\, i=1,\dots,N$. Therefore the index price $\widehat{\lambda}$ does not change. 
\item[P.10] \textit{Base reversibility (symmetric treatment of time)} $\widehat{\lambda}(\boldsymbol{p}_{1},\boldsymbol{p}_{2},\boldsymbol{q}_{1},\boldsymbol{q}_{2})=\widehat{\lambda}(\boldsymbol{p}_{2},\boldsymbol{p}_{1},\boldsymbol{q}_{2},\boldsymbol{q}_{1})^{-1}$;
\item[P.11] \textit{Transitivity} $\widehat{\lambda}(\boldsymbol{p}_{1},\boldsymbol{p}_{2},\boldsymbol{q}_{1},\boldsymbol{q}_{2})\widehat{\lambda}(\boldsymbol{p}_{2},\boldsymbol{p}_{3},\boldsymbol{q}_{2},\boldsymbol{q}_{3})=\widehat{\lambda}(\boldsymbol{p}_{1},\boldsymbol{p}_{3},\boldsymbol{q}_{1},\boldsymbol{q}_{3})$;
\item[P.12] \textit{Monotonicity}: If $\boldsymbol{p}_{2}=\beta \boldsymbol{p}_{1}$, then $\widehat{\lambda}(\boldsymbol{p}_{1},\beta\boldsymbol{p}_{1},\boldsymbol{q}_{1},\boldsymbol{q}_{2})= \beta$.
\end{itemize}

\begin{proposition}
\textit{The MPL index satisfies all properties. }
\end{proposition}
\begin{proof}
See Appendix~\ref{app:propgmpl}.
\end{proof}

\section{The Italian cultural supply: an application of the MPL index }\label{sec:empirical}
\noindent In this section, we provide an application of the MPL  index to Italian cultural supply data, such as revenues and the number of visitors to museums (i.e. monuments, archeological sites, museum circuits, $\dots$).
The availability of  temporal and  geographical data on Italian culture provides a stimulating basis for ascertaining the potential of the MPL price-index methodology set forth in this paper. The flexibility of the MPL index paves the way to moving beyond ISTAT (and similar) analyses, which are confined to price indices on the supply of data on Italian culture like access to museums and entertainment sectors, aggregated at the national level \citep{istatr2020}. 
In addition, to evaluate the performance of the MPL index, we have made a comparison with the CPD/TPD price indices \citep{diewert2005weighted, rao2004country}, using both real and simulated data. Reference has been made to this approach because, under the log-normality assumption of the error term, the maximum likelihood estimator of the said price index tallies with the least square one, likewise the MPL index.
As for the nature of the data, note that Italian cultural heritage is at the top of various world-class lists and plays a key role in the Italian economy \citep[see, e.g.,][]{di2019io}.\footnote{The Italian heritage supply chain accounts for 4,889 museums and the like; it generated almost 200 million euro of revenues in 2017 and employs 38.300 people \citep{istatr2019}. }
Lately, local cultural supply has evolved significantly. Indeed, most of the Italian museum circuits were founded relatively recently.\footnote{Approximately 2,300 sites (45.5\%)  of the Italian cultural supply chain were opened between 1960 and 1999, while 2,200 sites (38.6\%) were opened in 2000, taking advantage of the investments for economic recovery and infrastructure enhancement made for Italian cultural heritage sites \citep{istatr2016}.} 
In the following analysis, we have considered the ranking of the top 30 Italian cultural institutions (museums and the like) according to the highest number of annual visitors since 2004 (data source: \href{www.statistica.beniculturali.it}{www.statistica.beniculturali.it}). Among these, only 20 of the internationally renowned institutions remained ranked in the top 30 on a yearly basis. The other 10 positions were held by museums and institutions which only temporarily experienced an outstanding flow of visitors. These observations led to a twofold issue. First, the need to set-up a price index finalized at registering changes in the period under consideration. Second, a dynamic updating method of the index in order to preserve the information associated with the 10 positions not consistently present in the top 30 ranking (aspects which are not considered in the approaches usually adopted by most statistical agencies). 
The multi-period version of the MPL price index proves suitable for this scope. Hence, it has been applied to the data set which collects the number of visitors and revenues of the 30 leading Italian cultural institutions which, from 2004 to 2017, were ranked at least two times in the top 30.\footnote{All analyses in this investigation have been made with our own codes, written in R.} 

Figure~\ref{fig:5} shows the MPL price index together with its annual percentage variations for the period 2004--2017: 2004 being the base year and 2017 the year used for updating the index. Here the computation of the index has been done by assuming non spherical, and in particular heteroschedastic and uncorrelated (GLS-d) error terms.
Looking at the graph, we can note that in the early years of the new Millennium, when important investments started being made in the Italian cultural sector, the prices of museums (and the like) tickets grew (Figure~\ref{fig:5}). Thereafter, the price dynamic became more moderate and then tapered in 2009 and 2014 when, the so called ``W'' recession, namely the international financial and debt crisis in European peripheral countries, hit Italy. 

The MPL has been also computed under other specifications of the error terms and by taking into account the possibility of missing commodities. In particular, it has been also worked out by assuming heteroschedastic and correlated (GLS-f), stationary  (GLS-s) and spherical error terms (OLS), both in the case when the commodities are present in all period (case of complete price tableaux) and when some of them are missing (case of incomplete price tableaux). In Appendix~\ref{app:mplapp} the graphs of the MPL index in all these cases are shown.  
%
\begin{figure}[htbp]
\begin{center}
\includegraphics[scale=0.42]{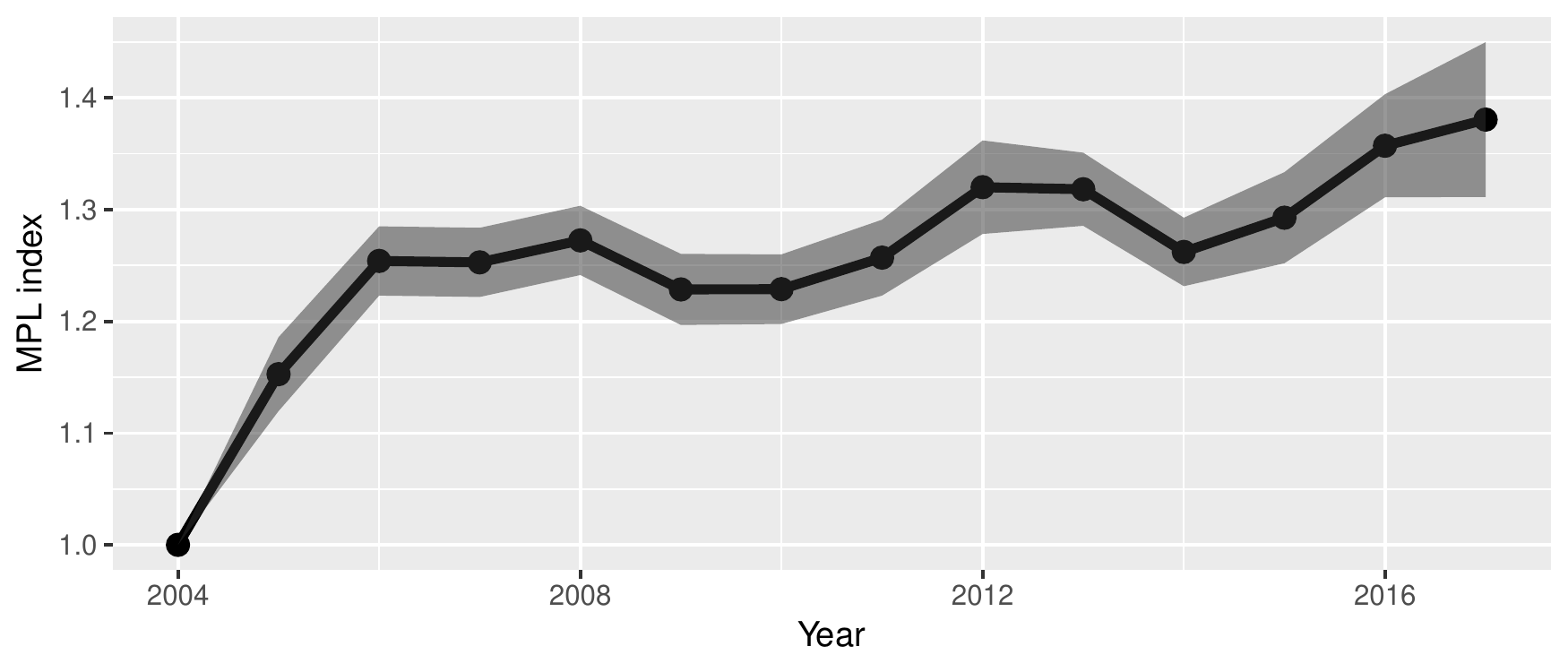}
\includegraphics[scale=0.42]{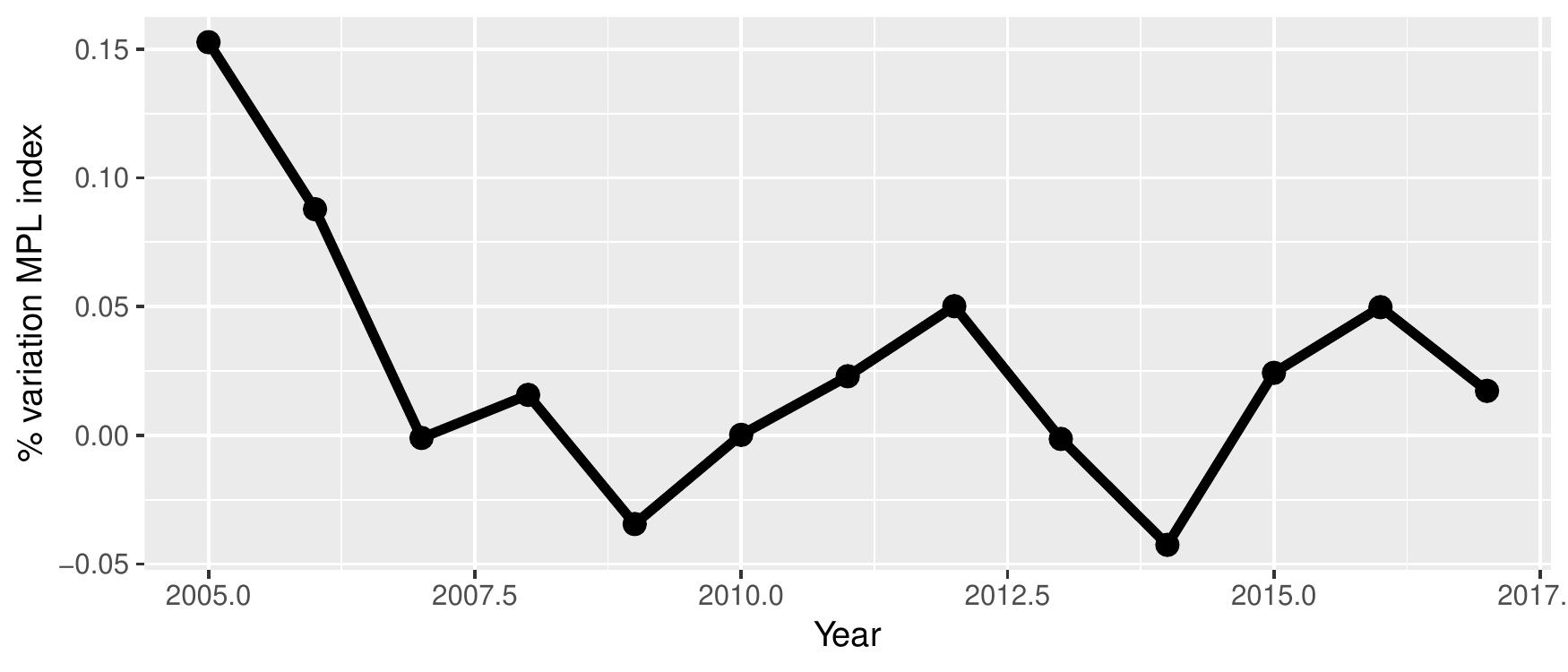}
\caption{MPL index and percentage annual change assuming heteroschedastic and uncorrelated error terms (GLS-d). }
\label{fig:5}
\end{center}
\end{figure}
%
 For the sake of further evidence from an empirical standpoint, the MPL has been also compared to the TPD index. 
%

Figure~\ref{fig:5g} shows both the MPL, computed under the assumption of heteroschedastic and uncorrelated error terms, and TPD price indices together with their $2\sigma$ confidence bounds. In the left panel the price indices have been computed only for those museums whose prices are available at all times. This has led to a drop in the number of museums/monuments/archaeological sites from 36 to 17 (note that this case corresponds to the ``standard'' reference basket). The right panel shows these indices computed with data from 36 museums ranked in the top 30 at least twice together with their $2\sigma$ confidence bounds. In this case, when not all items (museums) are priced in all periods, TPD estimates have been obtained by using the time version of the weighted CPD \citep[][pp. 420-421]{rao2004country}.  
Looking at both panels we see that the two indices are aligned, but that the MPL one always fall within the confidence bounds of TPD indices. This result provides evidence of the MPL greater efficiency, due to its lower standard errors. 
For completeness, in Appendix~\ref{app:mplapp} 
the estimated reference prices, under the different error specifications, have been displayed. 

\begin{figure}[htbp]
\begin{center}
\includegraphics[scale=0.42]{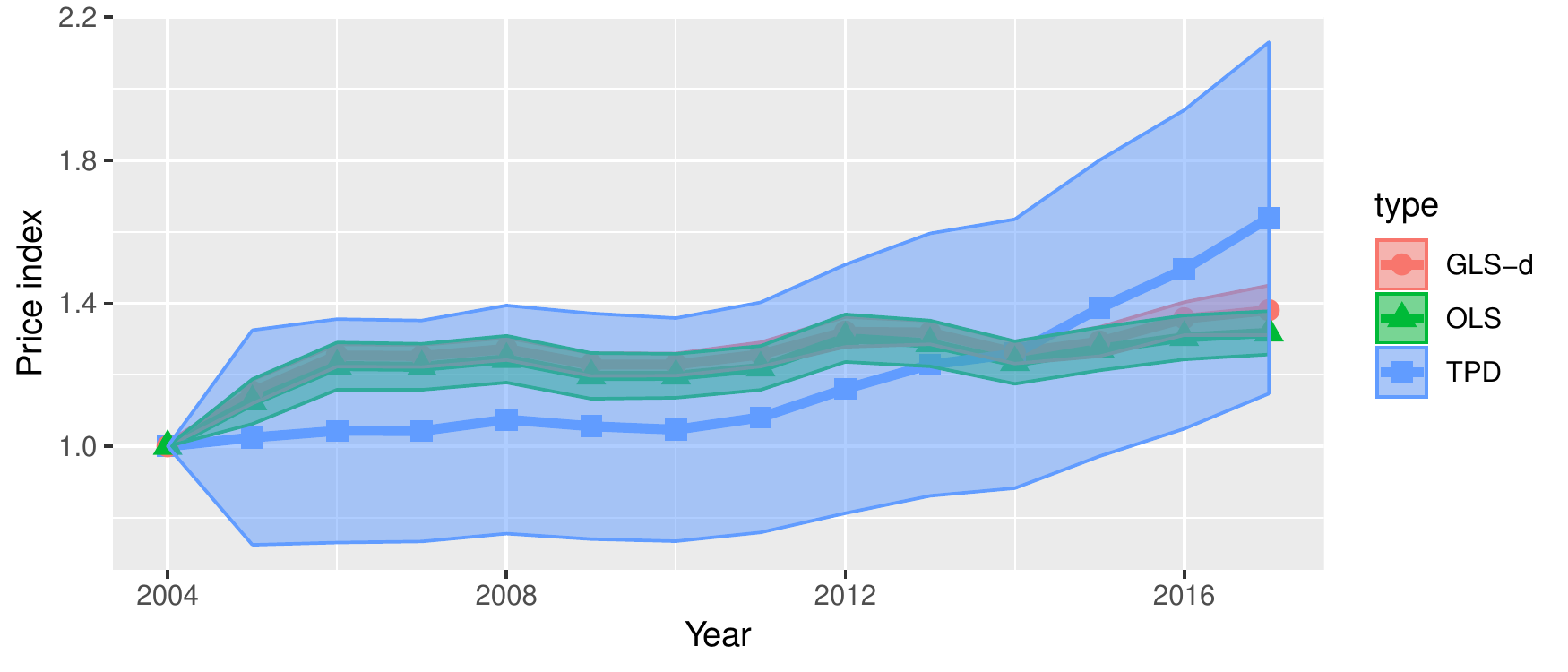} \includegraphics[scale=0.42]{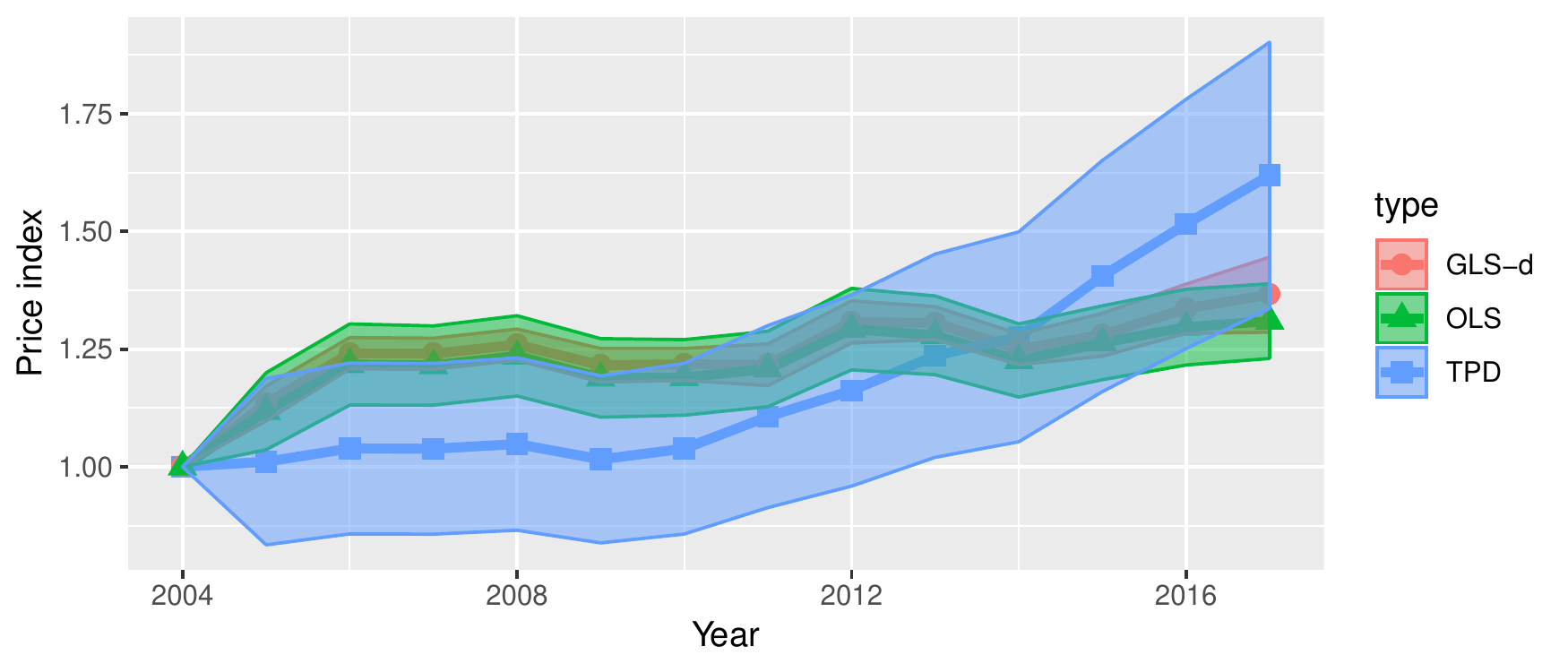}
\caption{The left panel compares the MPL index, computed for the GLS-d and OLS specification for the error terms, to the TPD index. Both the indices have been built by using data from 17 museums always ranked in the top 30; The right panel compares the MPL index, computed for the GLS-d and OLS specification for the error terms, to weighted TPD index. Here data from 36 museums ranked in the top 30 at least twice have been used to compute both indices.
}
\label{fig:5g}
\end{center}
\end{figure}

 The availability of data on visitors and revenues in 2017 for museums, monuments, archaeological sites, and museum circuits in the North-West, North-East, Centre and South (which includes the two islands Sicily and Sardinia) has allowed the computation of the multilateral version of the MPL index. Looking at the data, we see that almost half (46.3\%) are located in the North, while 28.5\% in the Centre, and 25.2\% in the South and Islands. The Regions with the highest number of cultural institutions are Tuscany (11\%), followed by Emilia-Romagna (9.6\%), Piedmont (8.6\%) and Lombardy (8.2\%) \citep{istatr2016}. However, alongside the more famous attractions, Italy is home to a wide and rich array of notable locations of cultural interest. A considerable percentage of these places (17.5\%) are found in municipalities with less than 2,000 inhabitants, but which can have up to four or five cultural sites in their small area. Almost a third (30.7\%) are distributed in 1,027 municipalities with a population varying from 2,000 to 10,000, and a bit more than half (51.8\%) are situated in 712 municipalities with a population of 10,000 to 50,000. Italy is, therefore, characterized by a strongly polycentric cultural supply distributed throughout its territory, even in areas considered as marginal from a geographic stance.

Table~\ref{tab:4b} reports, in the first three rows, the MPL index, computed under the assumption of stationary (GLS-s), non-spherical (GLS-d) and spherical (OLS) error terms, for three areas (North-West, North-East and Centre) considering the Centre as base area. The following three rows show updated values of the MPL index, computed for the  GLS-s, GLS-d and OLS specifications of the error terms, when the South-Islands are added to the data-set. As for the multi-period case, a comparison of the MPL estimates with the CPD ones is provided. The last row of Table~\ref{tab:4b} shows CPD estimates in the case of full price tableau, as all commodities are priced in the four geographic areas. Figure~\ref{fig:5c} shows both the MPL and CPD indices together with their $2\sigma$ confidence bounds. The comparison of the CPD and MPL indices computed under other specifications for the error terms (GLS-f and OLS) is provided in Appendix Appendix~\ref{app:mpllat}.  Once again, the estimates of the MPL index turn out to be more accurate than those provided by the CPD approach, as the former have standard errors lower than the latter. As in the comparison with the TPD index, the confidence bounds of CPD indices always include MPL estimates, thus suggesting the compatibility of the MPL index with the estimates provided with the CPD one. It is worth noting that in 2017, access to cultural sites in Southern Italy cost the most: almost twice as much as in the North-Eastern area. While the disparity could be ascribed to several factors, such as different costs of managing museums and similar institutions, tourism flows, etc: that type of analysis goes beyond the scope of the current investigation. 
\begin{table}[htbp]
\center
\caption{Updated MPL index assuming spherical (OLS) and non-spherical (GLS-d and GLS-s) error terms compared to the CPD index (standard error in parentheses).}
\label{tab:4b}
\begin{tabular}{lllll}
\hline\hline
                                     		& North West & North East  & Centre  & South    \\
\hline
MPL-GLS-s  	&1.070  (0.021) &0.634  (0.018)   &1.000    &   \\
MPL-GLS-d  	&1.162  (0.165) &0.632  (0.011)   &1.000    &   \\
MPL-OLS     &1.072  (0.200) &0.621  (0.194)   &1.000    &   \\
Updated MPL-GLS-f	&1.048 (0.078)  &0.665  (0.087)   &1.000    &1.107 (0.036) \\
Updated MPL-GLS-d		&1.148 (0.174)  &0.631  (0.008)   &1.000    &1.143 (0.001) \\
Updated MPL-OLS		&1.070 (0.176)  &0.622  (0.174)   &1.000    &1.142 (0.153) \\
CPD 					&1.524 (0.337)  &1.283  (0.284)	  &1.000 	  &1.021 (0.226)\\
\hline\hline
\end{tabular}
\end{table}
\begin{figure}[htbp]
\begin{center}
\includegraphics[scale=0.42]{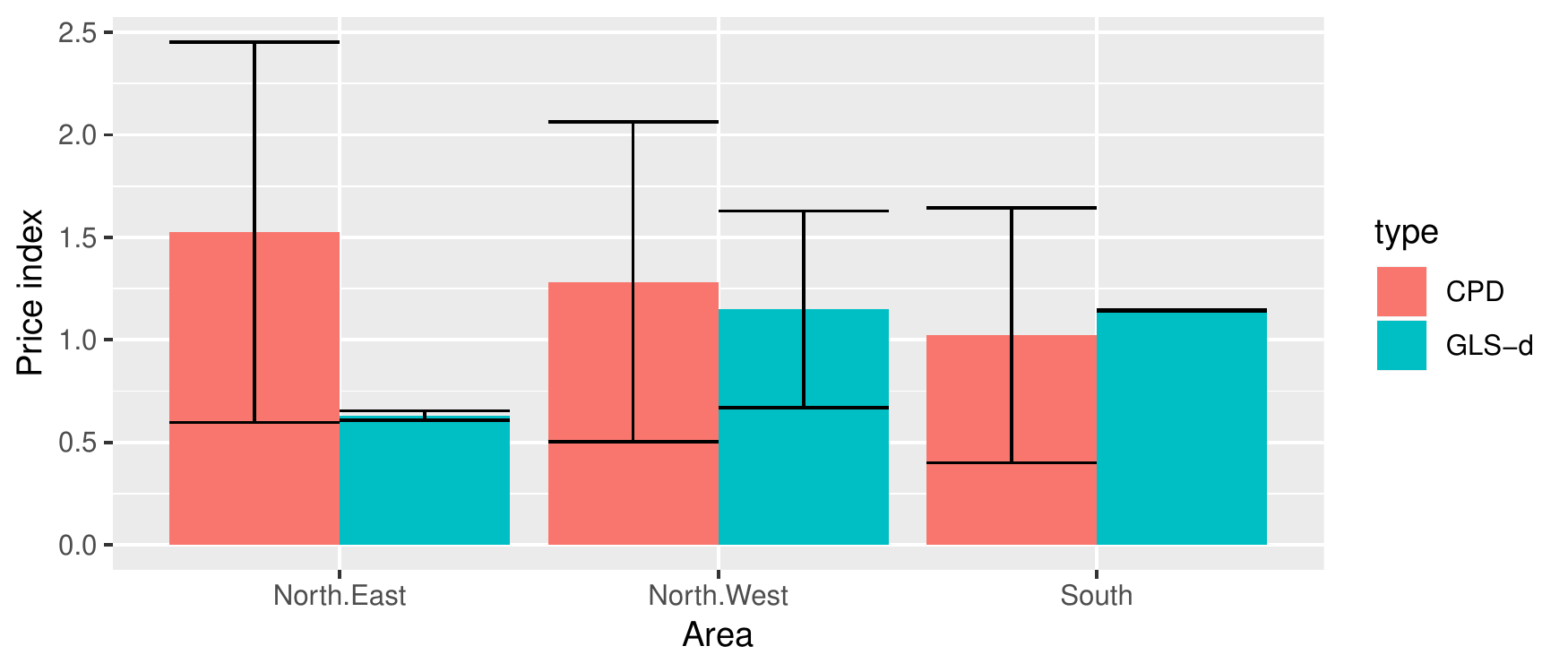}
\caption{GLS-d and CPD indices with their $2\sigma$ confidence bounds.}
\label{fig:5c}
\end{center}
\end{figure}

 Finally, in order to investigate more thoroughly the performance of the MPL index as compared to the TPD one, a simulation analysis has been performed   based on the perturbation of the original value matrix $\boldsymbol{V}_1$. In particular, one thousand simulations have been carried out by using perturbed values (and prices as a by-product) and assuming fixed quantities (i.e. equal to the original ones). Next, the simulated values (and prices) have been used to compute MPL and TPD indices in different settings: with and without missing quantities (and accordingly prices in the TPD model). The final MPL and TPD indices have been obtained as averages of all indices computed on simulated values (and prices). Two types of simulations have been carried out. First, simulated values from the $2$-nd to the $T$-th period (base period values, $\boldsymbol{v}_1$, being kept fixed) have been obtained by adding random perturbations, drawn from Normal laws with different means and variances, to the original values of $\boldsymbol{V}_1$. The plots in Figure~\ref{fig:sim1} show both the MPL and the TPD indices obtained by using these simulated data together with the associated $2\sigma$ confidence bounds. As before the MPL index  has been computed under the assumption of heteroschedastic and uncorrelated error terms (GLS-d).
\begin{figure}[htbp]
\begin{center}
\includegraphics[scale=0.42]{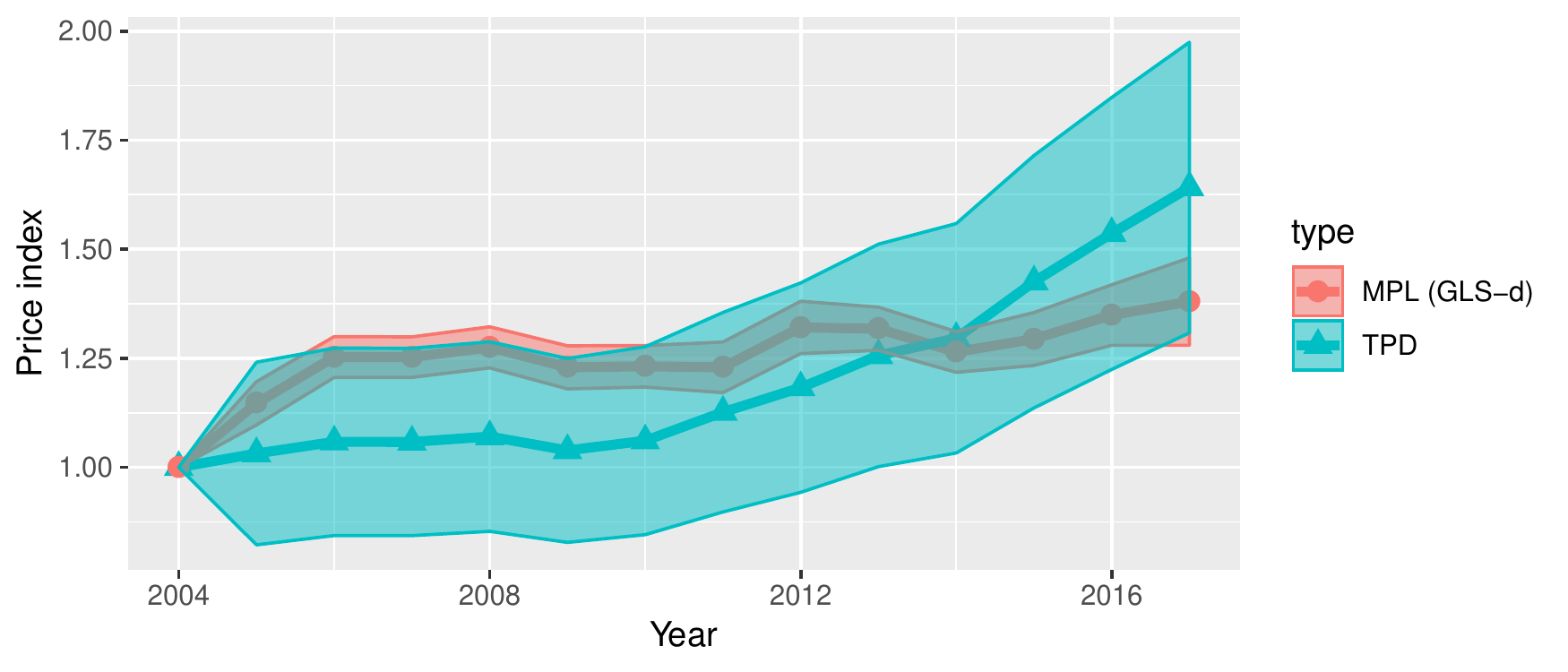} 
\includegraphics[scale=0.42]{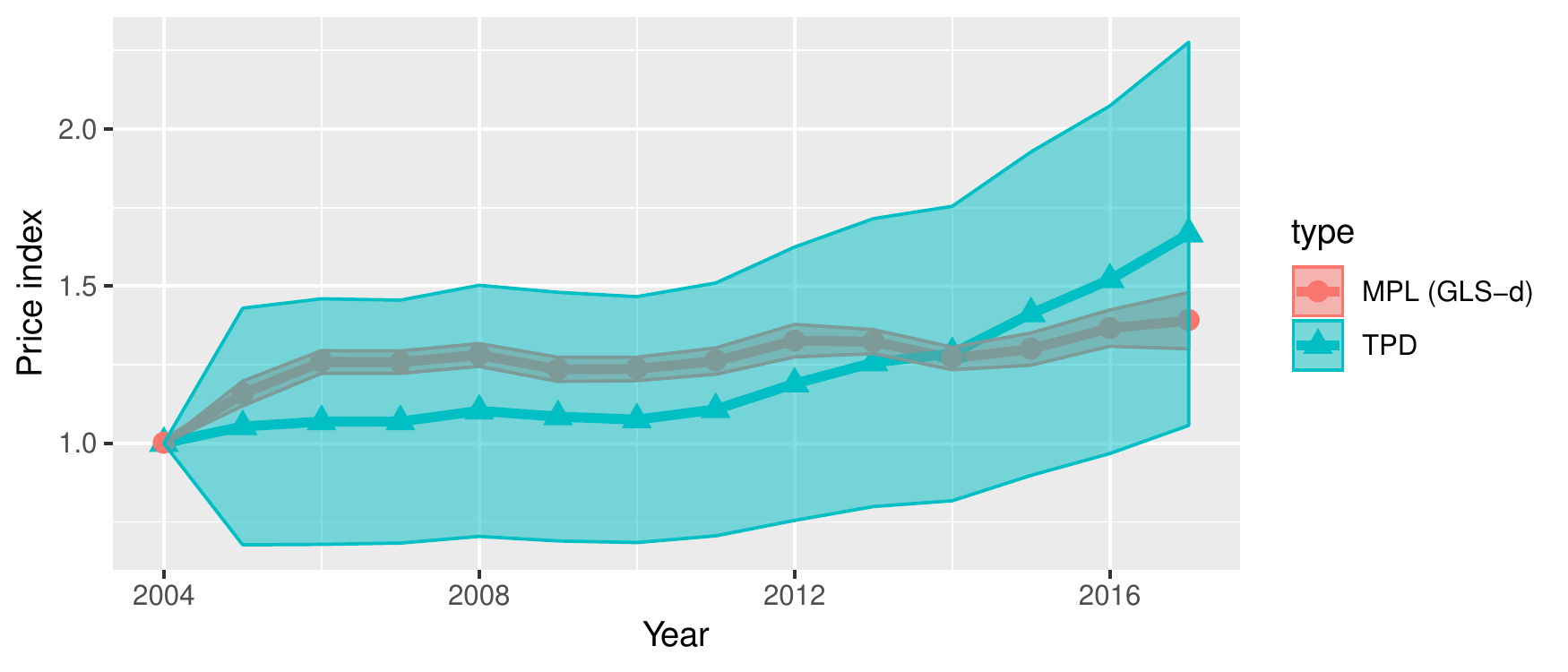}
\caption{MPL and TPD indices obtained by using simulated data obtained by adding random variables, generated from a Normal law (with a mean equal to 20000 and a standard error varying randomly from 0 to 1000), to the values of $\boldsymbol{V}_1$. The left and right panels refer to the complete and incomplete price tableau scenario, respectively.  
}
\label{fig:sim1}
\end{center}
\end{figure}
Then, following another approach, simulated values, from the $2$-nd to the $T$-th period (base period values, $\boldsymbol{v}_1$, being kept fixed), have been obtained from simulated values of the previous period with the addition of perturbation terms, drawn from Normal laws with given means and variances. Plots in Figure~\ref{fig:sim2} show both the MPL, and the TPD indices, for the case of complete and incomplete price tableau, together with the associated $2\sigma$ confidence bounds. 
As before the MPL index  has been computed under the assumption of heteroschedastic and uncorrelated error terms (GLS-d). Looking at these figures, we see that in both cases, the MPL estimates are in line with the TPD ones, but are more accurate than the latter as their tighter confidence bounds show. 
\begin{figure}[htbp]
\begin{center}
\includegraphics[scale=0.42]{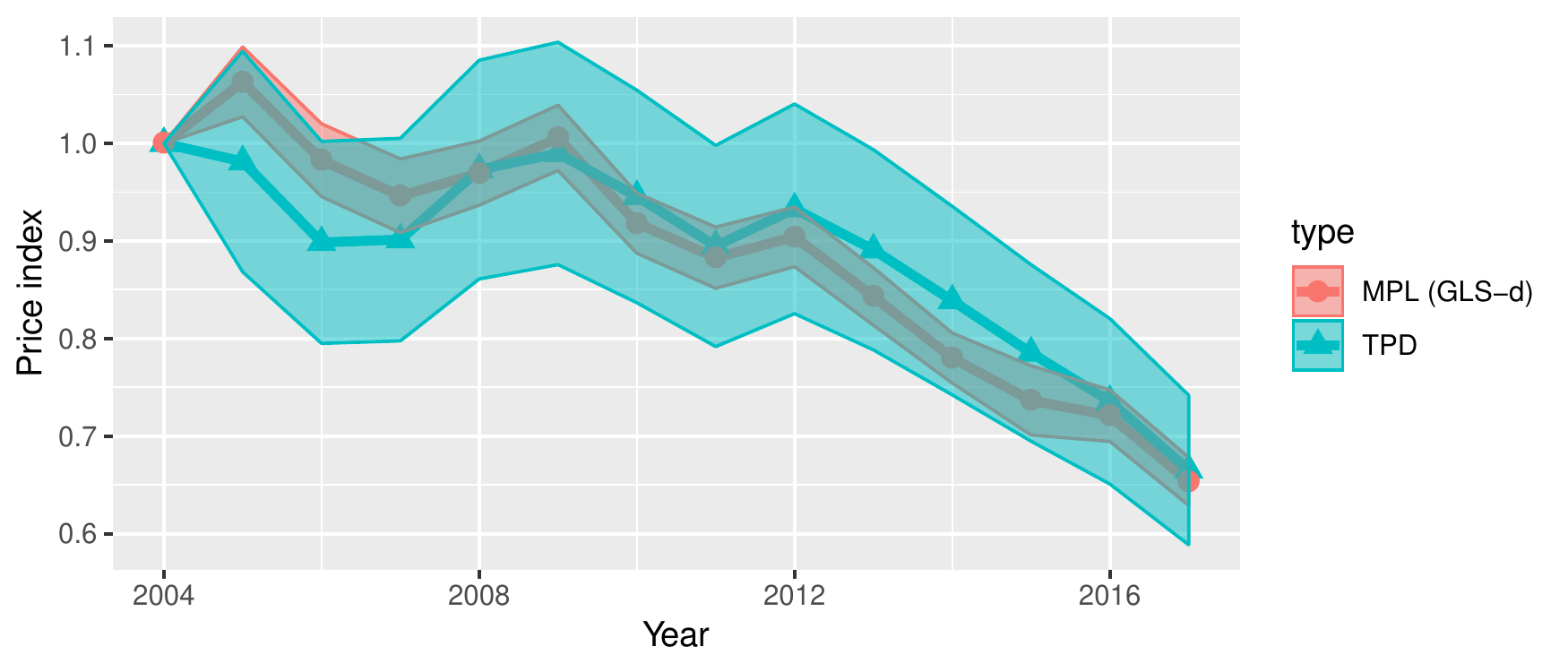} 
\includegraphics[scale=0.42]{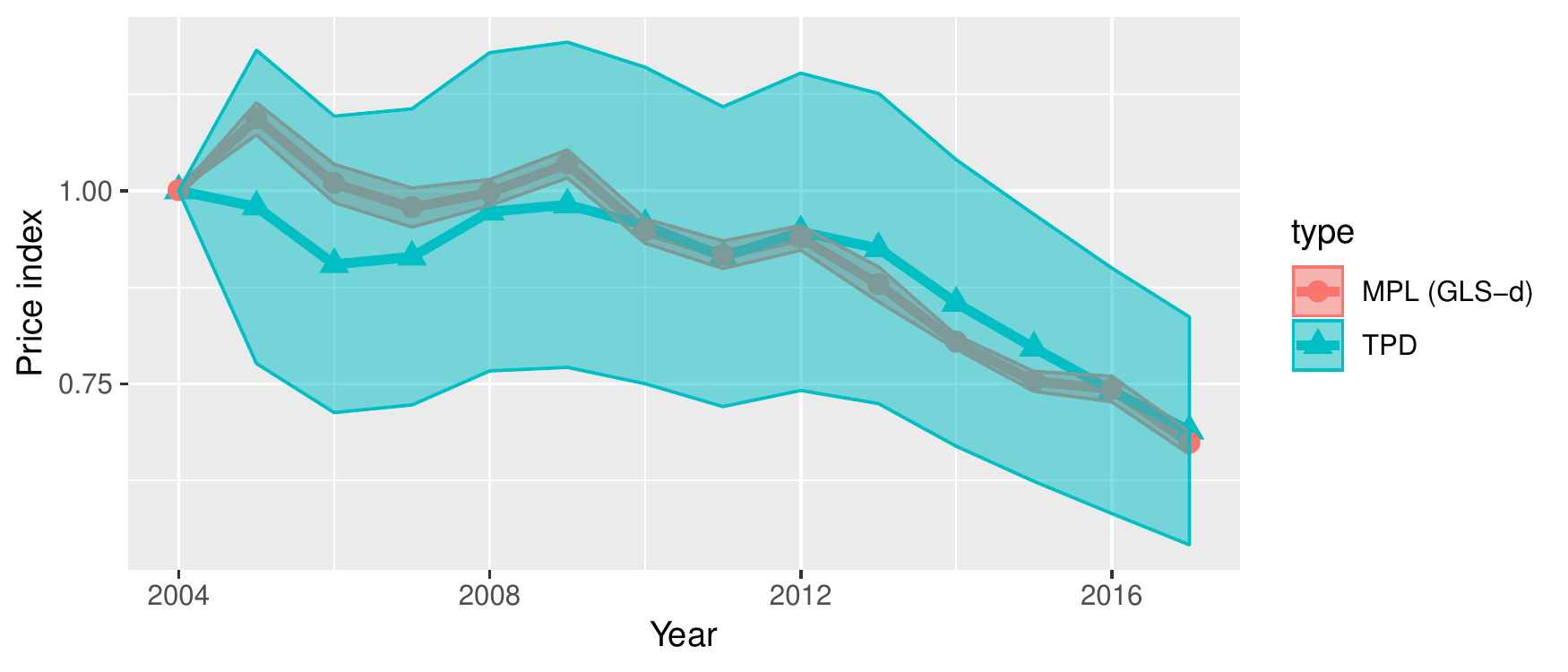}
\caption{MPL and TPD indices obtained, at each time $t$, for $t=2,\dots,T$, by adding perturbation terms (drawn from a Normal law with a mean equal to -5000 and a standard error varying randomly from 0 to 800) to the simulated values $\boldsymbol{\upsilon}_{t-1}$. The left and right panels refer to the complete and incomplete price tableau scenario respectively.}
\label{fig:sim2}
\end{center}
\end{figure}
In Appendix~\ref{app:mplsim}, a comparison of the TPD and the MPL indices, always computed on data simulated and by assuming other specifications for the error terms, (GLS-f, GLS-s and OLS), is provided.  

\section{A simulation example}\label{sec:simulation}
To further illustrate the potentialities of the MPL index, the latter has been computed by using a simulated data set built as follows. For given  
$\boldsymbol{q}_t$, $\tilde{\boldsymbol{p}}_t$ and $\lambda_t$, the values $\boldsymbol{v}_t$ have been computed according to Eq.~\ref{eq:kjk}. The random terms $\epsilon_t$ of this equation have been generated from a standard Normal distribution. These elements represent the ingredients of the system in Eq.~\ref{eq:sist1} used to work out the MPL index. The aim of this simulation is to see the capability of the MPL and TPD indices  to reproduce the values of the ``true'' index $\boldsymbol{\lambda}$.
%
%
Let's assume that the quantities, $\boldsymbol{Q}$, the reference prices, $\tilde{\boldsymbol{p}}$ and the values of the price index, $\boldsymbol{\lambda}$, of four commodities from 2015 to 2020 are specified as follows
\[\begin{split}
&\boldsymbol{Q}=
\begin{bmatrix}
5    & 8   & 10    &10   & 15   & 20\\
15   & 18 &   20&    10 &   15  &  10\\
25   & 27  &  30 &   35 &   30 &   20\\
5    & 5    & 5   & 10    &15  &  20
\end{bmatrix},\,\,\,
\boldsymbol{\lambda}=
\begin{bmatrix}
1.00\\
1.11\\
1.18\\
1.15\\
1.25\\
1.27
\end{bmatrix}
,\,\,\,
\tilde{\boldsymbol{p}}=
\begin{bmatrix}
2.1\\
1.5\\
0.9\\
1.9
\end{bmatrix}.
\end{split}
\] 
\noindent Then, with these data at hand, the values computed as in Eq.~\ref{eq:kjk} result to be
\[\begin{split}
&\boldsymbol{V}=
\begin{bmatrix}
 9.29 &19.10 &24.14 &23.36 &38.17 &53.59\\
22.78  &30.51 &34.40 &17.31 &28.40 &19.49\\
19.08 &26.43 &31.29 &37.17 &34.83 &22.21\\
7.15 &10.01 &10.18 &21.73 &33.28 &47.55
\end{bmatrix}
\end{split}
\] 
while the price matrix, needed to compute the TPD index, $\lambda_{TPD}$ hereafter, is
\[\begin{split}
&\boldsymbol{P}=
\begin{bmatrix}
1.86 &2.39 &2.41 &2.34 &2.54 &2.68\\
1.52 &1.69 &1.72 &1.73 &1.89 &1.95\\
0.94 &0.98 &1.04 &1.06 &1.16 &1.11\\
1.43 &2.00 &2.04 &2.17 &2.22 &2.38
\end{bmatrix}.
\end{split}
\] 

\noindent The matrices $\boldsymbol{Q}$ and $\boldsymbol{V}$ have been used to compute both the MPL and the TPD indixes, in a multi-period perspective. In particular, the MPL index has been computed under several specification of the error terms and, more precisely, 
stationary (GLS-s), heteroschedastic and uncorrelated (GLS-d), spherical (OLS) error terms.

We propose three different examples in which the MPL is compared with the TPD in the following cases
\begin{enumerate}
    \item complete price tableau, implying a reference basket including the complete set of the four commodities;
    \item incomplete price tableau, assuming missing the second and forth commodity, (that is $q_{41}=v_{41}=0$ and $q_{22}=v_{22}=0$), with a   ``standard'' reference basket, (see the left-hand side of Figure~\ref{fig:1}) that, accordingly includes only the first and the third commodities;
    \item incomplete price tableau assuming missing the second and forth commodity, (that is $q_{41}=v_{41}=0$ and $q_{22}=v_{22}=0$), with the   ``MPL'' reference basket, (see the right-hand side of Figure~\ref{fig:1}), that includes  commodities present in at least two periods, namely all the four commodities.
\end{enumerate}
The rows of Table~\ref{tab:perform} provides the sum of the squares of the differences between the estimated indices (GLS-s, GLS-d, OLS, and TPD) $\hat{\boldsymbol{\lambda}}$ and the index $\boldsymbol{\lambda}$ for the three cases . Looking at this table, we see that the MPL index, whatever is the specification assumed for the error terms, provides always the best fit to the index $\boldsymbol{\lambda}$.
Thus, for any specification of the error terms, the MPL index exhibits an higher performance than the TPD, except for the GLS-d in the third case.
 In Appendix~\ref{app:sim}, the graphs of both the TPD and MPL indixes are provided under different specifications for the error terms (GLS-s, GLS-d; OLS).  In all cases the MPL estimates turn out to be more accurate, as they have lower variances and, consequently, they are always included in a $2\sigma$ confidence band of the TPD (see Figure~\ref{fig:5sim2}). 
%
%

These examples are also particularly interesting to highlight the role played by the reference prices $\tilde{\boldsymbol{p}}$, which are the prices that consumers are expected to pay for the commodities in the base period/country.
Reference prices prove useful 
to obtain estimates of the prices of those commodities which, being missing in the basket, can not be determined. 
This case occurs in the third example with incomplete price tableau and reference basket including the complete set of commodities. In this case, if a commodity, say $j$, is missing in a period, say $i$, then its price, even if different from zero, turns out to be undetectable. 
However, in the  MPL approach, it can be determined, through the estimates of both its associated reference price $\hat{\tilde{p}}_j$ and price index $\hat{\lambda}_i$, as follows 
$\hat{\tilde{p}}_j\, \hat{\lambda}_i$.
This strategy has been used to estimate the prices of the second and  fourth commodity in the third case. \\To assess  the goodness of the estimates
$\hat{\tilde{p}}_j\, \hat{\lambda}_i$, $j=2,4$,  $i=1,...,5$ in reproducing 
the real prices $\tilde{p}_j\, \lambda_i$, $j=2,4, \enspace i=1,...,5$,
the sum of the squares between observed and estimated prices have been computed for all commodities, either included in the basket or missing. The results, reported in Table~\ref{tab:perform_ref} in Appendix~\ref{app:sim}, are very satisfactory.
Looking at Figure~\ref{fig:6sim2} in Appendix~\ref{app:sim}, which compares the estimates of the reference prices with the ``real'' ones,  it is clear that the MPL proves able to suitably estimate the reference prices for all commodities, also for the missing ones.
 Figure~\ref{fig:6sim3} compares the values of the MPL and TPD indexes $\hat{\boldsymbol{\lambda}}$ to the real price index $\boldsymbol{\lambda}$ of this experiment. Interestingly, differently from the TPD, the MPL  better captures the trend of the ``real'' price index over time, avoiding a TPD overestimation issue (see Table~\ref{tab:perform}) in all cases (except for the GLS-d incomplete price tableau with the novel basket). 
 
Finally, in Table~\ref{tab:perform_b} in Appendix~\ref{app:sim} ``real'' prices $\boldsymbol{P}$ have been compared with their estimates obtained by using the MPL and the TPD index, given by $\hat{\boldsymbol{P}}_{MPL}=\tilde{\boldsymbol{p}}\hat{\boldsymbol{\lambda}'}$ and $\widehat{\boldsymbol{P}}_{TPD}=\tilde{\boldsymbol{p}}\hat{\boldsymbol{\lambda}}_{TPD}'$ for the former and the latter,  respectively.  
The better performance of the MPL compared to the TPD one in reproducing the sequences of prices emerges from Table~\ref{tab:perform_b}, providing the sum of the squares of the differences between the real prices and the ones estimated by the two indexes.


\begin{table}[htb]
\center
\caption{Sum of square of the differences between the estimated indices and the real one.}
\label{tab:perform}
\begin{tabular}{lrrrrr}
\hline
\hline
Data                                                               &  & GLS-s   & GLS-d   & OLS     & TPD     \\
\hline
\hline
\begin{tabular}[c]{@{}l@{}}1. Complete price \\ tableau\end{tabular}  &       & 0.00308 & 0.00311 & 0.00333 & 0.04327 \\
\begin{tabular}[c]{@{}l@{}}2. Incomplete price\\ tableau (``classical basket'')\end{tabular} &       & 0.00053        & 0.00327        & 0.00126        &   0.00520     \\
\begin{tabular}[c]{@{}l@{}}3. Incomplete price \\ tableau (novel basket)\end{tabular} &       & 0.00219        & 0.00352        & 0.00212        &   0.00301     \\
\hline
\hline
\end{tabular}
\end{table}

\begin{figure}[htbp]
\begin{center}
\includegraphics[scale=0.42]{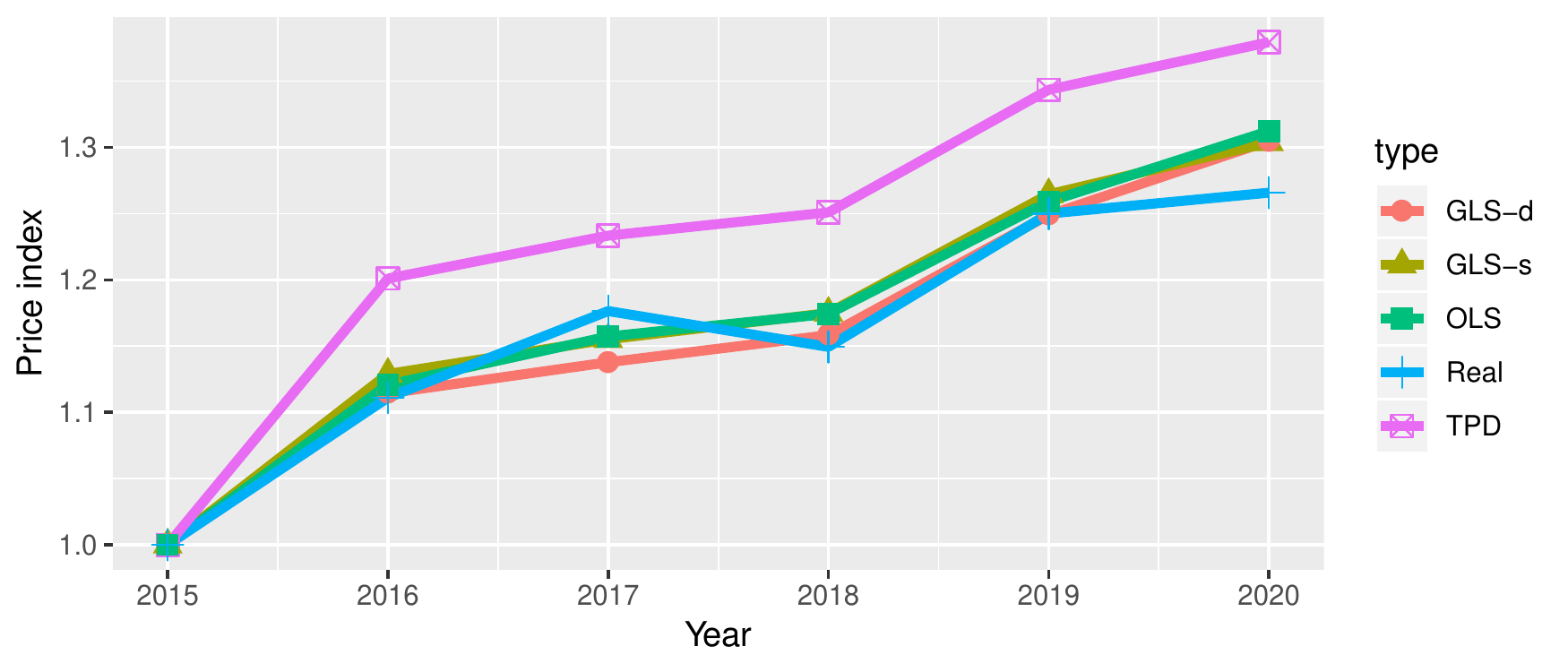} \includegraphics[scale=0.42]{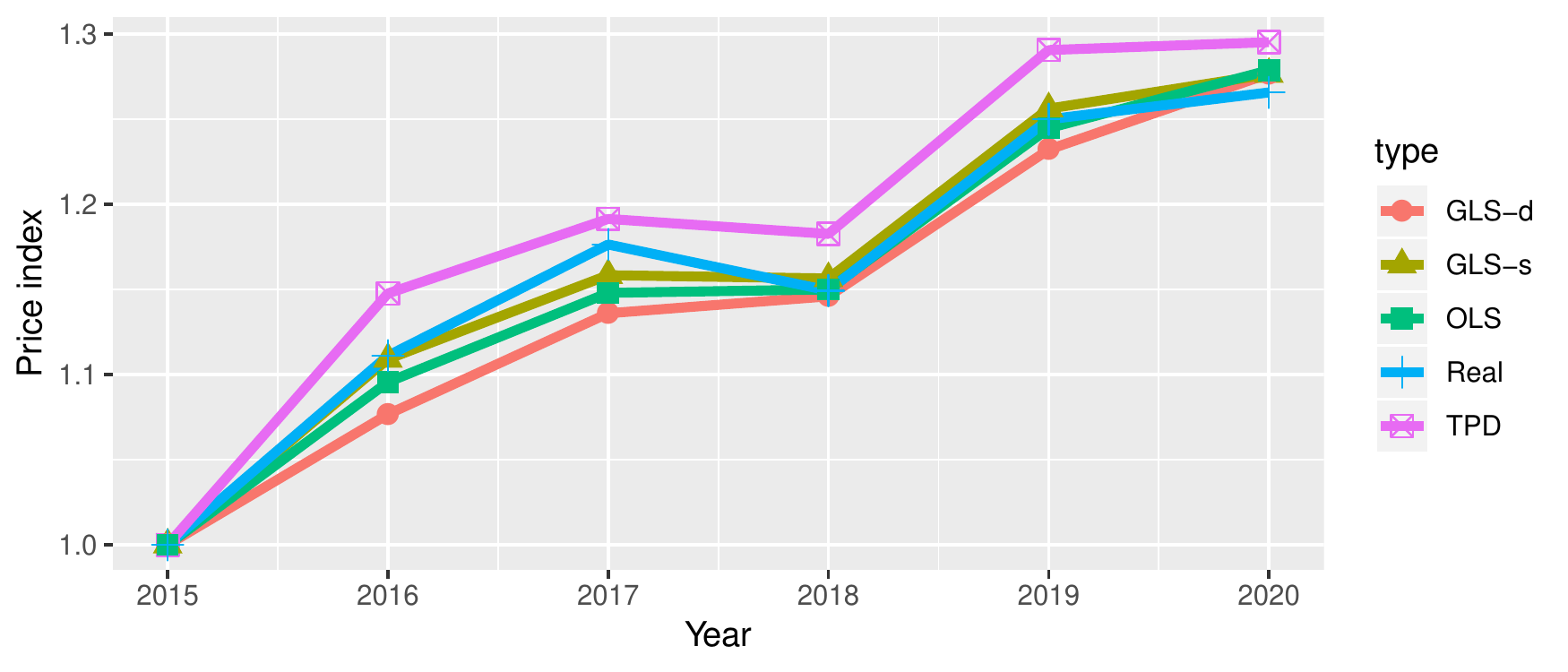}\\ \includegraphics[scale=0.42]{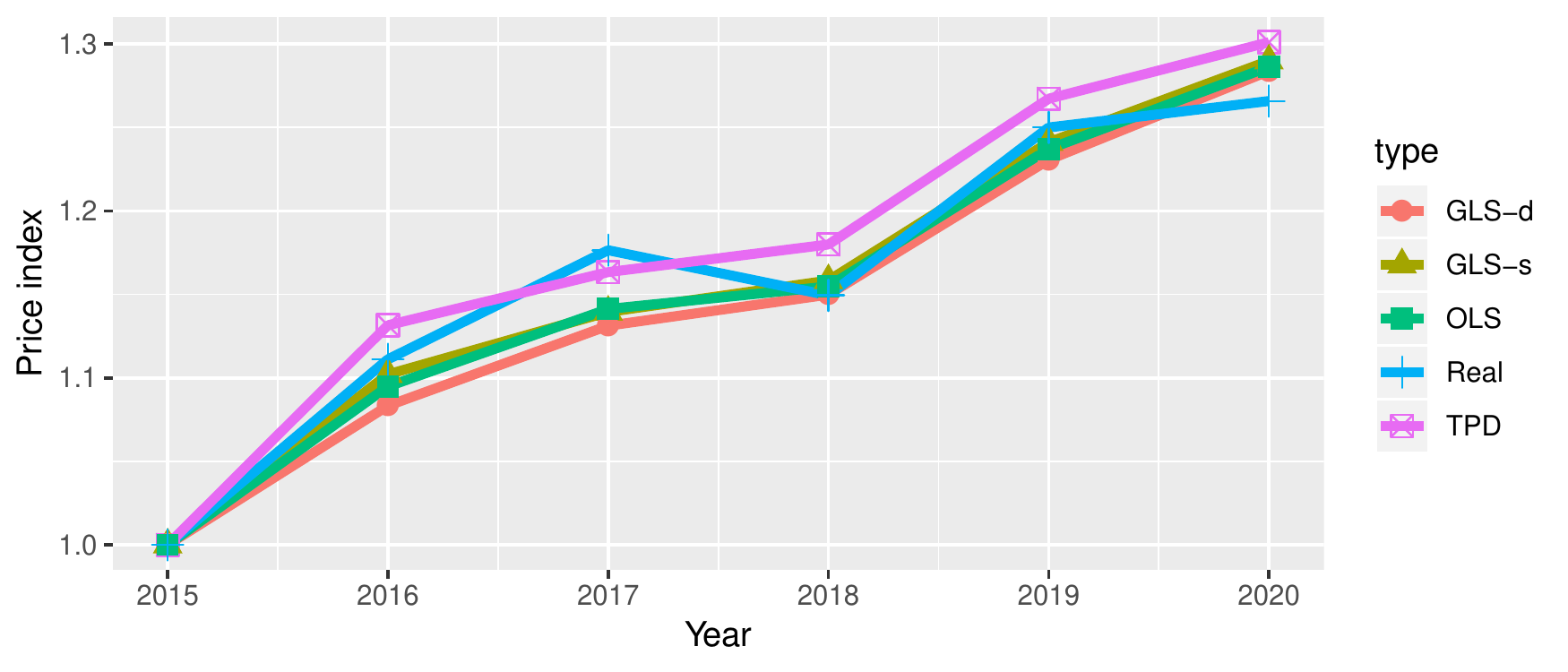}
\caption{
Comparison across different version of the MPL and TPD without confidence bands are proposed for the three examples: complete price tableau, incomplete price tableau ``standard'' reference basket, incomplete price tableau novel reference basket, respectively.
}
\label{fig:5sim3}
\end{center}
\end{figure}



\section{Conclusion}\label{sec:conclusion}
 The paper works out a novel price index that can be used either as a multi-period or as a multilateral index. This index, called MPL index, is obtained as a solution to an ``ad hoc'' minimum-norm criterion, within the framework of the stochastic approach. The computation of the  MPL index does not require the knowledge of commodity prices, but only their quantities and values. The reference basket of the MPL index, over periods or across countries, is more informative and complete than the ones commonly used by statistical agencies, and easy to update. The updating process is twofold depending on the multi-period or the multilateral use of the index. An application of the MPL index to the Italian cultural supply data provides proof of its positive performance. A comparison between the MPL and the CPD/TPD index on both real and simulated data provides evidence of the greater efficiency of the MPL estimates.
 
 Thus, the MPL index is very promising  both in the multi-period and the multilateral perspective, also considering the simple and efficient way of updating the index series  \citep{diewert2020substitution}. 
  
 The approach here proposed can be extended along several paths.
 For instance, the application of the MPL index for a multilateral comparison across countries using different currencies would require a suitable adjustment to ensure comparability among countries. This could be done by using a set of purchasing power parities to convert the different currencies into a common one. In this case, the MPL index could be built by employing  ``international'' quantities and ``country volumes'' as suggested by \citep{balk1996comparison}.
 Furthermore, the MPL approach could be 
 employed to construct price indexes across both space and time as in \citep{hill2004constructing}. Both these research lines are being investigated.

\section*{Acknowledgements}
 \noindent We sincerely thank Prof. E. Diewert for his valuable and constructive suggestions, as well as precious comments during the development of this article.
 
All remaining errors, typos or inconsistencies of this work are of our own.
 
\clearpage
\bibliographystyle{apalike}
\bibliography{FNZ_literature.bib}
%
\clearpage 
\processdelayedfloats
\clearpage

\newpage
%
\section*{Appendix}\label{app:proofs}
\small

\setcounter{equation}{0} \renewcommand{\theequation}{A.\arabic{equation}}
\setcounter{section}{0}

\section{Proofs of Theorems and Corollaries}

\subsection{Proof of Theorem~\ref{th:2gls}}\label{app:proofs2gls}
\begin{proof} 
In compact form, the model in Eq.~\eqref{eq:sist1} can be written as
\begin{equation}\label{eq:reg0}
\underset{(NT, 1)}{\boldsymbol{y}}=\underset{(NT, N+T-1)}{\boldsymbol{X}_{}}\underset{(N+T-1, 1)}{\boldsymbol{\beta}} + \underset{(NT, 1)}{\boldsymbol{\mu}}
\end{equation} 
where 
\[
\underset{(NT,1)}{\boldsymbol{y}}=\begin{bmatrix} \underset{(N,1)}{\boldsymbol{v}_1}\\\underset{(N(T-1),1)}{\boldsymbol{0}}
\end{bmatrix}, \,\,\underset{(NT,N+T-1)}{\boldsymbol{X}_{}}=\begin{bmatrix} \underset{(N,T-1)}{\boldsymbol{0}} &  \underset{(N,N)}{(\boldsymbol{q}'_1\otimes \boldsymbol{I}_N)\boldsymbol{R}_N'}\\ \underset{(N(T-1),T-1)}{(\boldsymbol{I}_{T-1}\otimes (-\boldsymbol{V}_1))\boldsymbol{R}_{T-1}'} & \underset{(N(T-1),N)}{(\boldsymbol{Q}'_1\otimes \boldsymbol{I}_N)\boldsymbol{R}_N'}
\end{bmatrix}
\]
and
\[
\underset{(N+T-1,1)}{\boldsymbol{\beta}}=\begin{bmatrix} \underset{(T-1,1)}{{\boldsymbol{\delta}}}\\\underset{(N,1)}{\tilde{\boldsymbol{p}}}
\end{bmatrix}, \,\,\underset{(NT,1)}{\boldsymbol{\mu}}=\begin{bmatrix} \underset{(N,1)}{\boldsymbol{\varepsilon}_1} \\ \underset{(N(T-1),1)}{\boldsymbol{\eta}}
\end{bmatrix}.
\]
\noindent The generalized least square estimator of the vector $\boldsymbol{\beta}$ is given by
\begin{equation}
\label{eq:beta}
\widehat{\boldsymbol{\beta}}_{GLS}=(\boldsymbol{X}'\boldsymbol{\Omega}^{-1}\boldsymbol{X})^{-1}\boldsymbol{X}'\boldsymbol{\Omega}^{-1}\boldsymbol{y}
\end{equation}
where, $\boldsymbol{\Omega}$ is as defined in Eq.~\eqref{eq:11}.

Some computation prove that, 
\[
\begin{split}
&\underset{(N+T-1,N+T-1)}{\boldsymbol{X}'\boldsymbol{\Omega}^{-1}\boldsymbol{X}}=\\
&
\begin{bmatrix} 
\boldsymbol{R}_{T-1}\left[(\boldsymbol{I}_{T-1}\otimes\boldsymbol{V}_1'){\boldsymbol{\Omega}^{*}}^{-1}(\boldsymbol{I}_{T-1}\otimes\boldsymbol{V}_1)\right] \boldsymbol{R}_{T-1}' 
&- \boldsymbol{R}_{T-1}\left[(\boldsymbol{I}_{T-1}\otimes\boldsymbol{V}_1'){\boldsymbol{\Omega}^{*}}^{-1}(\boldsymbol{Q}'_1\otimes\boldsymbol{I}_{N})\right] \boldsymbol{R}'_{N}
\\ -\boldsymbol{R}_{N}\left[(\boldsymbol{Q}_1\otimes\boldsymbol{I}_{N}){\boldsymbol{\Omega}^{*}}^{-1}(\bm{I}_{T-1}\otimes\boldsymbol{V}_1)\right]\boldsymbol{R}_{T-1}' 
& \boldsymbol{R}_{N} (\boldsymbol{q}_1\boldsymbol{q}_1'\otimes \boldsymbol{\Omega}_{11}^{-1})\bm{R}'_{N}+\bm{R}_N(\boldsymbol{Q}_1\otimes\boldsymbol{I_N}) {\boldsymbol{\Omega}^{*}}^{-1}(\boldsymbol{Q}_1'\otimes\boldsymbol{I_N})\boldsymbol{R}_N'
\end{bmatrix}=
\\
\\
&\begin{bmatrix}
\boldsymbol{R}_{T-1}\left[\sum_{j=1}^{T-1}(\boldsymbol{e}_{j}\boldsymbol{e}_{j}'\otimes \boldsymbol{V}_1'{\boldsymbol{\Omega}^{-1}_{j+1,j+1}}\boldsymbol{V}_1)\right] 
\boldsymbol{R}_{T-1}'
&-\boldsymbol{R}_{T-1}\left[\sum_{j=1}^{T-1}(\boldsymbol{e}_{j}\otimes\boldsymbol{V}_1'{\boldsymbol{\Omega}^{-1}_{j+1,j+1}})(\boldsymbol{q}_j'\otimes \boldsymbol{I}_{N} )\right]\boldsymbol{R}'_{N}
\\ -\boldsymbol{R}_{N}\left[\sum_{j=1}^{T-1}(\boldsymbol{e}_{j}'\otimes{\boldsymbol{\Omega}^{-1}_{j+1,j+1}}\boldsymbol{V}_1)(\boldsymbol{q}_j\otimes \boldsymbol{I}_{N} )\right]\boldsymbol{R}_{T-1}'
& \boldsymbol{R}_{N}\left[\sum_{j=1}^{T}\boldsymbol{q}_j\boldsymbol{q}_j'\otimes{\boldsymbol{\Omega}^{-1}_{j,j}}\right]\boldsymbol{R}_N'
\end{bmatrix}=\\
\\
&\begin{bmatrix}
\left[\sum_{j=1}^{T-1}(\boldsymbol{e}_{j}\boldsymbol{e}_{j}'\ast \boldsymbol{V}_1'{\boldsymbol{\Omega}^{-1}_{j+1,j+1}}\boldsymbol{V}_1)\right] 
&-\left[\sum_{j=1}^{T-1}(\boldsymbol{e}_{j}\boldsymbol{q}_{j+1}'\ast\boldsymbol{V}_1'{\boldsymbol{\Omega}^{-1}_{j+1,j+1}})\right] \\
-\left[\sum_{j=1}^{T-1}(\boldsymbol{e}_{j}'\boldsymbol{q}_{j+1}\ast{\boldsymbol{\Omega}^{-1}_{j+1,j+1}}\boldsymbol{V}_1)\right]
&\left[\sum_{j=1}^{T}\boldsymbol{q}_j\boldsymbol{q}_j'\ast{\boldsymbol{\Omega}^{-1}_{j,j}}\right]
\end{bmatrix}
=\\
\\
&\begin{bmatrix}
\boldsymbol{I}_{T-1}\ast \boldsymbol{\widetilde{V}}_1'\boldsymbol{V}_1
& - \boldsymbol{Q}_1'\ast \boldsymbol{\widetilde{V}}_1' \\
-\boldsymbol{Q}_1 \ast \boldsymbol{\widetilde{V}}_1
& \sum_{j=1}^{T}\boldsymbol{q}_j\boldsymbol{q}_j'\ast{\boldsymbol{\Omega}^{-1}_{j,j}}\\
\end{bmatrix}
=\begin{bmatrix}
\bm{A}&\bm{B}\\\bm{C}&\bm{D}
\end{bmatrix}
\end{split}
\]
where $\boldsymbol{e}_{j}$ is the $T-1$ dimensional $j$-th elementary vector and $\boldsymbol{\widetilde{V}}_1$ is a $N \times (T-1)$ matrix whose $j$-th column is ${\boldsymbol{\Omega}^{-1}_{j,j}}\boldsymbol{v_j}$.\\
Furthermore\footnote{Note that 
\[\boldsymbol{R}_N(\boldsymbol{q}_1\otimes \boldsymbol{\Omega}_{11}^{-1})\boldsymbol{v}_1=\boldsymbol{R}_N(\boldsymbol{q}_1\otimes \boldsymbol{\Omega}_{11}^{-1})(\boldsymbol{I}_1\otimes\boldsymbol{v}_1)\boldsymbol{R}_1'=\boldsymbol{R}_N(\boldsymbol{q}_1\otimes \boldsymbol{\Omega}_{11}^{-1}\boldsymbol{v}_1)\boldsymbol{R}_1'=\boldsymbol{q}_1* \boldsymbol{\Omega}_{11}^{-1}\boldsymbol{v}_1\]
where $\boldsymbol{R}_1=\underset{(1,1)}{\boldsymbol{e}'_1}\otimes \underset{(1,1)}{\boldsymbol{e}'_1}=1$.} 
\begin{equation}
\label{eq:xy11}
\underset{(N+T-1,1)}{\boldsymbol{X}'\boldsymbol{\Omega}^{-1}\boldsymbol{y}}=
\begin{bmatrix}
\underset{(T-1,1)}{\boldsymbol{0}}\\\underset{(N,1)}{\boldsymbol{R}_N(\boldsymbol{q}_1\otimes \bm{I}_N)(\boldsymbol{\Omega}_{11}^{-1}\boldsymbol{v}_1)}\end{bmatrix}
= 
\begin{bmatrix}
\underset{(T-1,1)}{\boldsymbol{0}}\\\underset{(N,1)}{\boldsymbol{q}_1 * \boldsymbol{\Omega}_{11}^{-1}\boldsymbol{v}_1}
\end{bmatrix}\footnote{Use has been made of the following relationship between the Kronecker and the Hadamard product \citep{faliva1996hadamard} \[\underset{(N,M)}{\boldsymbol{A}} * \underset{(N,M)}{\boldsymbol{B}}=\boldsymbol{R}_N(\boldsymbol{A}\otimes\boldsymbol{B})\boldsymbol{R}'_M\] to obtain the right-hand sides of Eq.~\eqref{eq:xy11}.}.
\end{equation}

\noindent Accordingly,
\begin{equation}
\widehat{{\boldsymbol{\delta}}}=\begin{bmatrix}\boldsymbol{I}_{T-1}\,\,\,\underset{(T-1,N)}{\boldsymbol{0}}
\end{bmatrix}\underset{(N+T-1,1)}{\widehat{\boldsymbol{\beta}}}=
\boldsymbol{\Lambda}_{12}(\boldsymbol{q}_1* \boldsymbol{\widetilde{v}}_1)
\label{eq:deltt1}
\end{equation}
where $\boldsymbol{\widetilde{v}}_1=\boldsymbol{\Omega}_{11}^{-1}\boldsymbol{v}_1$ and $\boldsymbol{\Lambda}_{12}$  is the upper off diagonal block of the inverse matrix and 
\begin{equation}\label{eq:invpart}
\underset{(N+T-1,N+T-1)}{(\boldsymbol{X}'\boldsymbol{\Omega}^{-1}\boldsymbol{X})^{-1}}
=\boldsymbol{\Lambda}= \begin{bmatrix}\underset{(T-1,T-1)}{\boldsymbol{\Lambda}_{11}}&\underset{(T-1,N)}{\boldsymbol{\Lambda}_{12}}\\
\underset{(N,T-1)}{\boldsymbol{\Lambda}_{21}} & \underset{(N,N)}{\boldsymbol{\Lambda}_{22}}
\end{bmatrix}.
\end{equation}
Partitioned inversion \citep[see][]{faliva2008dynamic} leads to 
\[\begin{split}\boldsymbol{\Lambda}_{12}=&
-\left(\bm{A}-\bm{B}\bm{D}^{-1}\bm{C}\right)^{-1}\bm{B}\bm{D}^{-1}
\\=&  
\left\{\boldsymbol{I}_{T-1}\ast\boldsymbol{\widetilde{V}}_1'\boldsymbol{V}_1-(\boldsymbol{Q}'_1\ast\boldsymbol{\widetilde{V}}_1')\left[\sum_{j=1}^{T}\boldsymbol{q}_j\boldsymbol{q}_j'\ast{\boldsymbol{\Omega}^{-1}_{j,j}}\right]^{-1}(\boldsymbol{Q}_1\ast\boldsymbol{\widetilde{V}}_1)\right\}^{-1}\cdot\\&(\boldsymbol{Q}'_1\ast\boldsymbol{\widetilde{V}}_1')\left[\sum_{j=1}^{T}\boldsymbol{q}_j\boldsymbol{q}_j'\ast{\boldsymbol{\Omega}^{-1}_{j,j}}\right]^{-1}.
\end{split}\]
This, together with Eq.~\eqref{eq:deltt1}, yields Eq.~\eqref{eq:deltagg12}.
\end{proof}

\subsection{Proof of Corollary~\ref{cor:glsstat1}}\label{app:proofs2g}
\begin{proof}
The proof follows the same lines of Theorem~\ref{th:2gls}. Let us assume 
\[\underset{(NT,NT)}{\boldsymbol{\Omega}}
=\mathds{E}(\boldsymbol{\mu}\boldsymbol{\mu}')= \begin{bmatrix}\underset{(N,N)}{\widetilde{\boldsymbol{\Omega}}}&\underset{(N,N(T-1))}{\boldsymbol{0}}\\
\underset{(N(T-1),N)}{\boldsymbol{0}} & \underset{(N(T-1),N(T-1))}{\boldsymbol{I}_{T-1}\otimes\widetilde{\boldsymbol{\Omega}}}
\end{bmatrix}.
\]
Then, some computations prove that 
\[\begin{split}
\underset{(N+T-1,N+T-1)}{\boldsymbol{X}'\boldsymbol{\Omega}^{-1}\boldsymbol{X}}&=
\begin{bmatrix} \boldsymbol{R}_{T-1}(\boldsymbol{I}_{T-1}\otimes\boldsymbol{V}_1'\widetilde{\boldsymbol{\Omega}}^{-1}\boldsymbol{V}_1) \boldsymbol{R}_{T-1}' & -\boldsymbol{R}_{T-1}(\boldsymbol{Q}_1'\otimes (\boldsymbol{V}_1'\widetilde{\boldsymbol{\Omega}}^{-1}))\boldsymbol{R}'_{N}
\\ -\boldsymbol{R}_{N}(\boldsymbol{Q}_1\otimes (\widetilde{\boldsymbol{\Omega}}^{-1}\boldsymbol{V}_1))\boldsymbol{R}_{T-1}' & \boldsymbol{R}_{N} ((\boldsymbol{q}_1\boldsymbol{q}_1'+\boldsymbol{Q}_1\boldsymbol{Q}_1')\otimes \widetilde{\boldsymbol{\Omega}}^{-1})\boldsymbol{R}_N'
\end{bmatrix}
\\
&=\begin{bmatrix}\boldsymbol{I}_{T-1}*\boldsymbol{V}_1'\widetilde{\boldsymbol{\Omega}}^{-1}\boldsymbol{V}_1 & -\boldsymbol{Q}_1'*\boldsymbol{V}_1'\widetilde{\boldsymbol{\Omega}}^{-1}\\
-\boldsymbol{Q}_1*\widetilde{\boldsymbol{\Omega}}^{-1}\boldsymbol{V}_1 &  (\boldsymbol{q}_1\boldsymbol{q}_1'+\boldsymbol{Q}_1\boldsymbol{Q}_1')* \widetilde{\boldsymbol{\Omega}}^{-1}
\end{bmatrix}
\end{split}
\]
and
\[
\underset{(N+T-1,1)}{\boldsymbol{X}'\boldsymbol{\Omega}^{-1}\boldsymbol{y}}=
\begin{bmatrix}
\underset{(T-1,1)}{\boldsymbol{0}}\\\underset{(N,1)}{\boldsymbol{R}_N(\boldsymbol{q}_1\otimes \widetilde{\boldsymbol{\Omega}}^{-1})\boldsymbol{R}_1\boldsymbol{v}_1}\end{bmatrix}
= 
\begin{bmatrix}
\underset{(T-1,1)}{\boldsymbol{0}}\\\underset{(N,1)}{\boldsymbol{q}_1 * \widetilde{\boldsymbol{\Omega}}^{-1}\boldsymbol{v}_1}
\end{bmatrix}.
\]
Accordingly,
\begin{equation}
\widehat{{\boldsymbol{\delta}}}_{GLS}=\begin{bmatrix}\boldsymbol{I}_{T-1}\,\,\,\underset{(T-1,N)}{\boldsymbol{0}}
\end{bmatrix}\underset{(N+T-1,1)}{\widehat{\boldsymbol{\beta}}_{GLS}}=
\boldsymbol{\Lambda}_{12}(\boldsymbol{q}_1* \boldsymbol{\widetilde{\Omega}}^{-1}\boldsymbol{v}_1)
\label{eq:deltt}
\end{equation}
where $\boldsymbol{\Lambda}_{12}$  is 
\[\begin{split}\boldsymbol{\Lambda}_{12}=&\left\{
(\boldsymbol{I}_{T-1}*\boldsymbol{V}_1'\widetilde{\boldsymbol{\Omega}}^{-1}\boldsymbol{V}_1)-(\boldsymbol{Q}_1'*\boldsymbol{V}_1'\widetilde{\boldsymbol{\Omega}}^{-1})\left[
(\boldsymbol{q}_1\boldsymbol{q}_1'+\boldsymbol{Q}_1\boldsymbol{Q}_1')*\widetilde{\boldsymbol{\Omega}}^{-1}
\right]^{-1}(\boldsymbol{Q}_1*\widetilde{\boldsymbol{\Omega}}^{-1}\boldsymbol{V}_1)
\right\}^{-1}\cdot
\\
&(\boldsymbol{Q}_1'*\boldsymbol{V}_1'\widetilde{\boldsymbol{\Omega}}^{-1})\left[
(\boldsymbol{q}_1\boldsymbol{q}_1'+\boldsymbol{Q}_1\boldsymbol{Q}_1')*\widetilde{\boldsymbol{\Omega}}^{-1}
\right]^{-1}
\end{split}\]
and this yields Eq.~\eqref{eq:cc1}.
\end{proof}

\subsection{Proof of Corollary~\ref{cor:1g}}\label{app:proofs3g}
\begin{proof} When $T=2$, $\boldsymbol{Q}_1=\underset{(N,1)}{\boldsymbol{q}_2}$ and $\boldsymbol{V}_1=\underset{(N,1)}{\boldsymbol{v}_2}$. Then, by assuming stationary disturbances, i.e $\boldsymbol{\Omega}_{j,j}=\boldsymbol{\widetilde{\Omega}}$ for $j=1,2$, the following holds
\begin{equation}
\label{eq:xp}
\begin{split}
\underset{(N+1,N+1)}{\boldsymbol{X}' \boldsymbol{\Omega}^{-1}\boldsymbol{X}}
&=\begin{bmatrix}\boldsymbol{v}_2'\widetilde{\boldsymbol{\Omega}}^{-1} \boldsymbol{v}_2 & -\boldsymbol{q}_2'*\boldsymbol{v}_2'\widetilde{\boldsymbol{\Omega}}^{-1}\\
-\boldsymbol{q}_2*\widetilde{\boldsymbol{\Omega}}^{-1}\boldsymbol{v}_2 &  (\boldsymbol{q}_1\boldsymbol{q}_1'+\boldsymbol{q}_2\boldsymbol{q}_2')* \widetilde{\boldsymbol{\Omega}}^{-1}
\end{bmatrix}
\end{split}
\end{equation}
and
\[
\underset{(N+1,1)}{\boldsymbol{X}'\boldsymbol{\Omega}^{-1}\boldsymbol{y}}=
\begin{bmatrix}
\underset{(1,1)}{\boldsymbol{0}}\\\underset{(N,1)}{\boldsymbol{q}_1 *\widetilde{\boldsymbol{\Omega}}^{-1}\boldsymbol{v}_1}
\end{bmatrix}.
\]
Accordingly, the GLS estimate of the deflator $\delta_{GLS}$ turns out to be
\[
\begin{split}
\widehat{{\delta}}_{GLS-S}=&\left\{(
\boldsymbol{v}_2'\widetilde{\boldsymbol{\Omega}}^{-1}\boldsymbol{v}_2)-(\boldsymbol{q}_2'*\boldsymbol{v}_2'\boldsymbol{\widetilde{\Omega}}^{-1})\left[
(\boldsymbol{q}_1\boldsymbol{q}_1'+\boldsymbol{q}_2\boldsymbol{q}_2')*\boldsymbol{\widetilde{\Omega}}^{-1}
\right]^{-1}(\boldsymbol{q}_2*\boldsymbol{\widetilde{\Omega}}^{-1}\boldsymbol{v}_2)
\right\}^{-1}\cdot\\
&\cdot(\boldsymbol{q}_2'*\boldsymbol{v}_2'\boldsymbol{\widetilde{\Omega}}^{-1})\left[
(\boldsymbol{q}_1\boldsymbol{q}_1'+\boldsymbol{q}_2\boldsymbol{q}_2')*\boldsymbol{\widetilde{\Omega}}^{-1}
\right]^{-1}(\boldsymbol{q}_1* \boldsymbol{\widetilde{\Omega}}^{-1}\boldsymbol{v}_1).
\end{split}
\]
\noindent Then, by denoting with $q_{it}$ and $v_{it}$ the quantity and the value of the $i$-th good at time $t$ and by assuming diagonal the matrix $\boldsymbol{\widetilde{\Omega}}$, with diagonal entries $\vartheta_i$, some computations yield
\begin{equation}\label{eq:tf}
\widehat{{\delta}}_{GLS}=\left\{\sum_{i=1}^N\frac{q^2_{i1}v^2_{i2}\frac{1}{\vartheta_i}}{q^2_{i1}+q^2_{i2}}
\right\}^{-1}
\left\{\sum_{i=1}^N\frac{q_{i2}v_{i2}q_{i1}v_{i1}\frac{1}{\vartheta_i}}{q^2_{i1}+q^2_{i2}}\right\}.
\end{equation}
The reciprocal of Eq.~\eqref{eq:tf} yields the GLS estimate of the index.
\end{proof}

\subsection{Proof of Corollary~\ref{cor:51g}}\label{app:proofs1}
\begin{proof}
The optimization problem is
\[
 \min_{\lambda}\,\,\, \boldsymbol{e}'\boldsymbol{A}\boldsymbol{e}.
 \]
 \noindent Thus, the first order conditions for a minimum are
 \[\begin{split}\frac{\partial  \boldsymbol{e}'\boldsymbol{A}\boldsymbol{e}}{\partial \lambda}=0\end{split}\]
which lead to the solution
 \begin{equation}\label{eq:000}
\lambda=\frac{\boldsymbol{p}_2'\boldsymbol{A}\boldsymbol{p}_1}{\boldsymbol{p}_1'\boldsymbol{A}\boldsymbol{p}_1}.
\end{equation}
Setting $\boldsymbol{A}=(\boldsymbol{\overline{\pi}}\,\boldsymbol{\overline{\pi}}')$, with $\boldsymbol{\overline{\pi}}$ as defined in Corollary~\ref{cor:1g}, in Eq.~\eqref{eq:000} yields the MPL index.
\end{proof} 

%
%
%

\subsection{Proof of Corollary~\ref{cor:3}}\label{app:proofs6} 
\begin{proof}
The variance-covariance matrix of the vector $\boldsymbol{\hat {\delta}}_{GLS}$, as defined in Eq.~\eqref{eq:deltt}, turns out to be 
\[V(\boldsymbol {\widehat{\delta}}_{GLS})=\sigma^2 [\boldsymbol{I}_{T-1}, \boldsymbol{0}_{T-1,N}][\boldsymbol{X}'\boldsymbol{\Omega}^{-1}\boldsymbol{X}]^{-1}\begin{bmatrix}
\boldsymbol {I}_{T-1}\\
\boldsymbol{0}_{N,T-1}
\end{bmatrix}
=\boldsymbol{\Lambda}_{11}
\]
where $\boldsymbol{\Lambda}_{11}$ is specified as in Eq.~\eqref{eq:invpart}. Thus, partitioned inversion rules lead to  
\[
\boldsymbol{\Lambda}_{11}=
-\left(\bm{A}-\bm{B}\bm{D}^{-1}\bm{C}\right)^{-1}=
\left\{\boldsymbol{I}_{T-1}\ast\boldsymbol{\widetilde{V}}_1'\boldsymbol{V}_1-(\boldsymbol{Q}'_1\ast\boldsymbol{\widetilde{V}}_1')\left[\sum_{j=1}^{T}\boldsymbol{q}_j\boldsymbol{q}_j'\ast{\boldsymbol{\Omega}^{-1}_{j,j}}\right]^{-1}(\boldsymbol{Q}_1\ast\boldsymbol{\widetilde{V}}_1)\right\}^{-1}.
\]
\end{proof}

\subsection{Proof of Theorem~\ref{th:4}}\label{app:proofs7} 
\begin{proof}
When the values, $\boldsymbol{v}_{T+1}$, and the quantities, $\boldsymbol{q}_{T+1}$, of  $N$ commodities in a reference basket become available for the $(T+1)$-th additional country, the reference equation system for updating the MPL index becomes
\[
\begin{cases}
&\boldsymbol{V}_{}\, \, \boldsymbol{D}_{\boldsymbol{\delta}}=  \boldsymbol{D}_{\tilde{\boldsymbol{p}}}\, \, \boldsymbol{Q}_{} + \boldsymbol{E}_{}
\\& \boldsymbol{v}_{T+1}{\delta}_{T+1}=\boldsymbol{D}_{\tilde{\boldsymbol{p}}}\boldsymbol{q}_{T+1}+\boldsymbol{\varepsilon}_{T+1}
\end{cases}
\rightarrow
\begin{cases}
& \boldsymbol{v}_1=\boldsymbol{D}_{\tilde{\boldsymbol{p}}}\boldsymbol{q}_1+\boldsymbol{\varepsilon}_1
\\& \boldsymbol{V}_1 \tilde{\boldsymbol{D}}_{\boldsymbol{\delta}}=\boldsymbol{D}_{\tilde{\boldsymbol{p}}}\boldsymbol{Q}_1+\boldsymbol{E}_1
\\& \boldsymbol{v}_{T+1}{\delta}_{T+1}=\boldsymbol{D}_{\tilde{\boldsymbol{p}}}\boldsymbol{q}_{T+1}+\boldsymbol{\varepsilon}_{T+1}
\end{cases}.
\]
After some computations, the above system can be also written as
\[
\begin{cases}
&  \boldsymbol{v}_1=(\boldsymbol{q}_1' \otimes \boldsymbol{I}_N) \boldsymbol{R}_N' \tilde{\boldsymbol{p}}+\boldsymbol{\varepsilon}_1
\\
& \underset{(N(T-1),1)}{\boldsymbol{0}}=(\boldsymbol{I}_{T-1}\otimes (\boldsymbol{-V}_1))\boldsymbol{R}_{T-1}'\boldsymbol{\delta}+(\boldsymbol{Q}_1'\otimes \boldsymbol{I}_N)\boldsymbol{R}_N'\tilde{\boldsymbol{p}}+\boldsymbol{\eta}\\
& \boldsymbol{0}=-\boldsymbol{v}_{T+1}\delta_{T+1}+ (\boldsymbol{q}_{T+1}' \otimes \boldsymbol{I}_N) \boldsymbol{R}_N' \tilde{\boldsymbol{p}}+\boldsymbol{\varepsilon}_{T+1}
\end{cases}
\]
or, in compact form, as
\[
\underset{(NT+N, 1)}{\boldsymbol{y}_u}=\underset{(NT+N, N+T)}{\boldsymbol{X}_u}\underset{(N+T, 1)}{\boldsymbol{\beta}_u} + \underset{(NT+N, 1)}{\boldsymbol{\mu}_u}
\] where 
\[
\underset{(NT+N,1)}{\boldsymbol{y}_u}=\begin{bmatrix} \underset{(N,1)}{\boldsymbol{v}_1}\\\underset{(N(T-1),1)}{\boldsymbol{0}}\\\underset{(N,1)}{\boldsymbol{0}}
\end{bmatrix},
\underset{(N+T,1)}{\boldsymbol{\beta}_u}=\begin{bmatrix} \underset{(T-1,1)}{{\boldsymbol{\delta}}}\\\underset{(1,1)}{{\delta_{T+1}}}\\\underset{(N,1)}{\tilde{\boldsymbol{p}}}
\end{bmatrix}, \,\,\underset{(NT+N,1)}{\boldsymbol{\mu}_u}=\begin{bmatrix} \underset{(N,1)}{\boldsymbol{\varepsilon}_1} \\ \underset{(N(T-1),1)}{\boldsymbol{\eta}}\\ \underset{(N,1)}{\boldsymbol{\varepsilon}_{T+1}}
\end{bmatrix}
\]
\[
\underset{(NT+N,N+T)}{\boldsymbol{X}_u}
=\begin{bmatrix} 
\underset{(N,T-1)}{\boldsymbol{0}} &  \underset{(N,1)}{\boldsymbol{0}}&\underset{(N,N)}{(\boldsymbol{q}_1'\otimes \boldsymbol{I}_N)\boldsymbol{R}_N'}\\ 
\underset{(N(T-1),T-1)}{(\boldsymbol{I}_{T}\otimes - \boldsymbol{V}_1)\boldsymbol{R}_{T-1}'} & \underset{(N(T-1),1)}{\boldsymbol{0}}&\underset{(N(T-1),N)}{(\boldsymbol{Q}_1'\otimes \boldsymbol{I}_N)\boldsymbol{R}_N'}\\ 
\underset{(N,T-1)}{\boldsymbol{0}}&\underset{(N,1)}{-\boldsymbol{v}_{T+1}}  & \underset{(N,N)}{(\boldsymbol{q}_{T+1}'\otimes \boldsymbol{I}_N)\boldsymbol{R}_N'}
\end{bmatrix}\]
with
\[
E(\boldsymbol{\mu}_u \boldsymbol{\mu}_u')=\underset{(N(T+1),N(T+1))}{\boldsymbol{\Omega}_u}
=\begin{bmatrix} \underset{(N,N)}{{\boldsymbol{\Omega_{1,1}}}} & \underset{(N,N(T-1))}{\boldsymbol{0}} &\underset{(N,N)}{\boldsymbol{0}} \\ 
\underset{(N(T-1),N(T-1))}{\boldsymbol{0}} &\underset{(N(T-1),N(T-1))}{\boldsymbol{\Omega}^{*}} &\underset{(N(T-1),N)}{\boldsymbol{0}} \\
\underset{(N,N)}{\boldsymbol{0}} &\underset{(N,N(T-1))}{\boldsymbol{0}} &\underset{(N,N)}{\boldsymbol{\Omega_{T+1,T+1}}}
\end{bmatrix}.
\]
Here $\boldsymbol{\Omega}_{1,1}$ and $\boldsymbol{\Omega}^{*}$ are as defined in Eq.~\eqref{eq:11} and $\boldsymbol{\Omega}_{T+1,T+1}=E(\boldsymbol{\epsilon}_{T+1}\boldsymbol{\epsilon}_{T+1}')$.

Following the same argument of Theorem~\ref{th:2gls}, we obtain that
\[
\underset{(N+T,N+T)}{\boldsymbol{X}_u' \boldsymbol{\Omega^{-1}}_u\boldsymbol{X}_u}
=\begin{bmatrix} 
\boldsymbol{I}_{T-1}*\boldsymbol{\widetilde{V}}_1\boldsymbol{V}_1' & \boldsymbol{0} &-\boldsymbol{Q}_1'*\boldsymbol{\widetilde{V}}_1'
\\ 
\boldsymbol{0}&\boldsymbol{\widetilde{v}}_{T+1}'\boldsymbol{v}_{T+1}&-\boldsymbol{\widetilde{v}}_{T+1}'*\boldsymbol{q}_{T+1}'\\
-\boldsymbol{Q}_1*\boldsymbol{\widetilde{V}}_1&-\boldsymbol{\widetilde{v}}_{T+1}*\boldsymbol{q}_{T+1}&(\sum_{j=1}^{T+1}{\boldsymbol{q}}_j{\boldsymbol{q}}_j'*\boldsymbol{\Omega^{-1}}_{j,j})
\end{bmatrix}
\]
and
\[
\underset{(N+T,1)}{\boldsymbol{X}_u'\boldsymbol{\Omega^{-1}_{u}}\boldsymbol{y}_u}
=\begin{bmatrix} 
\underset{(T,1)}{\boldsymbol{0}}\\\underset{(N,1)}{\boldsymbol{q}_1*\boldsymbol{\widetilde{v}}_1}
\end{bmatrix}
\]
where $\boldsymbol{\widetilde{V}}_1$ and $\boldsymbol{\widetilde{v}}_1$ are defined as in Theorem~\ref{th:2gls} and $\boldsymbol{\widetilde{v}}_{T+1}=\boldsymbol{\Omega^{-1}}_{T+1,T+1} \boldsymbol{v}_{T+1}$.

Then, upon nothing that 
\[
\begin{bmatrix}
\underset{(T-1,1)}{\widehat{\boldsymbol{\delta}}}
\\\underset{(1,1)}{\widehat{\delta}_{T+1}}
\end{bmatrix}=
\begin{bmatrix}
\boldsymbol{I}_T&\underset{(T,N)}{\boldsymbol{0}}
\end{bmatrix} \underset{(N+T,1)}{\boldsymbol{\widehat{\beta}_u}} = \boldsymbol{\Lambda_{12}}\left(\boldsymbol{q}_1\ast\boldsymbol{\widetilde{v}}_1\right)
\]
where $\boldsymbol{\Lambda}_{12}$ is the $T \times N$ upper off diagonal block of the inverse matrix $(\boldsymbol{X}_u'\boldsymbol{\Omega^{-1}}_u\boldsymbol{X}_u)^{-1}$,
partitioned inversion leads to 
\[
\begin{split}
\underset{(T,N)}{\boldsymbol{\Lambda}_{12}}=&\\
=&\left\{
\begin{bmatrix} 
\boldsymbol{I}_{T-1}\ast\boldsymbol{\widetilde{V}}'_1\boldsymbol{V}_1 & \boldsymbol{0}  \\ 
\boldsymbol{0} & \boldsymbol{\widetilde{v}}'_{T+1}\boldsymbol{v}_{T+1}  \end{bmatrix}-
\begin{bmatrix} 
\boldsymbol{Q}'_1\ast\boldsymbol{\widetilde{V}}'_1  \\ 
\boldsymbol{\widetilde{v}}'_{T+1}\ast\boldsymbol{q}'_{T+1}  \end{bmatrix}
\begin{bmatrix} 
\left(\sum_{j=1}^{T+1}\boldsymbol{q}_j\boldsymbol{q}'_j \ast \boldsymbol{\Omega^{-1}}_{j,j}\right)\end{bmatrix}^{-1} \right.
\\&
\left.\begin{bmatrix} 
\boldsymbol{Q}_1\ast\boldsymbol{\widetilde{V}}_1  & 
\boldsymbol{\widetilde{v}}_{T+1}\ast\boldsymbol{q}_{T+1}  \end{bmatrix}
\right\}^{-1}
\begin{bmatrix} 
\boldsymbol{Q}'_1\ast\boldsymbol{\widetilde{V}}'_1  \\ 
\boldsymbol{\widetilde{v}}'_{T+1}\ast\boldsymbol{q}'_{T+1}  \end{bmatrix}
\begin{bmatrix} 
\left(\sum_{j=1}^{T+1}\boldsymbol{q}_j\boldsymbol{q}'_j \ast \boldsymbol{\Omega^{-1}}_{j,j}\right)\end{bmatrix}^{-1}.
\end{split}
\]
Then, pre-multiplying  $(\boldsymbol{q}_1\ast\boldsymbol{\widetilde{v}}_1)$ by $\boldsymbol{\Lambda}_{12}$ yields the estimator $\begin{bmatrix}
\underset{(T-1,1)}{\widehat{\boldsymbol{\delta}}}
\\\underset{(1,1)}{\widehat{\delta}_{T+1}}
\end{bmatrix}$. The reciprocal of the (non-null) elements of this estimator provides the values of the updated multilateral version of the MPL index. 
\end{proof}

\subsection{Proof of Theorem~\ref{th:5}}\label{app:proofs8} 
\begin{proof}
When the values, $\boldsymbol{v}_{T+1}$, and the quantities, $\boldsymbol{q}_{T+1}$, of $N$ commodities of a reference basket become available at time $T+1$, the updating of the multi-period version of the MPL index must not change its past values with meaningful computational advantages. In order to get the required updating formula, let us rewrite Eq.~\eqref{eq:1} as follows
\begin{equation}\label{eq:1proof8}
\left[\underset{(N, T)}{\boldsymbol{V}_{}}\, \, \underset{(N, 1)}{\boldsymbol{v}_{T+1}}\right]
\underset{(T,T)}{\boldsymbol{D}^*_{{{\boldsymbol{\delta}}}}} 
=  \underset{(N, N)}{\boldsymbol{D}_{\tilde{\boldsymbol{p}}}}\, \, \left[\underset{(N, T)}{\boldsymbol{Q}}\,\,\underset{(N,1)}{\boldsymbol{q}_{T+1}}\right] + \left[\underset{(N, T)}{\boldsymbol{E}_{}}\,\,\underset{(N,1)}{\varepsilon_{T+1}}
\right]
\end{equation}
where ${\boldsymbol{D}^*_{{{\boldsymbol{\delta}}}}}$ is specified as follows
\[
\underset{(T+1,T+1)}{\boldsymbol{D}^{*}_{{{\boldsymbol{\delta}}}}}=
\begin{bmatrix}
\underset{(T,T)}{\widehat{\boldsymbol{D}}_{{{\boldsymbol{\delta}}}}}& \underset{(T,1)}{\boldsymbol{0}}
\\
\underset{(1,T)}{\boldsymbol{0}'}&\underset{(1,1)}{{\delta}_{T+1}}
\end{bmatrix}.
\]
Here ${\widehat{\boldsymbol{D}}_{{{\boldsymbol{\delta}}}}}$ denotes the estimate of ${\boldsymbol{D}}_{{{\boldsymbol{\delta}}}}$, defined as in Eq.~\eqref{eq:1b}, that is 
\[
\underset{(T,T)}{\widehat{\boldsymbol{D}}_{{{\boldsymbol{\delta}}}}}=
\begin{bmatrix}
\underset{(1,1)}{1}&\underset{(1,T-1)}{\boldsymbol{0}'}\\
\underset{(T-1,1)}{\boldsymbol{0}}&
\begin{bmatrix}
\widehat{{\delta}}_2&0&\cdots&0\\
0&\widehat{{\delta}}_3&\cdots&0\\
0&0&\cdots&0\\
0&0&\cdots&\widehat{{\delta}}_T\\
\end{bmatrix}
\end{bmatrix}
=
\begin{bmatrix}
1&\boldsymbol{0}'\\
0&\widehat{\tilde{\boldsymbol{D}_{\boldsymbol{\delta}}}}
\end{bmatrix}
\]
where the entries $\widehat{{\delta}}_2$, $\widehat{{\delta}}_3, \dots, {\widehat{\delta}}_T$ are the elements of the vector $\widehat{{\boldsymbol{\delta}}}$ given in Eq.~\eqref{eq:deltagg12}. 
The system in Eq.~\eqref{eq:1proof8} can be also written as
\begin{equation}
\label{eq:2proof81}
\begin{cases}
& \boldsymbol{v}_1=\boldsymbol{D}_{\tilde{\boldsymbol{p}}}\boldsymbol{q}_1+\boldsymbol{\varepsilon}_1
\\& \boldsymbol{V}_1 \widehat{\tilde{\boldsymbol{D}}}_{\boldsymbol{\delta}}=\boldsymbol{D}_{\tilde{\boldsymbol{p}}}\boldsymbol{Q}_1+\boldsymbol{E}_1
\\& \boldsymbol{v}_{T+1}{\delta}_{T+1}=\boldsymbol{D}_{\tilde{\boldsymbol{p}}}\boldsymbol{q}_{T+1}+\boldsymbol{\varepsilon}_{T+1}
\end{cases}
\rightarrow\, 
\begin{cases} \boldsymbol{V}\,{\widehat{\boldsymbol{D}}}_{\boldsymbol{\delta}} = \boldsymbol{D}_{\tilde{\boldsymbol{p}}} \boldsymbol{Q}+\boldsymbol{E}
\\
\boldsymbol{v}_{T+1}{\delta}_{T+1}=\boldsymbol{D}_{\tilde{\boldsymbol{p}}}\boldsymbol{q}_{T+1}+\varepsilon_{T+1}
\end{cases}.
\end{equation}
The application of the $vec$ operator to the left hand-side block of equations in Eq.~\eqref{eq:2proof81} yields
\begin{equation}
\label{eq:gtt}
\begin{cases}
\boldsymbol{\overline{v}}_1=\left(\boldsymbol{Q}'\otimes \boldsymbol{I}_N\right)\boldsymbol{R}_N' \tilde{\boldsymbol{p}}+\boldsymbol{\eta}\\
\boldsymbol{0}=-\boldsymbol{v}_{T+1}{\delta}_{T+1}+\left(\boldsymbol{q}_{T+1}'\otimes \boldsymbol{I}_N\right)\boldsymbol{R}_N' \tilde{\boldsymbol{p}}+\varepsilon_{T+1}
\end{cases}
\end{equation}
where $\boldsymbol{\overline{v}}_1=vec(\boldsymbol{V}\,{\widehat{\boldsymbol{D}}}_{\boldsymbol{\delta}})=\left( \boldsymbol{I}_T \otimes \boldsymbol{V}\right)\boldsymbol{R}_T' \tilde{{\boldsymbol{\delta}}}$ with $\tilde{{\boldsymbol{\delta}}}'=[1,\widehat{{\boldsymbol{\delta}}}]'$ \footnote{Note that, differently from the proof of Theorem~\ref{th:4}, the vector $\boldsymbol{\delta}$ does not enter in the updating estimation process as it is considered given.} and $\boldsymbol{\eta}$ is equal to $vec(\boldsymbol{E})$. 

It is worth noting that  Eq.~\eqref{eq:gtt} can also be written in vector form as
\[
\underset{(N(T+1), 1)}{\boldsymbol{y}_u}=\underset{(N(T+1), N+1)}{\boldsymbol{X}_{u}}\underset{(N+1, 1)}{\boldsymbol{\beta}_u} + \underset{(N(T+1), 1)}{\boldsymbol{\mu}_u}
\]
where 
\[
\underset{(N(T+1),1)}{\boldsymbol{y}_u}=\begin{bmatrix} \boldsymbol{\overline{v}}_1\\\boldsymbol{0}
\end{bmatrix}, \,\,\underset{(N(T+1),N+1)}{\boldsymbol{X}_{u}}=\begin{bmatrix} \boldsymbol{0} &  (\boldsymbol{Q}'\otimes \boldsymbol{I}_N)\boldsymbol{R}_N'\\ -\boldsymbol{v}_{T+1} & (\boldsymbol{q}_{T+1}'\otimes \boldsymbol{I}_N)\boldsymbol{R}_N'
\end{bmatrix}
\]
and
\[
\underset{(N+1,1)}{\boldsymbol{\beta}_u}=\begin{bmatrix} {\delta}_{T+1}\\\tilde{\boldsymbol{p}}
\end{bmatrix}, \,\,\underset{(N(T+1),1)}{\boldsymbol{\mu}_u}=\begin{bmatrix} \boldsymbol{\eta} \\ \boldsymbol{\varepsilon}_{T+1}
\end{bmatrix}
\]
with
\[
E(\boldsymbol{\mu}_u\boldsymbol{\mu}_u')=\boldsymbol{\Omega}_u=
\begin{bmatrix}\boldsymbol{\Omega}, \  \ \ \ \boldsymbol{0} \\
\boldsymbol{0},\ \boldsymbol{\Omega_{T+1,T+1}}
\end{bmatrix}
\]
where $\boldsymbol{\Omega}$ is as defined in Eq.~\eqref{eq:11} and $\boldsymbol{\Omega}_{T+1,T+1}=E(\boldsymbol{\epsilon}_{T+1}\boldsymbol{\epsilon}_{T+1}')$.

The GLS estimator of the vector $\boldsymbol{\beta}_u$ is given by
\[
\widehat{\boldsymbol{\beta}_u}=(\boldsymbol{X}_u'\boldsymbol{\Omega}^{-1}_{u}\boldsymbol{X}_u)^{-1}\boldsymbol{X}_u'\boldsymbol{\Omega}^{-1}_{u}\boldsymbol{y}_u
\]
where
\[
\begin{split}
\underset{(N+1,N+1)}{\boldsymbol{X}_u'\boldsymbol{\Omega}^{-1}_{u} \boldsymbol{X}_u}&=
\begin{bmatrix} 
\boldsymbol{\widetilde{v}}_{T+1}'\boldsymbol{v}_{T+1}& -\boldsymbol{\widetilde{v}}_{T+1}'\left(\boldsymbol{q}_{T+1}'\otimes \boldsymbol{I}_N\right)\boldsymbol{R}'_N
\\ -\boldsymbol{R}_{N}(\boldsymbol{q}_{T+1}\otimes (\boldsymbol{I}_N))\boldsymbol{\widetilde{v}}_{T+1} & \boldsymbol{R}_{N} ((\boldsymbol{Q}\otimes \boldsymbol{I}_{N})\boldsymbol{\Omega} (\boldsymbol{Q}'\otimes \boldsymbol{I}_{N}))\boldsymbol{R}_N'+\boldsymbol{R}_{N}(\boldsymbol{q}_{T+1}\boldsymbol{q}_{T+1}'\otimes \boldsymbol{I}_N)\boldsymbol{R}_N'
\end{bmatrix}
\\
&=\begin{bmatrix}\boldsymbol{\widetilde{v}}_{T+1}'\boldsymbol{v}_{T+1} & -\boldsymbol{\widetilde{v}}_{T+1}'\ast \boldsymbol{q}_{T+1}'\\
-\boldsymbol{q}_{T+1}\ast\boldsymbol{\widetilde{v}}_{T+1} &  \sum_{j=1}^{T+1}\boldsymbol{q}_{j}\boldsymbol{q}_{j}'\ast \boldsymbol{\Omega}_{j,j}
\end{bmatrix}
\end{split}
\]
and
\[\underset{(N+1,1)}{\boldsymbol{X}_u'\boldsymbol{\Omega}^{-1}_{u}\boldsymbol{y}_u}=
\begin{bmatrix}
\underset{(1,1)}{0}\\\underset{(N,1)}{\boldsymbol{R}_N(\boldsymbol{Q}\otimes \boldsymbol{I}_N)\boldsymbol{\overline{v}}_1}
\end{bmatrix}
= 
\begin{bmatrix}
\underset{(1,1)}{0}\\
\underset{(N,1)}{\left(\boldsymbol{Q}\ast \boldsymbol{\widetilde{V}}\right)\tilde{{\boldsymbol{\delta}}}}
\end{bmatrix}
\]
where $\boldsymbol{\widetilde{V}}$ is a $N \times T$ matrix whose $j$-th column is $\boldsymbol{\Omega}_{j,j}\boldsymbol{v}_{j}$.

Now, upon nothing that 
\[
\underset{(1,1)}{\widehat{\delta}_{T+1}}=
\begin{bmatrix}
1&\underset{(1,N)}{\boldsymbol{0}'}
\end{bmatrix} \underset{(N+1,1)}{\boldsymbol{\widehat{\beta}_u}} = \boldsymbol{\Lambda_{12}}{\left(\boldsymbol{Q}\ast \boldsymbol{V}\right)\tilde{{\boldsymbol{\delta}}}},
\]
where $\boldsymbol{\Lambda}_{12}$ is the $1 \times N$ upper off diagonal block of the inverse matrix $(\boldsymbol{X}_u' \boldsymbol{\Omega}^{-1}_{u}\boldsymbol{X}_u)^{-1}$, partitioned inversion leads to 
\[
\begin{split}
\boldsymbol{\Lambda}_{12}=&\left\{
\boldsymbol{\widetilde{v}}_{T+1}'\boldsymbol{v}_{T+1}-\left(\boldsymbol{q}_{T+1}'\ast\boldsymbol{\widetilde{v}}_{T+1}'\right)\left[
\sum_{j=1}^{T+1}\boldsymbol{q}_{j}\boldsymbol{q}_{j}'\ast \boldsymbol{\Omega}_{j,j}
\right]^{-1}\left(\boldsymbol{q}_{T+1}\ast\boldsymbol{\widetilde{v}}_{T+1}\right)
\right\}^{-1}\cdot \\&\left(\boldsymbol{q}_{T+1}'\ast\boldsymbol{\widetilde{v}}_{T+1}'\right)\left[
\sum_{j=1}^{T+1}\boldsymbol{q}_{j}\boldsymbol{q}_{j}'\ast \boldsymbol{\Omega}_{j,j}
\right]^{-1}.
\end{split}
\]
\noindent Then, pre-multiplying ${\left(\boldsymbol{Q}\ast \boldsymbol{\widetilde{V}}\right)\tilde{{\boldsymbol{\delta}}}}$ by $\boldsymbol{\Lambda}_{12}$ yields the estimator of $\widehat{\delta}_{T+1}$ given in Theorem~\ref{th:5}. The reciprocal of this estimator provides the updated value of the multi-period version of the MPL index. 
\end{proof}

\subsection{Proof of the MPL index properties}\label{app:propgmpl}
\begin{proof}
In this Appendix the main properties enjoyed by the MPL index, $\widehat{{\lambda}}_{GLS}$, as defined in Eq.~\eqref{eq:ffrt0g}, are proved. To this end, let $\widehat{{\lambda}}(\boldsymbol{p}_{1},\boldsymbol{p}_{2},\boldsymbol{q}_{1},\boldsymbol{q}_{2})$ denote the MPL price index with $\boldsymbol{p}_{t}$ and $\boldsymbol{q}_{t}$ vectors of $N$ prices and quantities at time $t$.
\begin{enumerate}
\item[P.1] \textit{Strong identity}: 
Simple computation prove that
\[
\widehat{\lambda}(\boldsymbol{p}_{2},\boldsymbol{p}_{2},\boldsymbol{q}_{1},\boldsymbol{q}_{2})=\frac{\boldsymbol{p}_{2}'\overline{\boldsymbol{\pi}}}{\boldsymbol{p}_{2}'\overline{\boldsymbol{\pi}}}=1.
\]
where the weights $\overline{\boldsymbol{\pi}}$ are as defined in \eqref{eq:ffrt0g}
\item[P.2] \textit{Commensurability}: 
Let $\boldsymbol{p}_{j}^{\ast}=\boldsymbol{p}_{j} \ast \boldsymbol{\gamma}$ and $\boldsymbol{q}_{j}^{\ast}=\boldsymbol{q}_{j} \ast \boldsymbol{\gamma}^{-1}$ and denote with $p_{ij}^{\ast}$, $q_{ij}^{\ast}$ and $\gamma_i$ the $i$-th entries of $\boldsymbol{p}_{j}^{\ast}$ and $\boldsymbol{q}_{j}^{\ast}$ and $\boldsymbol{\gamma}$ respectively.

Then, simple computations prove that the following holds for the $i-$th weight, $\overline{\pi}_{i}^{\ast}$, of the MPL index 
\[
\overline{\pi}_{i}^{\ast}=\frac{1}{\vartheta_{i}}p_{i2}^{\ast}\frac{q_{i2}^{\ast 2}q_{i1}^{\ast 2}}{q_{i2}^{\ast 2}+q_{i1}^{\ast 2}}=\frac{1}{\gamma_{i}}\overline{\pi}_{i}.
\]
where $\overline{\pi}_{i}$ is as defined in \eqref{eq:gf}. This means that $\overline{\boldsymbol{\pi}}^{\ast}=\boldsymbol{\gamma}^{-1} \ast \overline{\boldsymbol{\pi}}$ and, accordingly
\[
\widehat{\lambda}(\boldsymbol{\gamma}\ast\boldsymbol{p}_{1},\boldsymbol{\gamma}\ast\boldsymbol{p}_{2},\boldsymbol{\gamma}^{-1}\ast\boldsymbol{q}_{1},\boldsymbol{\gamma}^{-1} {q}_{2})=\frac{\boldsymbol{p}_{2}^{\ast}\overline{\boldsymbol{\pi}}^{\ast}}{\boldsymbol{p}_{1}^{\ast}\overline{\boldsymbol{\pi}}^{\ast}}=\frac{(\boldsymbol{p}_{2} \ast \boldsymbol{\gamma})'(\boldsymbol{\gamma}^{-1} \ast \overline{\boldsymbol {\pi}})}{(\boldsymbol{p}_{1} \ast \boldsymbol{\gamma})'(\boldsymbol{\gamma}^{-1} \ast \overline{\boldsymbol{\pi}})}=\frac{\boldsymbol{p}_{2}'\overline{\boldsymbol{\pi}}}{\boldsymbol{p}_{1}'\overline{\boldsymbol{\pi}}}.
\]
\item[P.3] \textit{Proportionality}: 
Simple computations prove that when $\boldsymbol{p}_{2}^{\ast}=\alpha \boldsymbol{p}_{2}$, the weights of the MPL index become
\begin{equation}\label{eq:jh}
\overline{\boldsymbol {\pi}}^{\ast}=\alpha\overline{\boldsymbol{\pi}}.
\end{equation}
Thus
\[
\widehat{\lambda}(\boldsymbol{p}_{1},\alpha\boldsymbol{p}_{2},\boldsymbol{q}_{1},\boldsymbol{q}_{2})=
\frac{\boldsymbol{p}_{2}^{*'}\overline{\boldsymbol{\pi}}^{*}}{\boldsymbol{p}_{1}'\overline{\boldsymbol{\pi}}*}
=\frac{\alpha^{2}\boldsymbol{p}_{2}'\overline{\boldsymbol{\pi}}}{\alpha\boldsymbol{p}_{1}'\overline{\boldsymbol{\pi}}}=\frac{\alpha\boldsymbol{p}_{2}'\overline{\boldsymbol{\pi}}}{\boldsymbol{p}_{1}'\overline{\boldsymbol{\pi}}}.
\]
\item[P.4] \textit{Dimensionality}: 
When $\boldsymbol{p}_{2}^{\ast}=\alpha \boldsymbol{p}_{2}$ and $\boldsymbol{p}_{1}^{\ast}=\alpha \boldsymbol{p}_{1}$ the weights of the MPL index are as in Eq.~\eqref{eq:jh}. Accordingly
\[
\widehat{\lambda}(\alpha\boldsymbol{p}_{1},\alpha\boldsymbol{p}_{2},\boldsymbol{q}_{1},\boldsymbol{q}_{2})=
\frac{\boldsymbol{p}_{2}^{\ast'}\overline{\boldsymbol{\pi}}^{\ast}}{\boldsymbol{p}_{1}^{\ast'}\overline{\boldsymbol {\pi}}^{\ast}}=
\frac{\alpha^{2}\boldsymbol{p}_{2}'\overline{\boldsymbol{\pi}}}{\alpha^{2}\boldsymbol{p}_{1}'\overline{\boldsymbol{\pi}}}=
\frac{\boldsymbol{p}_{2}'\overline{\boldsymbol{\pi}}}{\boldsymbol{p}_{1}'\overline{\boldsymbol{\pi}}}.
\]

\item[P.5] \textit{Monotonicity}: Let $\boldsymbol{k}$ be a vector with entries greater than 1, that is $\boldsymbol{k}>\boldsymbol{u}$ where $\boldsymbol{u}$ is the unit vector. Let $\boldsymbol{p}_{2}^{\ast}=\boldsymbol{p}_{2} \ast \boldsymbol{k}$ and note that the entries of $\boldsymbol{p}_{2}^{\ast}$ are greater than those 
of $\boldsymbol{p}_{2}$, namely $\boldsymbol{p}_{2}^{\ast} > \boldsymbol{p}_{2}$. Furthermore, when $\boldsymbol{p}_{2}^{\ast}$ are the prices of the second period, the weights of the MPL index become $\overline{\pi}^{\ast}_{i}=\overline{\pi}_{i}\ast k_{i}$ where $k_{i}$ is the $i$-th element of $\boldsymbol{k}$. Accordingly,  $\boldsymbol{\overline{\pi}}^{\ast}=\boldsymbol{\overline{\pi}} \ast \boldsymbol{k}$ and the MPL index becomes
\[
\widehat{{\lambda}}(\boldsymbol{p}_{1},\boldsymbol{k}\ast\boldsymbol{p}_{2},\boldsymbol{q}_{1},\boldsymbol{q}_{2})=
\frac{\boldsymbol{p}_{2}^{\ast'}\boldsymbol{\overline{\pi}}^{\ast}}{\boldsymbol{p}_{1}'\boldsymbol{\overline{\pi}}^{\ast}}>
\frac{\boldsymbol{p}_{2}'\boldsymbol{\overline{\pi}}}{\boldsymbol{p}_{1}'\boldsymbol{\overline{\pi}}}.
\]
Similarly setting ${\boldsymbol{p}}_{1}^{\ast}={\boldsymbol{p}}_{1} \ast \boldsymbol{k}$, with $\boldsymbol{k}$ specified as before, the weights of the MPL index do not change and the index becomes
\[
\widehat{{\lambda}}(\boldsymbol{k}\ast {\boldsymbol{p}}_{1},{\boldsymbol{p}}_{2},{\boldsymbol{q}}_{1},{\boldsymbol{q}}_{2})=\frac{{\boldsymbol{p}}_{1}'{\boldsymbol{\overline{\pi}}}}{{\boldsymbol{p}}_{1}^{\ast'}{\boldsymbol{\overline{\pi}}}}<\frac{{\boldsymbol{p}}_{2}'{\boldsymbol{\overline{\pi}}}}{\boldsymbol{p}_{1}'\boldsymbol{\overline{\pi}}}
\]
as the entries of $\boldsymbol{p}_{1}^{\ast}$ are greater than those of $\boldsymbol{p}_{1}$, namely $\boldsymbol{p}_{1}^{\ast} > \boldsymbol{p}_{1}$.
\end{enumerate}

\noindent It is worth noticing that the MPL index enjoys also the following properties:
\begin{itemize}
\item[P.6] \textit{Positivity}: When $\boldsymbol{p}_{1}^{\ast}=\alpha \boldsymbol{p}_{1}$, the MPL index becomes
\begin{equation}\label{eq:wj}
\widehat{\lambda}(\alpha\boldsymbol{p}_{1},\boldsymbol{p}_{2},\boldsymbol{q}_{1},\boldsymbol{q}_{2})=\frac{1}{\alpha}\frac{\boldsymbol{p}_{2}'\overline{\boldsymbol{\pi}}}{ \boldsymbol{p}_{1}'\overline{\boldsymbol{\pi}}}\geq0.
\end{equation}
\item[P.7] \textit{Inverse proportionality in the base period}: For the proof see Eq.~\eqref{eq:wj}. 
\item[P.8] \textit{Commodity reversal property}: It follows straightforward that the index price is invariant with respect to any permutation $(i)$:
\[\widehat{\lambda}(\boldsymbol{p}_{1},\boldsymbol{p}_{2},\boldsymbol{q}_{1},\boldsymbol{q}_{2})= \frac{\sum_{(i)=1}^{(N)} p_{(i)2} \overline{\pi}_{(i)}}{\sum_{(i)=1}^{(N)}p_{(i)1} \overline{\pi}_{(i)}}= \frac{\sum_{i=1}^N p_{i2} \overline{\pi}_{i}}{\sum_{i=1}^Np_{i1} \overline{\pi}_{i}}.\]
\item[P.9] \textit{Quantity reversal test}:  Simple computations prove that the index price $\widehat{\lambda}$ does not change as a consequence of a change in the quantity which affects only the weights $\overline{\pi}_i$. 
\item[P.10] \textit{Base reversibility (symmetric treatment of time)}: This property is satisfied by the MPL index under a suitable choice of the weights. Setting $\vartheta_{i}=p_{i2}$, or $\vartheta_{i}=z_{i}p_{i2}$ where $z_i$ is a variable or more simply a constant term $z_{i}=z$, leads to weights, $\overline{\pi}_{i}^{\ast}$, of the MPL index that does not depend on prices
\[
\overline{\pi}_{i}^{\ast}=\frac{1}{z_i}\frac{q_{i1}^{2}q_{i2}^{2}}{q_{i1}^{2}+q_{i2}^{2}}, \ i=1,2,\dots,n.
\]
Hence, by denoting with $\overline{\boldsymbol{\pi}^{\ast}}$ the vector whose $i$-th entry is $\overline{\pi}_{i}^{\ast}$, simple computations prove that
\[
\widehat{\lambda}(\boldsymbol{p}_{2},\boldsymbol{p}_{1},\boldsymbol{q}_{2},\boldsymbol{q}_{1})=\frac{\boldsymbol{p}_{1}'\overline{\boldsymbol{\pi}}^{\ast}}{\boldsymbol{p}_{2}'\overline{\boldsymbol{\pi}^{\ast}}}=\left(\frac{\boldsymbol{p}_{2}'\overline{\boldsymbol{\pi}}^{\ast}}{\boldsymbol{p}_{1}'\overline{\boldsymbol{\pi}}^{\ast}}\right)^{-1}.
\]
\item[P.11] \textit{Transitivity}: For a particular choice of $\theta_{i}$ and by using the square roots of the quantities, we have proved that the MPL index tallies with the GK index which satisfies all tests for multilateral comparison proposed by \citet{balk1996comparison} except for the proportionality one. 
\item[P.12] \textit{Monotonicity}: If $\boldsymbol{p}_{2}=\beta \boldsymbol{p}_{1}$ then the following holds for the weights of the MPL index
\[
\overline{\pi}_{i}^{\ast}=\beta\left(p_{i1}\frac{1}{\vartheta_i}\frac{q_{i1}^2q_{i2}^2}{q_{i1}^2+q_{i2}^2}\right)=\beta p_{i1}\gamma_{i}.
\]
Hence
\[
\widehat{\lambda}(\boldsymbol{p}_{1},\beta\boldsymbol{p}_{1},\boldsymbol{q}_{1},\boldsymbol{q}_{2})=\beta \frac{\boldsymbol{p}_{1}'\overline{\boldsymbol{\pi}}^{\ast}}{\boldsymbol{p}_{1}'\overline{\boldsymbol{\pi}}^{\ast}}=\beta.
\]
\end{itemize}

\end{proof}

\clearpage
\section{Empirical application of MPL: graphics }

\subsection{MPL index and the cultural supply: a multi-period perspective}\label{app:mplapp}
\subsubsection{Incomplete and complete price tableau}\label{app:mpl_m}
\begin{figure}[htbp]
\begin{center}
\includegraphics[scale=0.42]{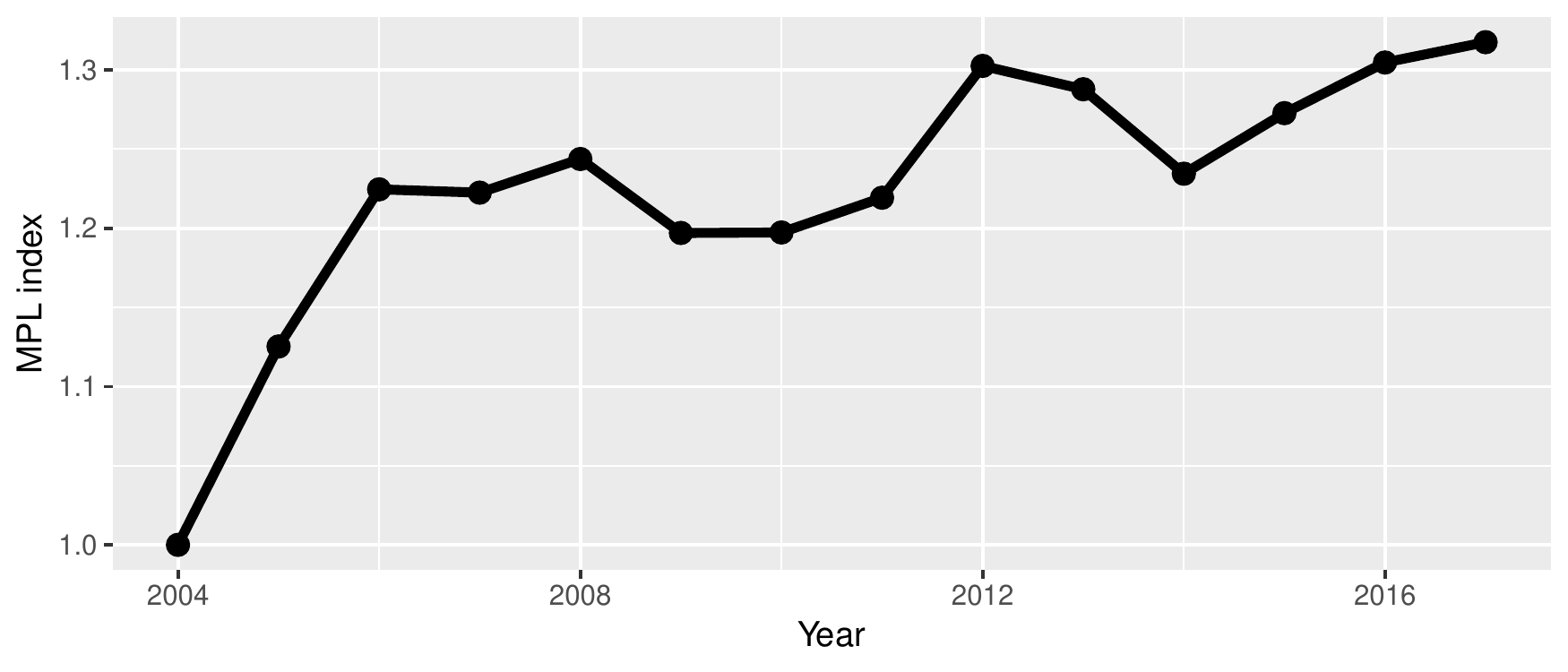}
\includegraphics[scale=0.42]{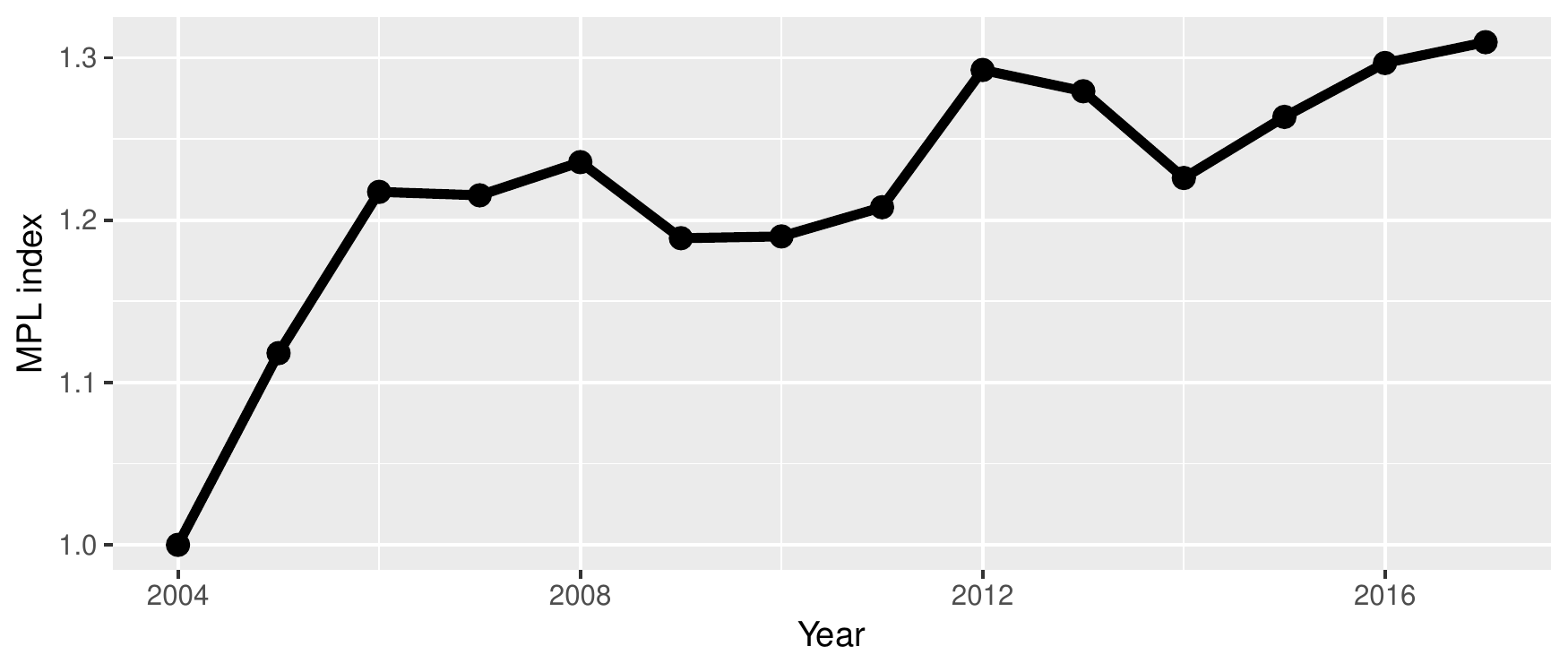}\\
\includegraphics[scale=0.42]{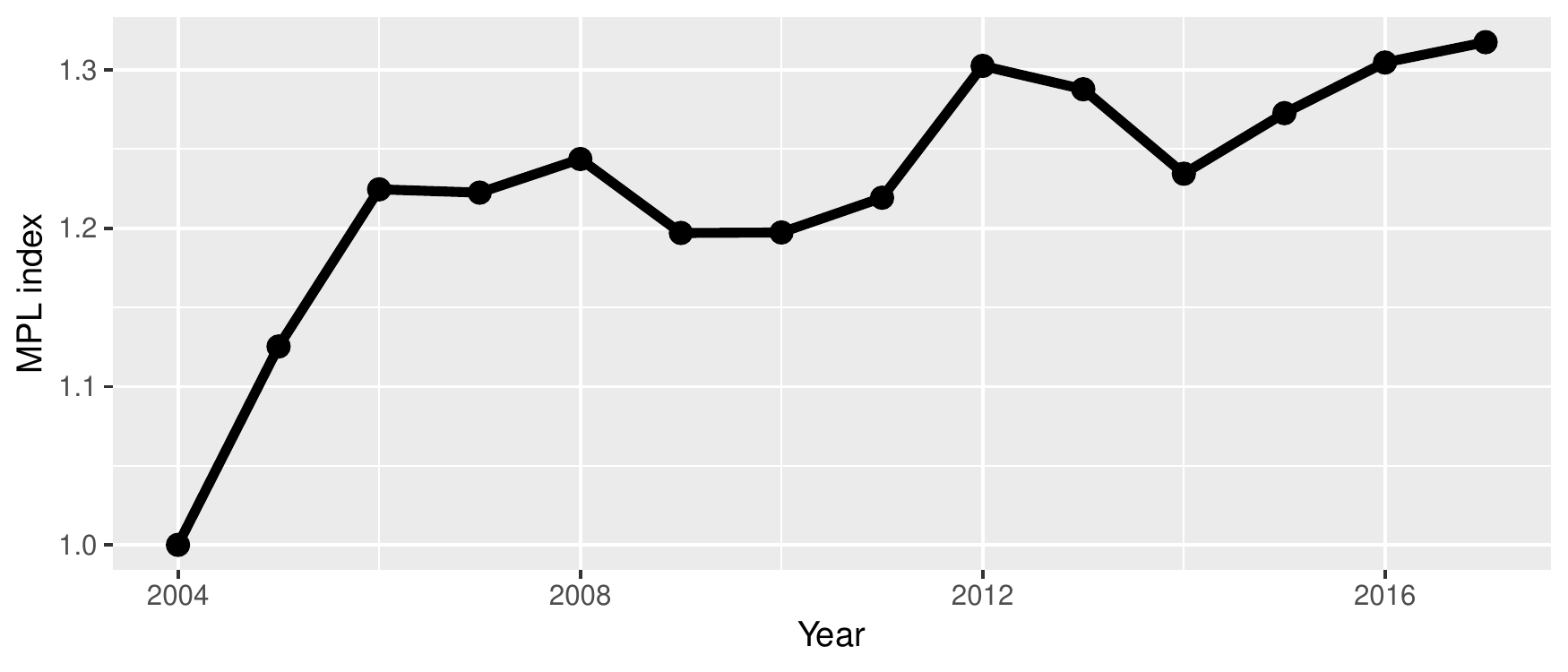}
\includegraphics[scale=0.42]{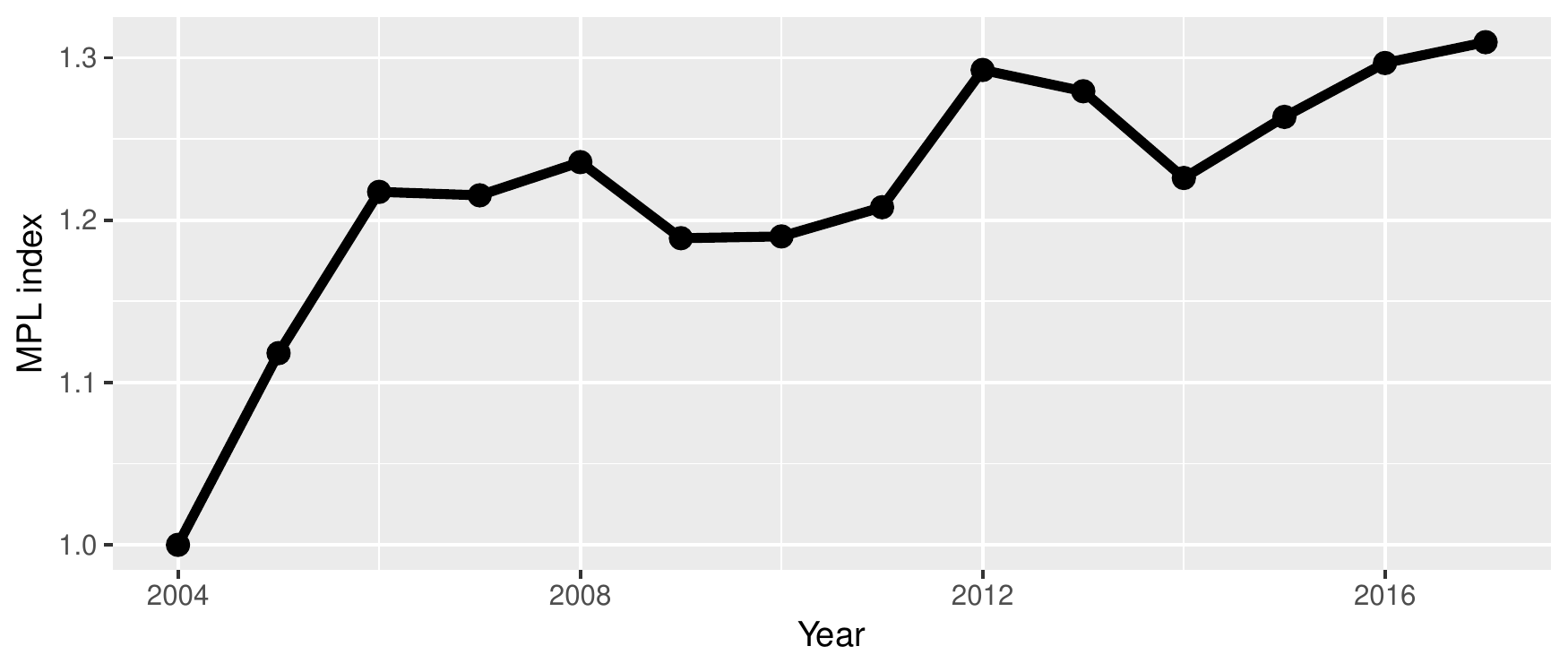}\\
\includegraphics[scale=0.42]{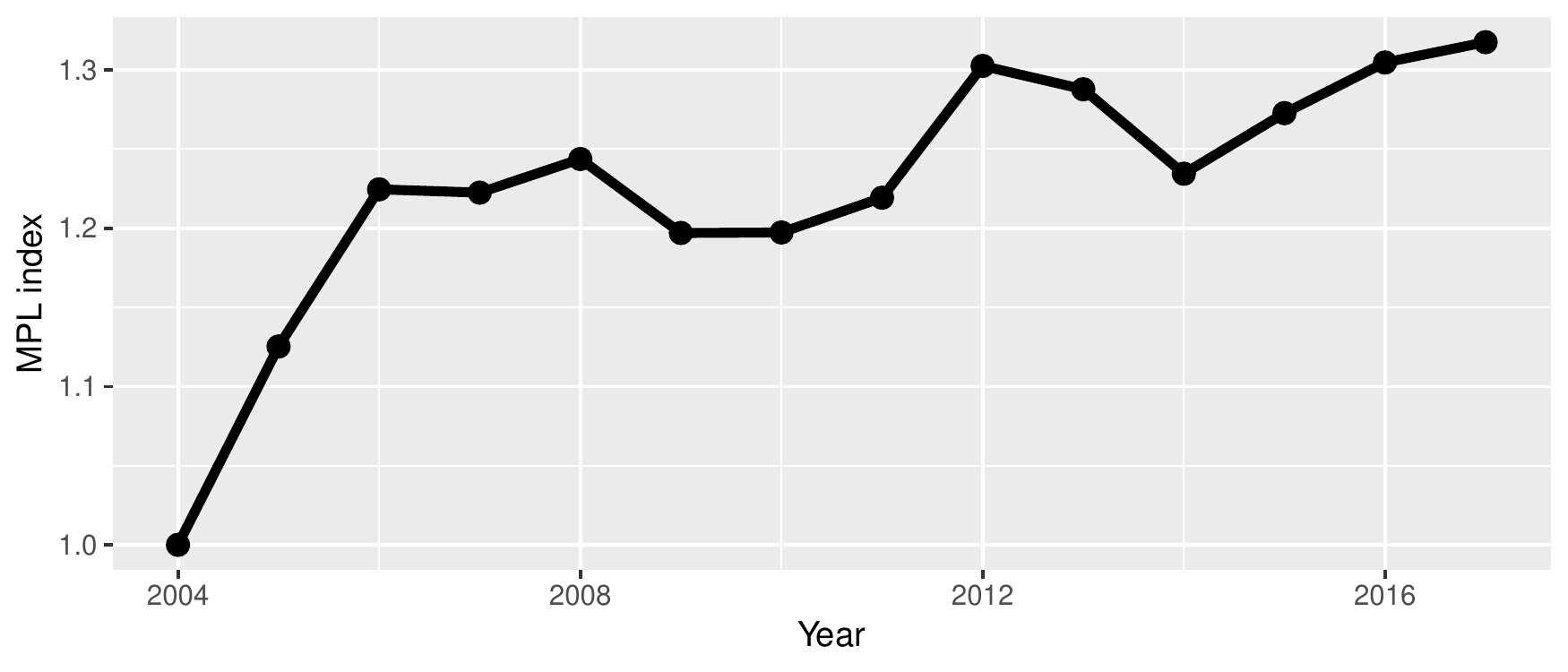}
\includegraphics[scale=0.42]{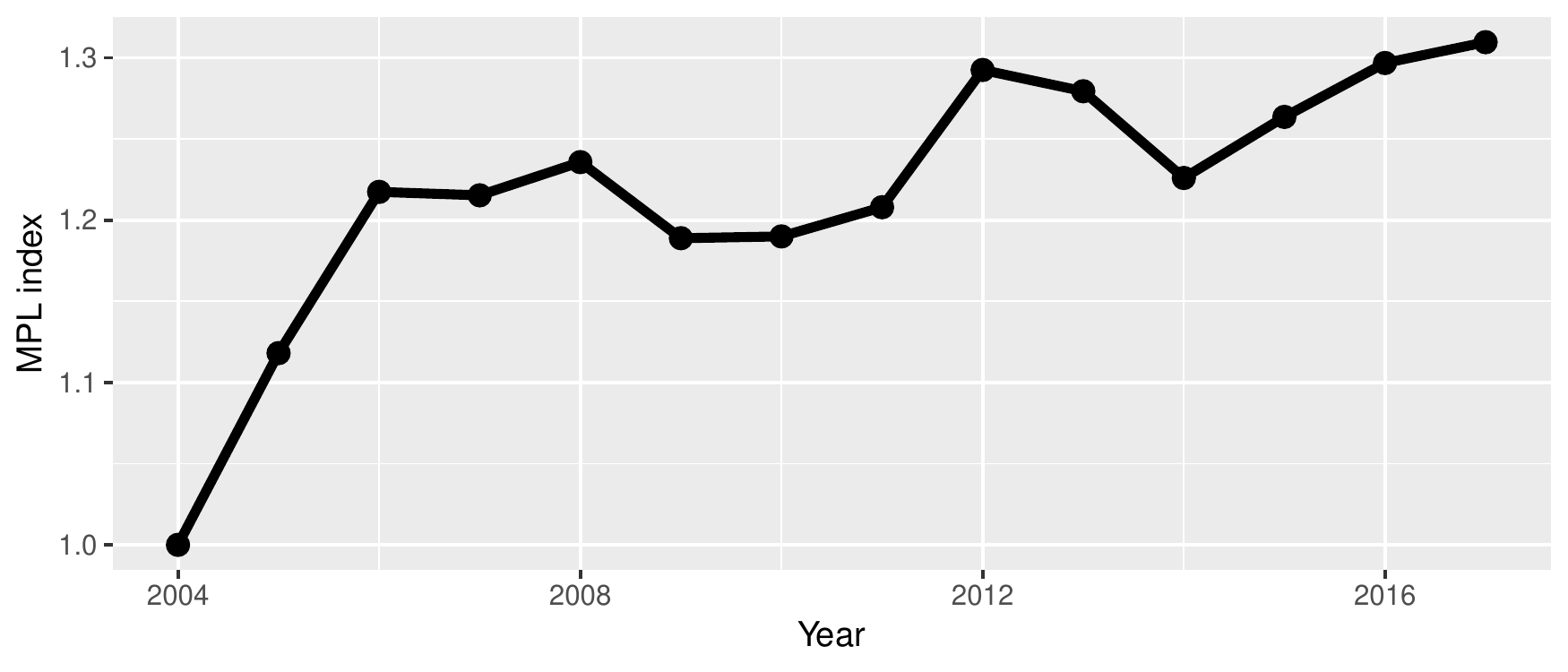}
\caption{GLS-f, GLS-s, and OLS versions of the MPL index are shown in the top, middle and bottom panel, respectively for the case of incomplete (left) and complete (right) price tableau.}
\label{fig:5app_mis}
\end{center}
\end{figure}


\begin{figure}[htbp]
\begin{center}
\includegraphics[scale=0.42]{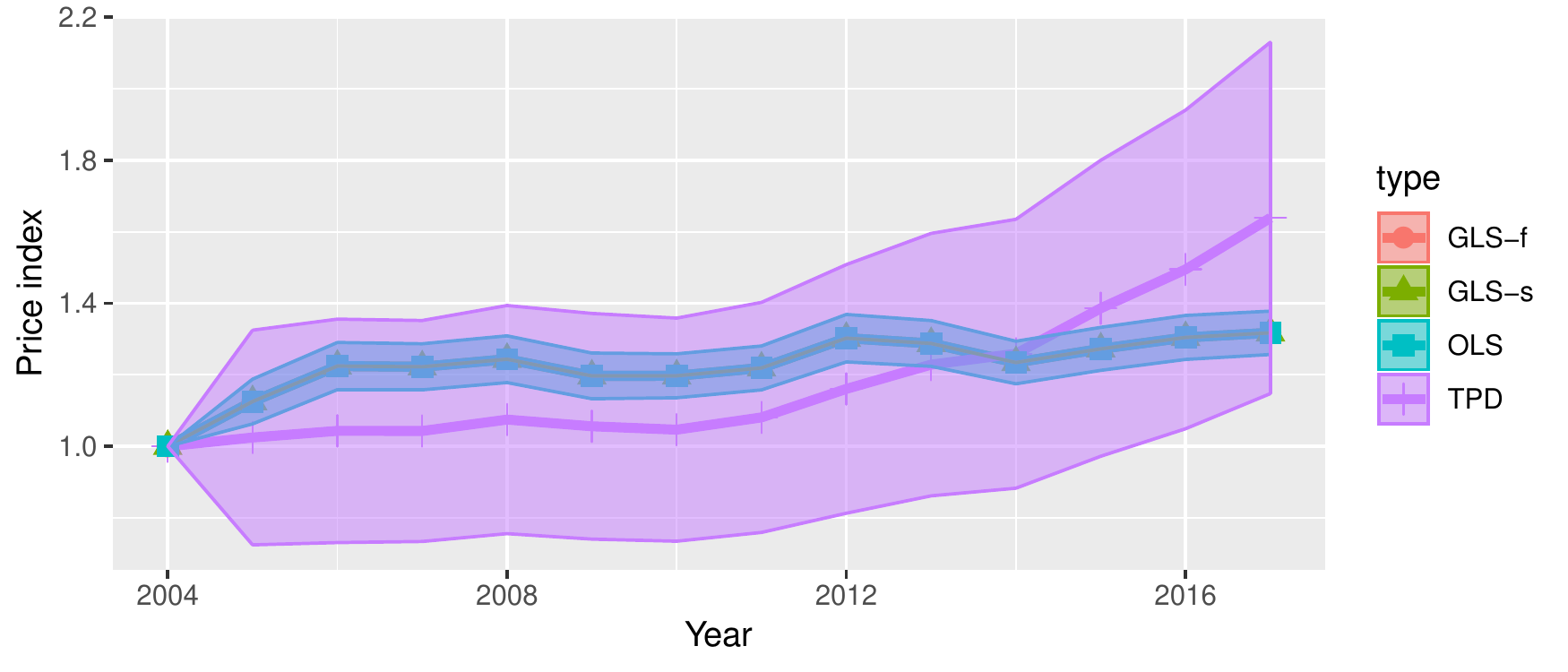}
\includegraphics[scale=0.42]{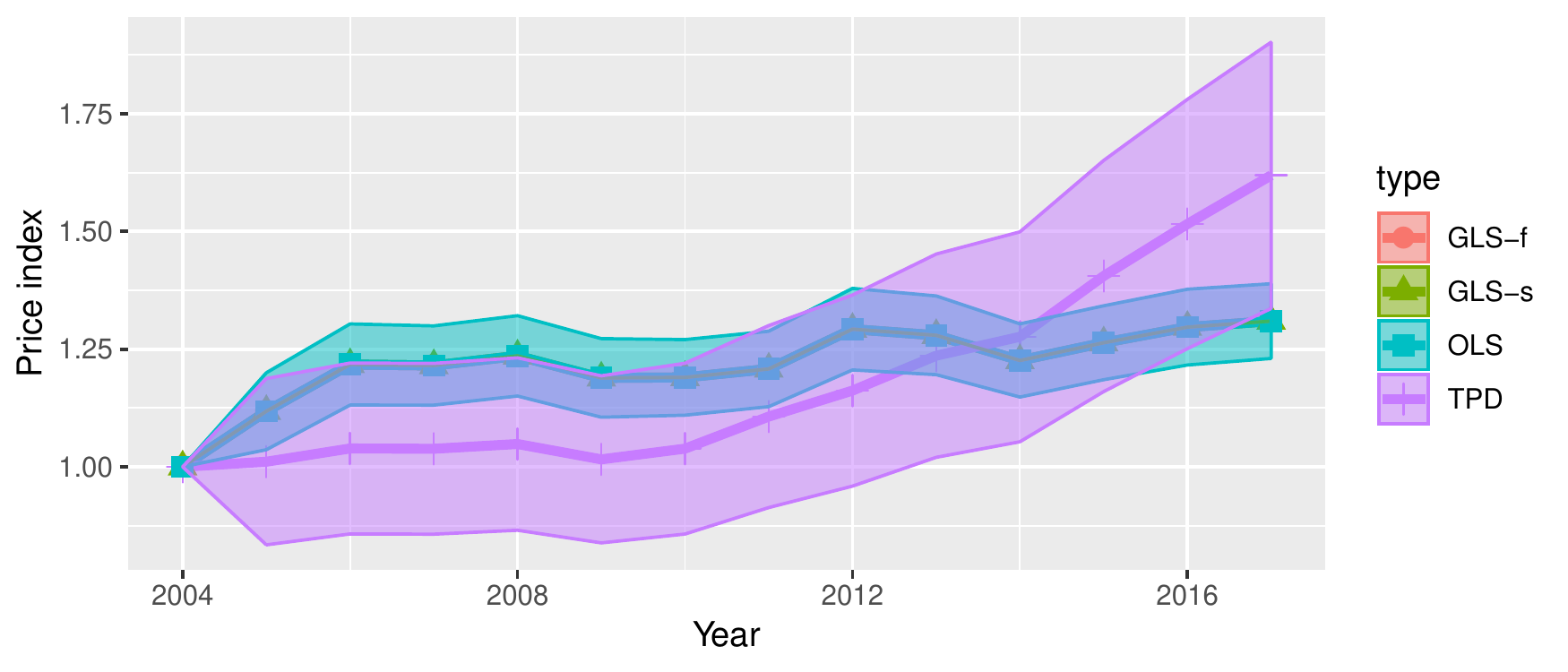}
\caption{Comparison of the GLS-f, GLS-s  and OLS versions of the MPL index with the TPD for the case of incomplete (left) and complete (right) tableau.}
\label{fig:5app_mis1}
\end{center}
\end{figure}

\begin{figure}[htbp]
\begin{center}
\includegraphics[scale=0.42]{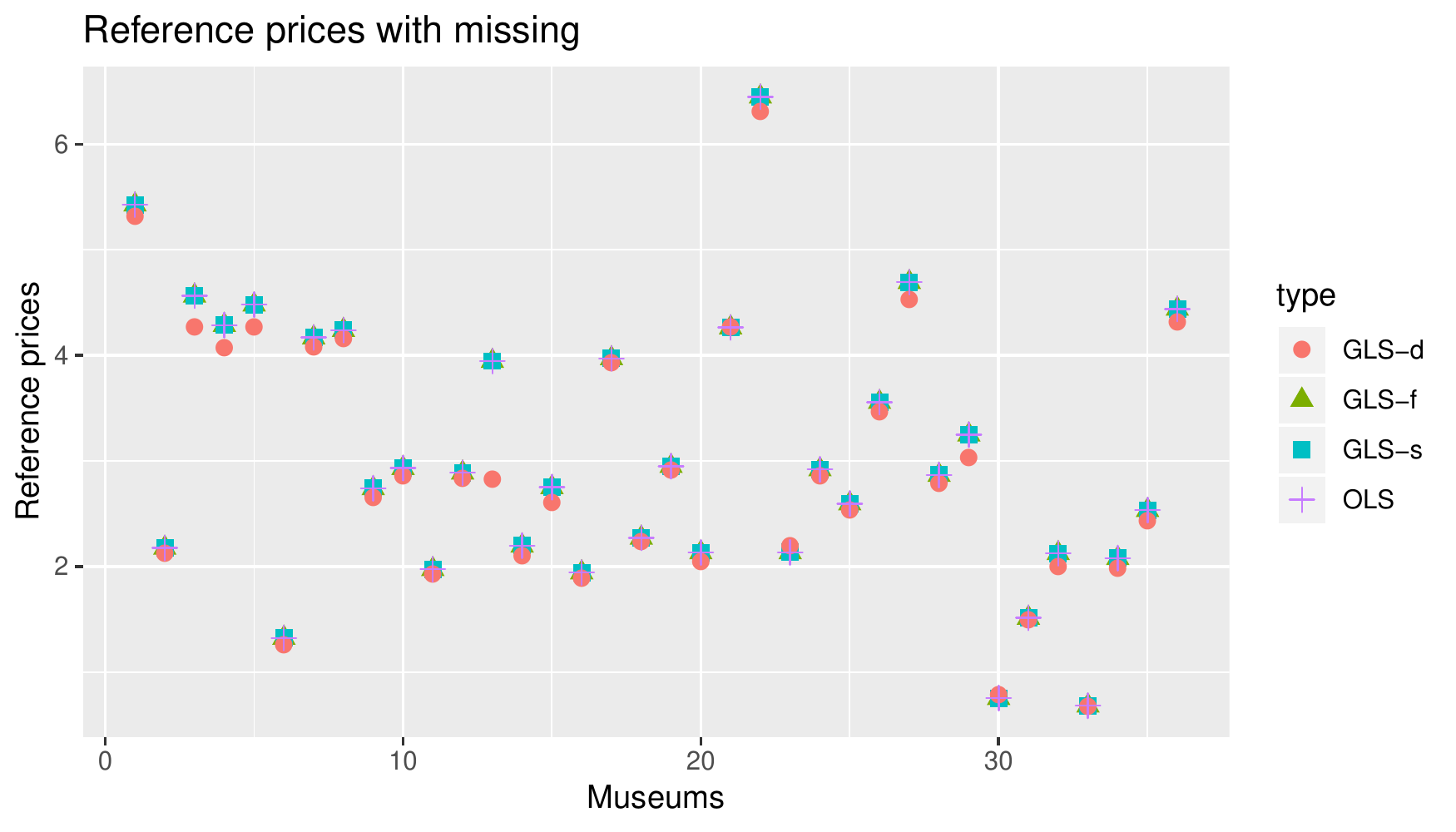}
\includegraphics[scale=0.42]{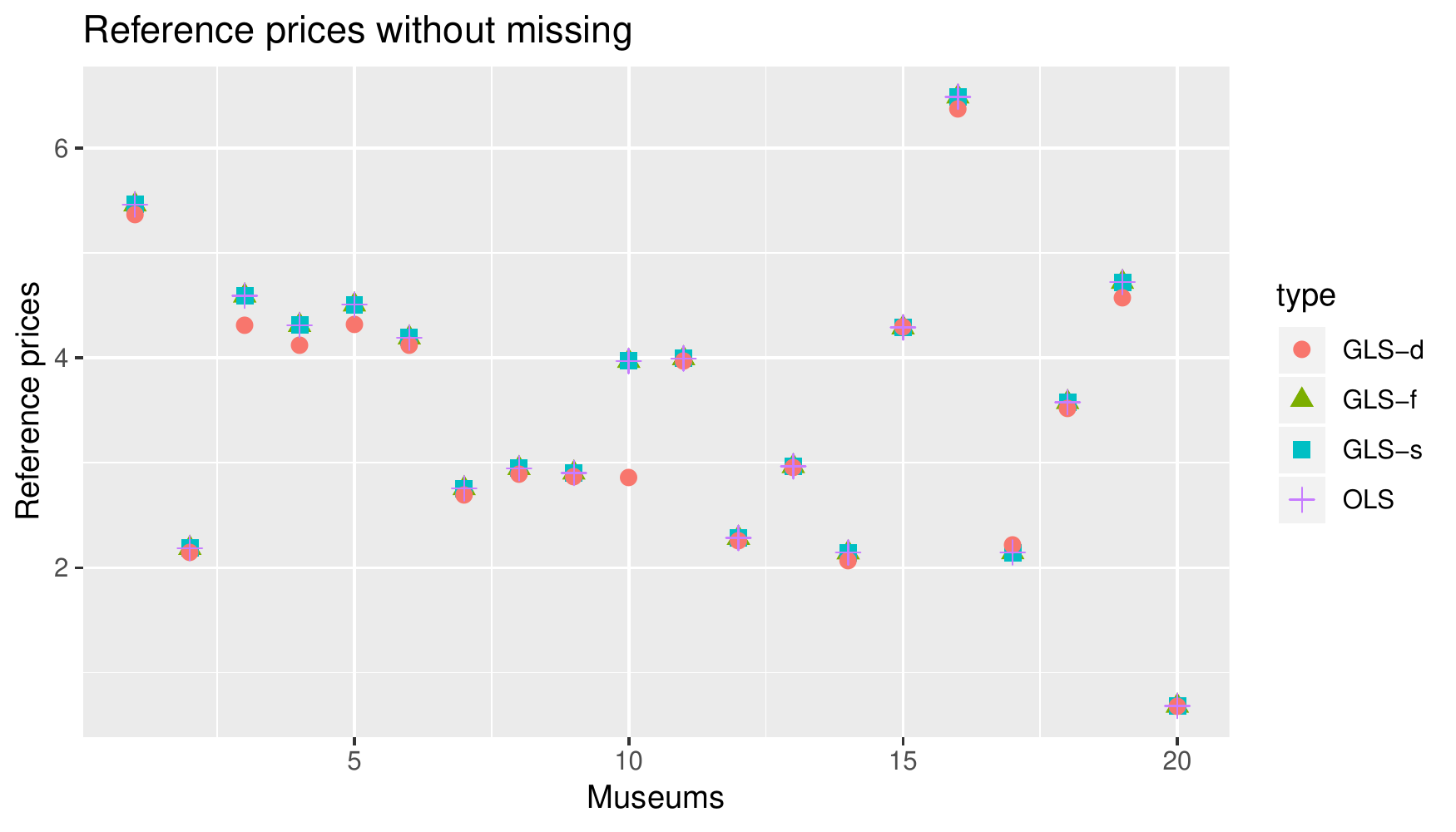}
\caption{Comparison of the GLS-f, GLS-s, GLS-d and OLS versions of the MPL reference prices for the case of incomplete (left) and complete (right) tableau.}
\label{fig:5app_mis1b}
\end{center}
\end{figure}



\clearpage
%
%
%
%
%
%

\subsection{MPL index and the cultural supply: a multilateral perspective}\label{app:mpllat}

\begin{figure}[htbp]
\begin{center}
\includegraphics[scale=0.42]{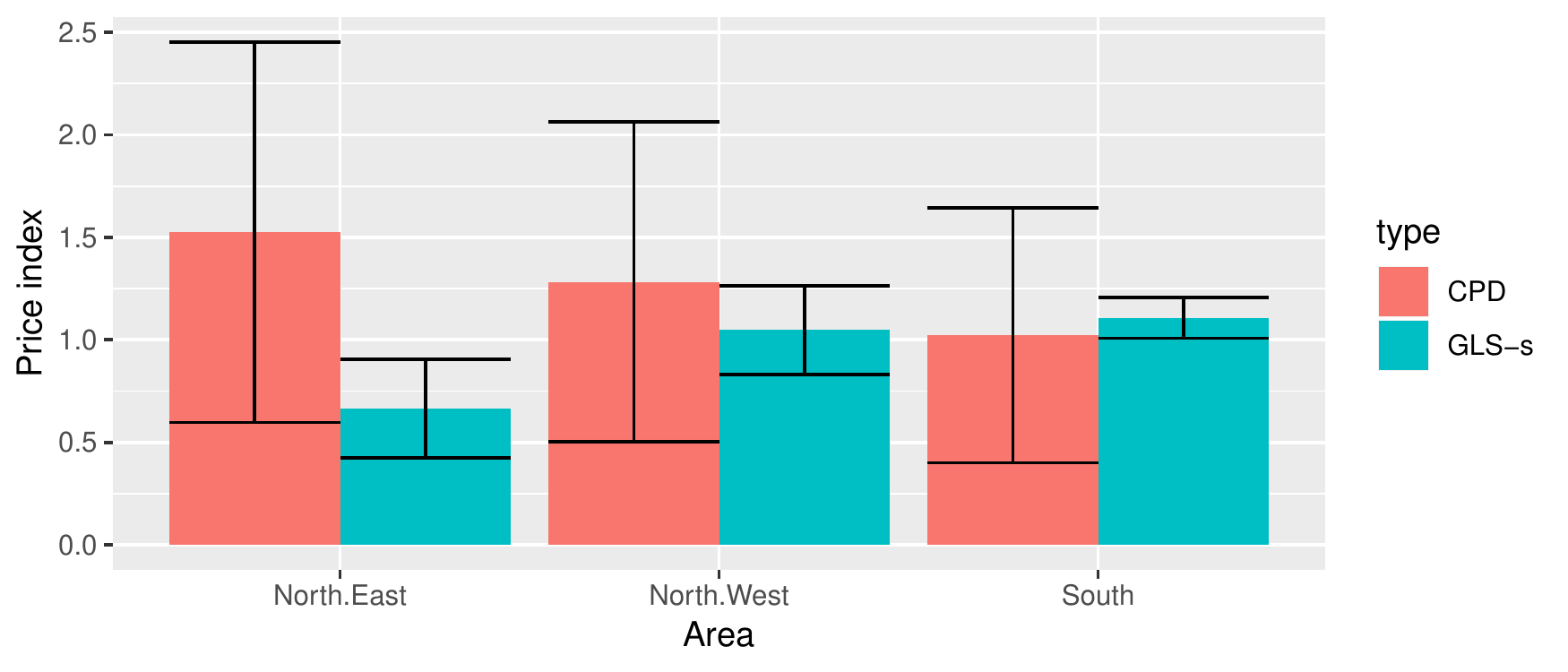}
\includegraphics[scale=0.42]{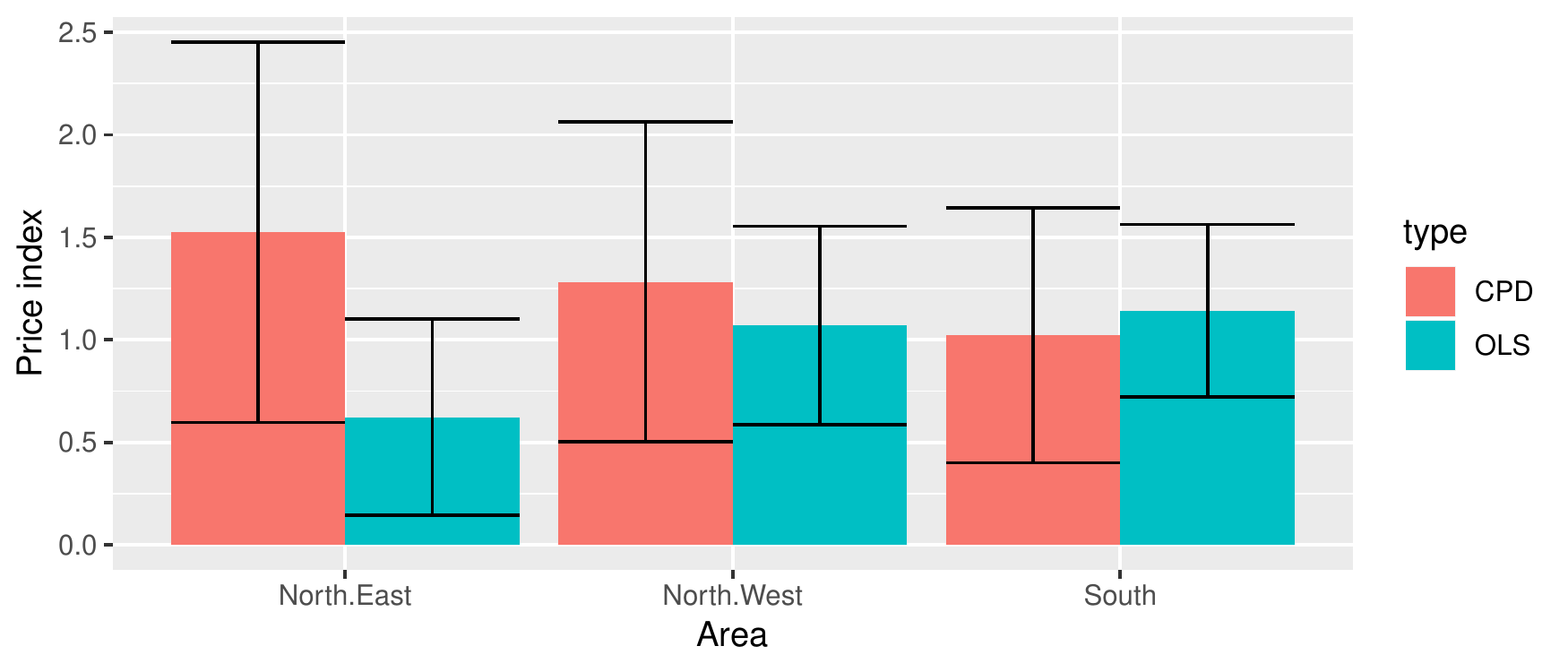}
\caption{MPL and CPD indices with their $2\sigma$ confidence bounds. Left panel shows the GLS-s and the right one the OLS version of the MPL index. }
\label{fig:5capp}
\end{center}
\end{figure}

\subsection{MPL index and cultural supply: a simulation}\label{app:mplsim}

\begin{figure}[htbp]
\begin{center}
\includegraphics[scale=0.42]{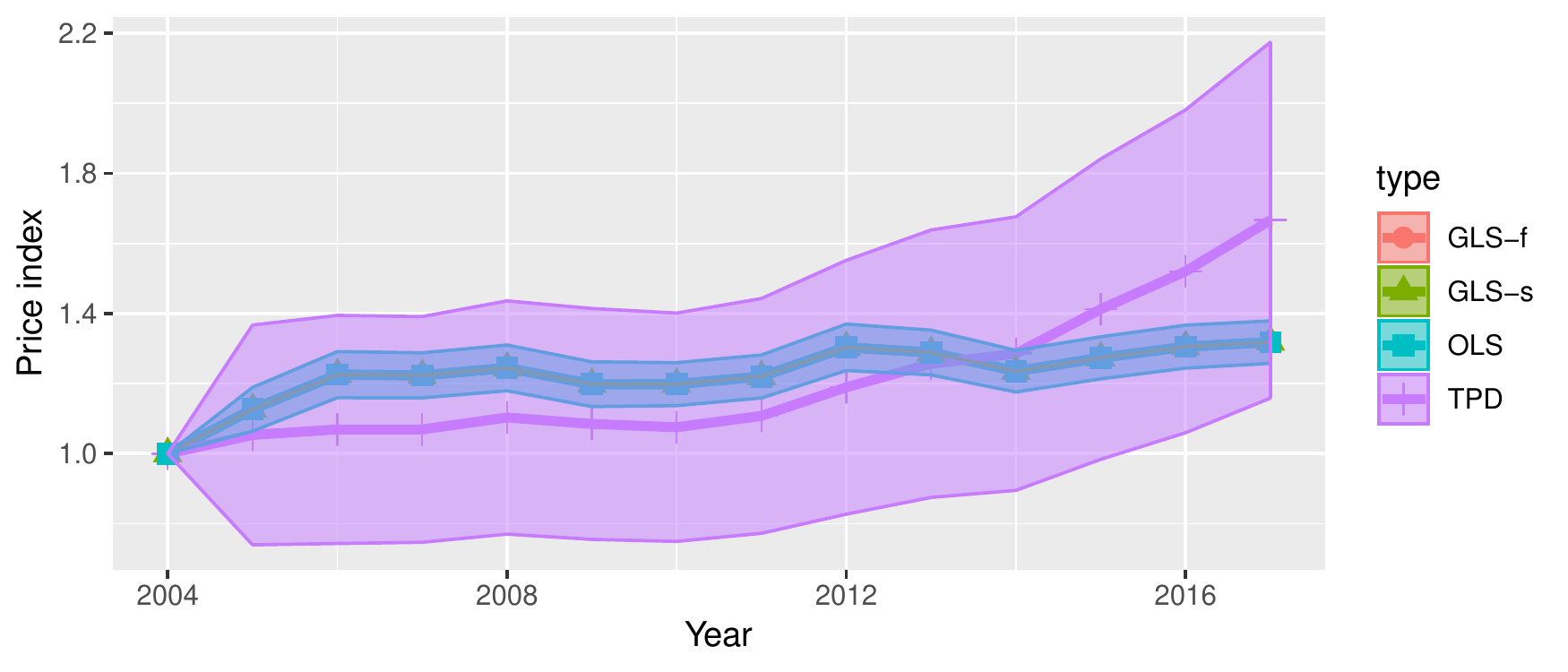}
\includegraphics[scale=0.42]{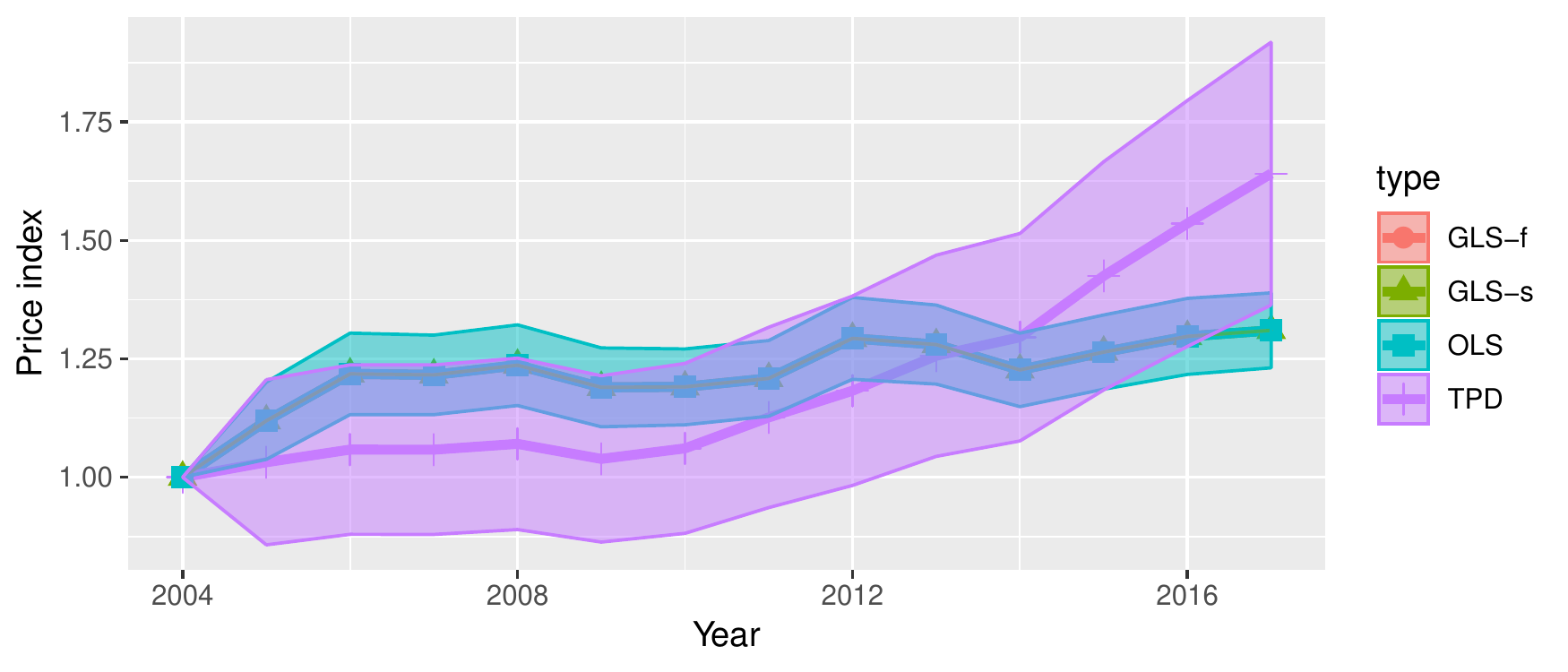} 
\caption{
Comparison of the GLS-f, GLS-s and OLS  versions of the MPL index with the TPD index for the case of complete (left panels) and incomplete (right panels) price tableau. The MPL indices have been obtained on simulated data by adding to $\boldsymbol{V}_1$ random terms drawn from a Normal law with a mean equal to 20000 and a standard error varying randomly from 0 to 1000. 
Confidence bands are also plotted.
}
\label{fig:sim1all}
\end{center}
\end{figure}

\begin{figure}[htbp]
\begin{center}
\includegraphics[scale=0.42]{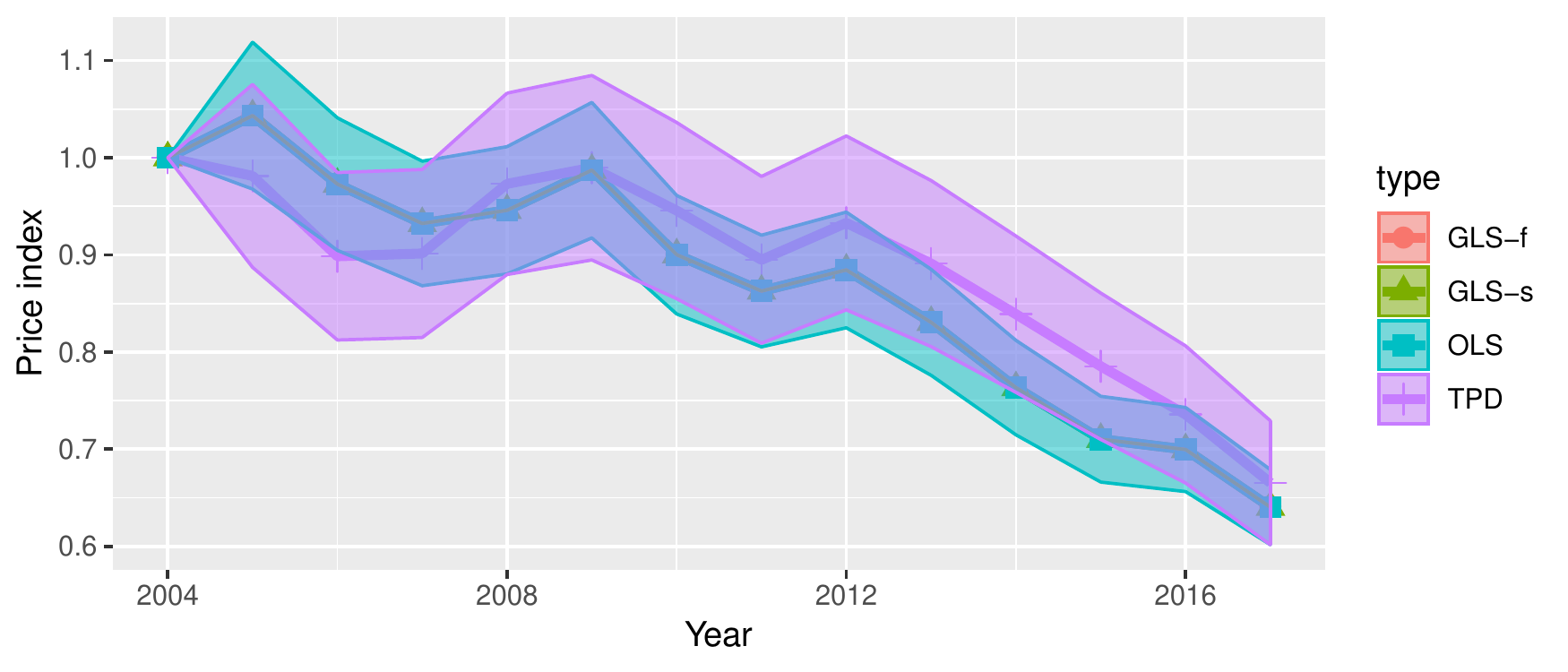} 
\includegraphics[scale=0.42]{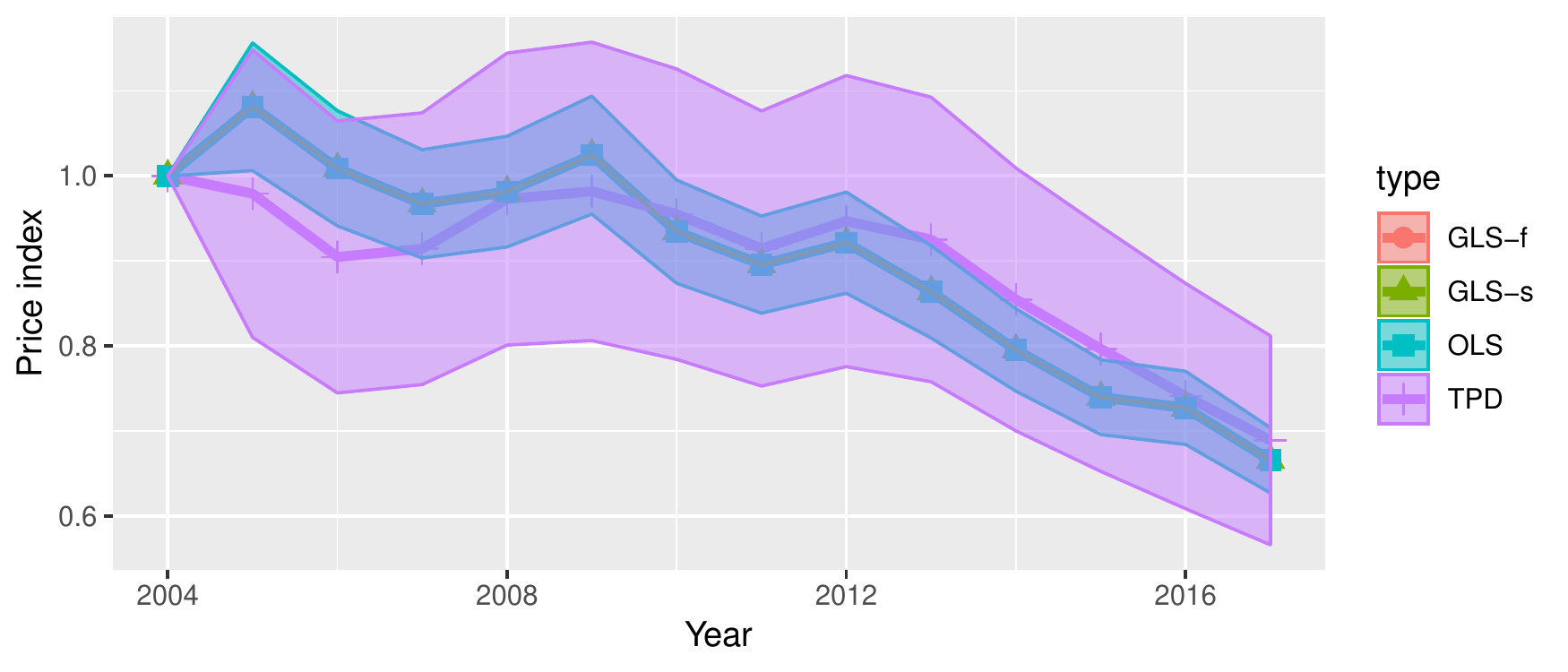}\\
\caption{
Comparison of the GLS-f, GLS-s  and OLS  versions of the MPL index with the TPD index for the case of complete (left panels) and incomplete (right panels) price tableau. The MPL indices have been obtained on simulated data by adding to $\boldsymbol{\upsilon}_{t}$, for $t=2,\dots,T$, stochastic terms, drawn from a Normal law with a mean equal to -5000 and a standard error varying randomly from 0 to 800, to $\boldsymbol{\upsilon}_{t-1}$. Confidence bands are also plotted.
}

\label{fig:sim1sta}
\end{center}
\end{figure}

\clearpage
\section{Simulations}\label{app:sim}

\begin{figure}[htbp]
\begin{center}
\includegraphics[scale=0.28]{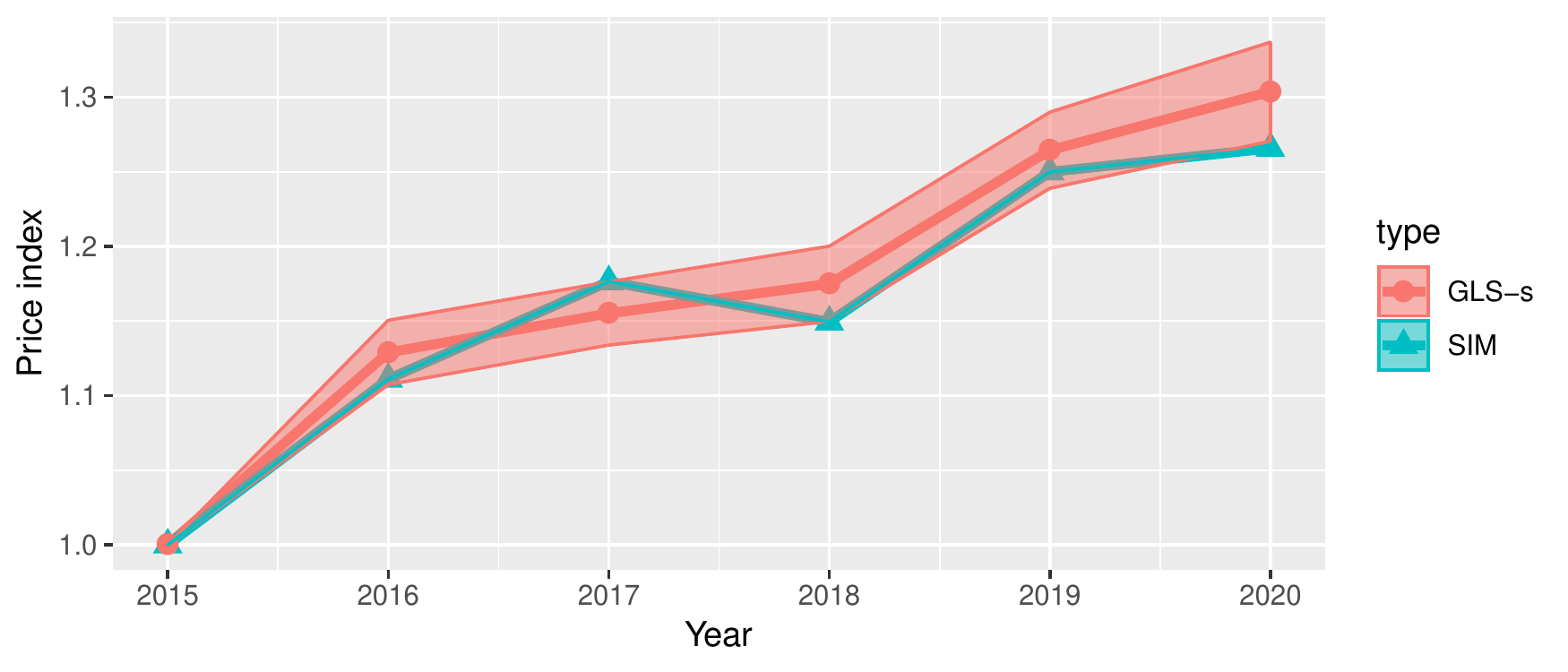}\includegraphics[scale=0.28]{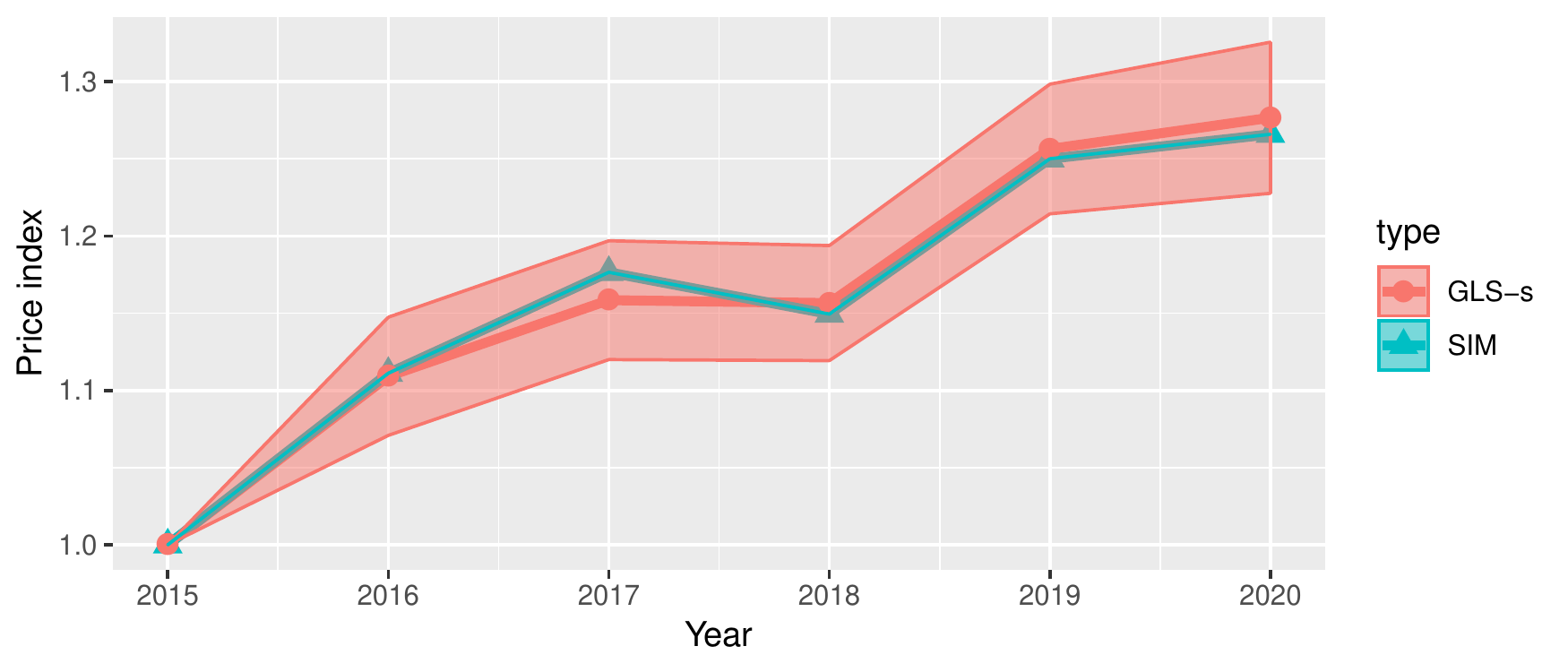}
\includegraphics[scale=0.28]{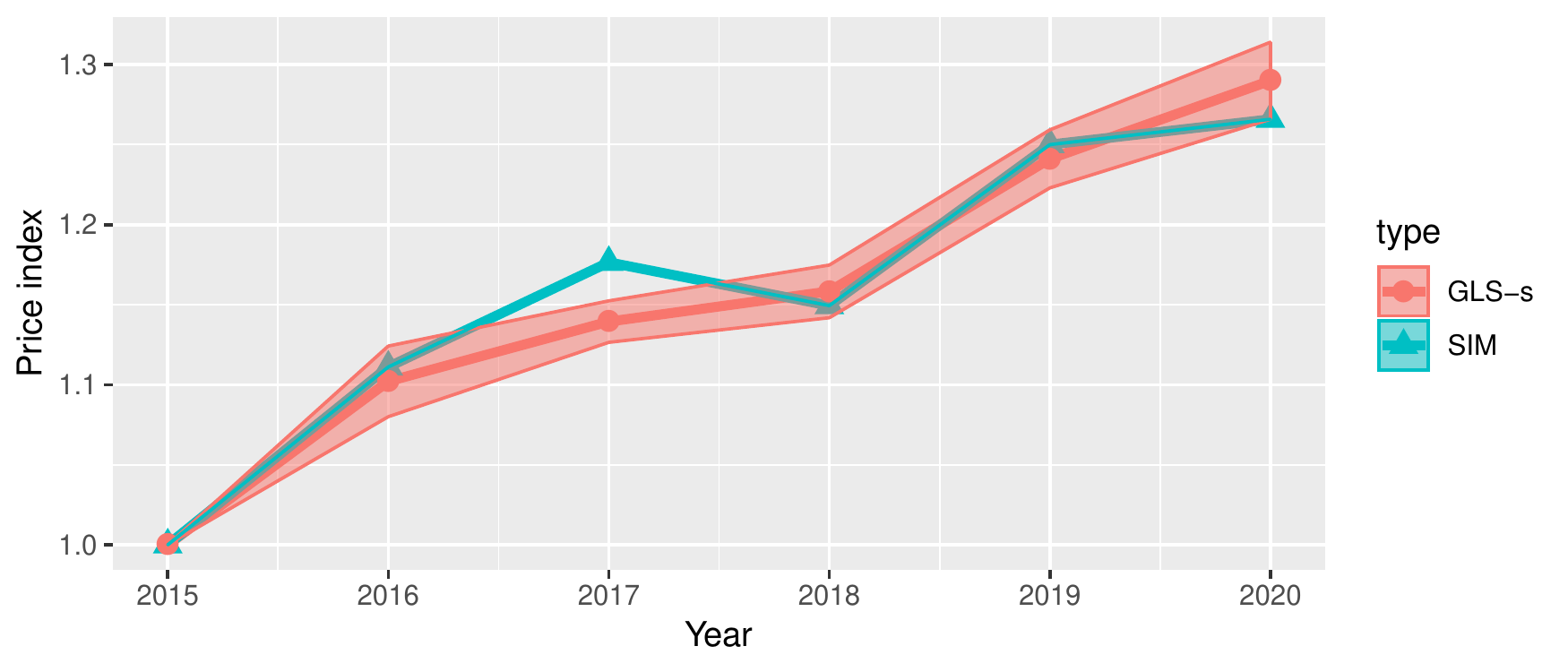}\\
\includegraphics[scale=0.28]{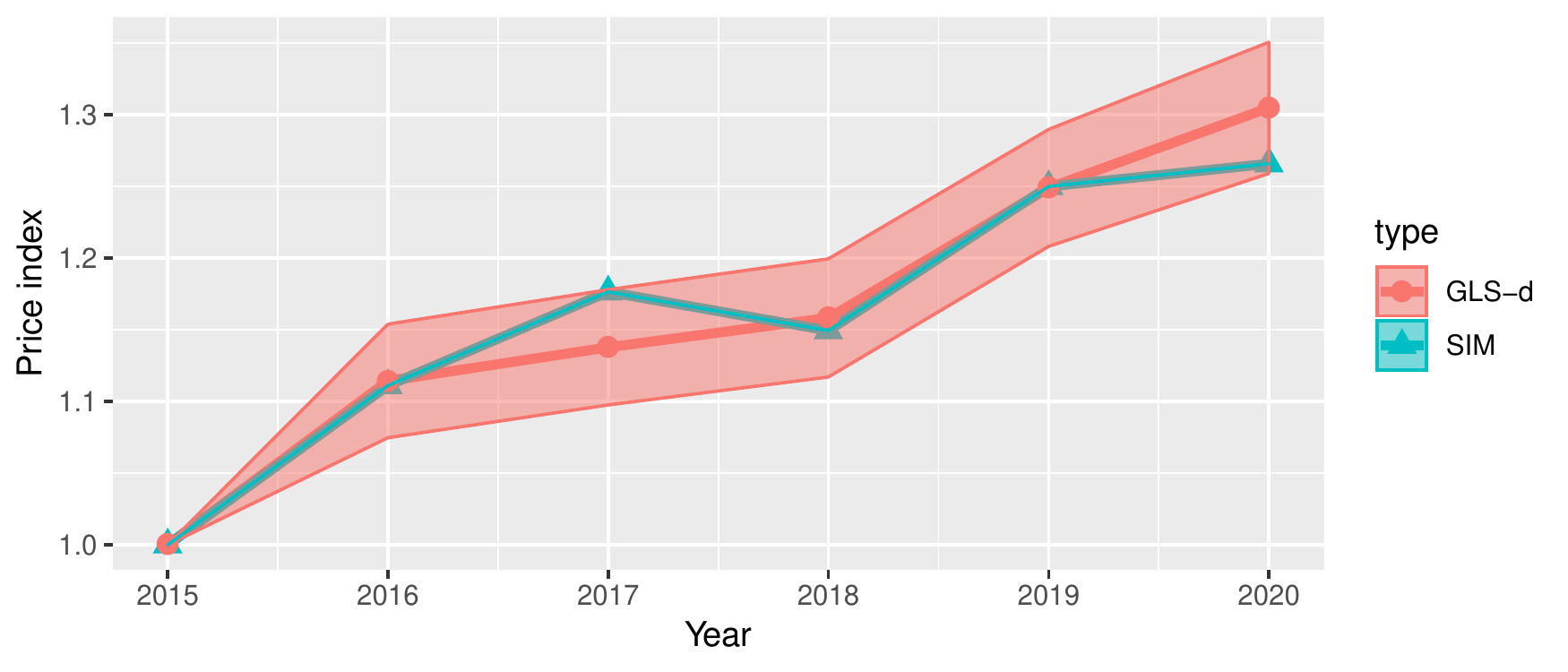} \includegraphics[scale=0.28]{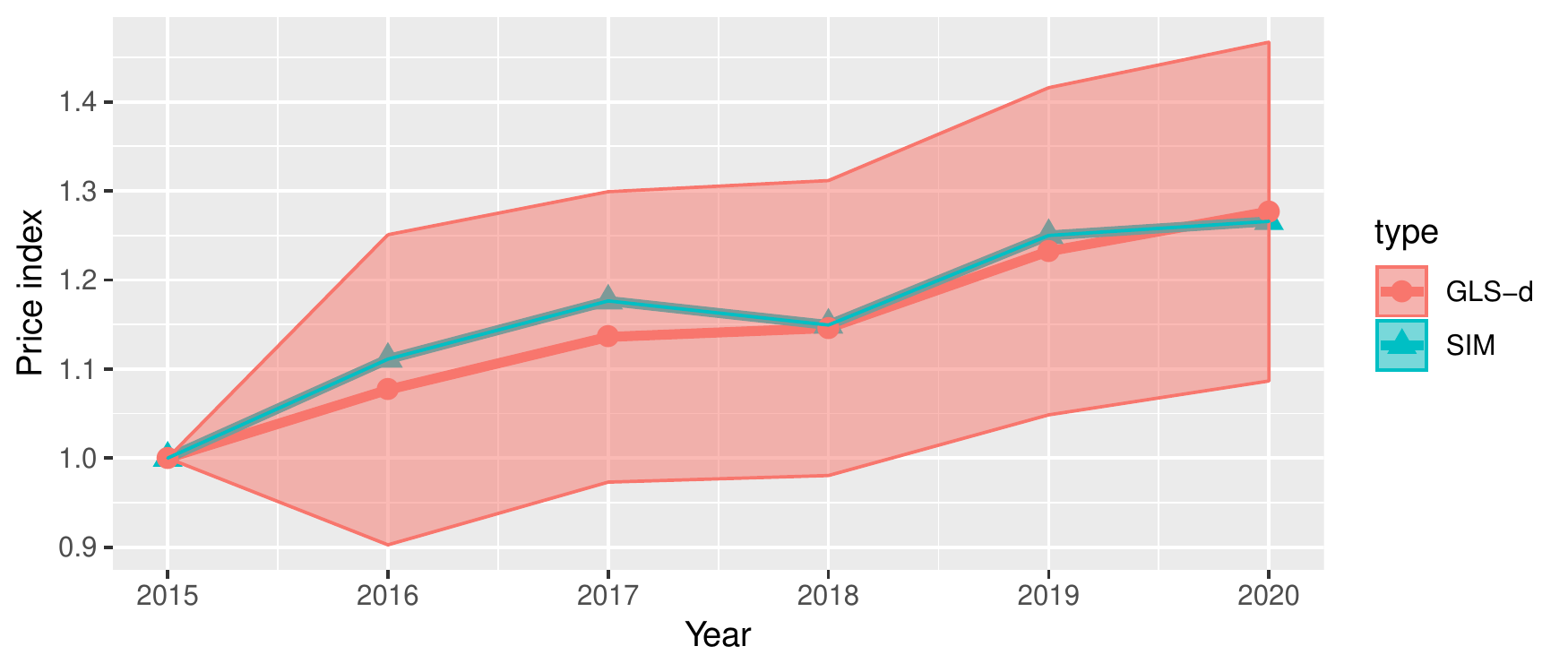}
\includegraphics[scale=0.28]{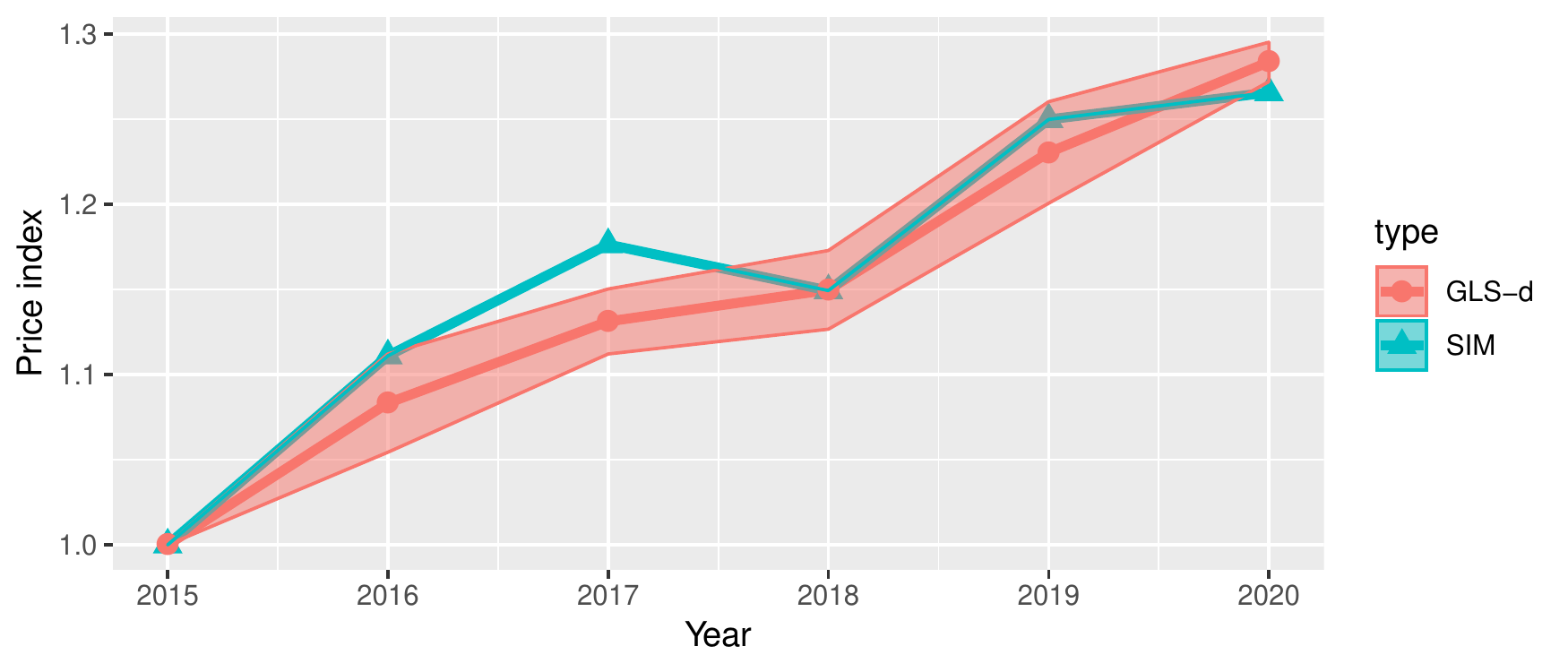}\\
\includegraphics[scale=0.28]{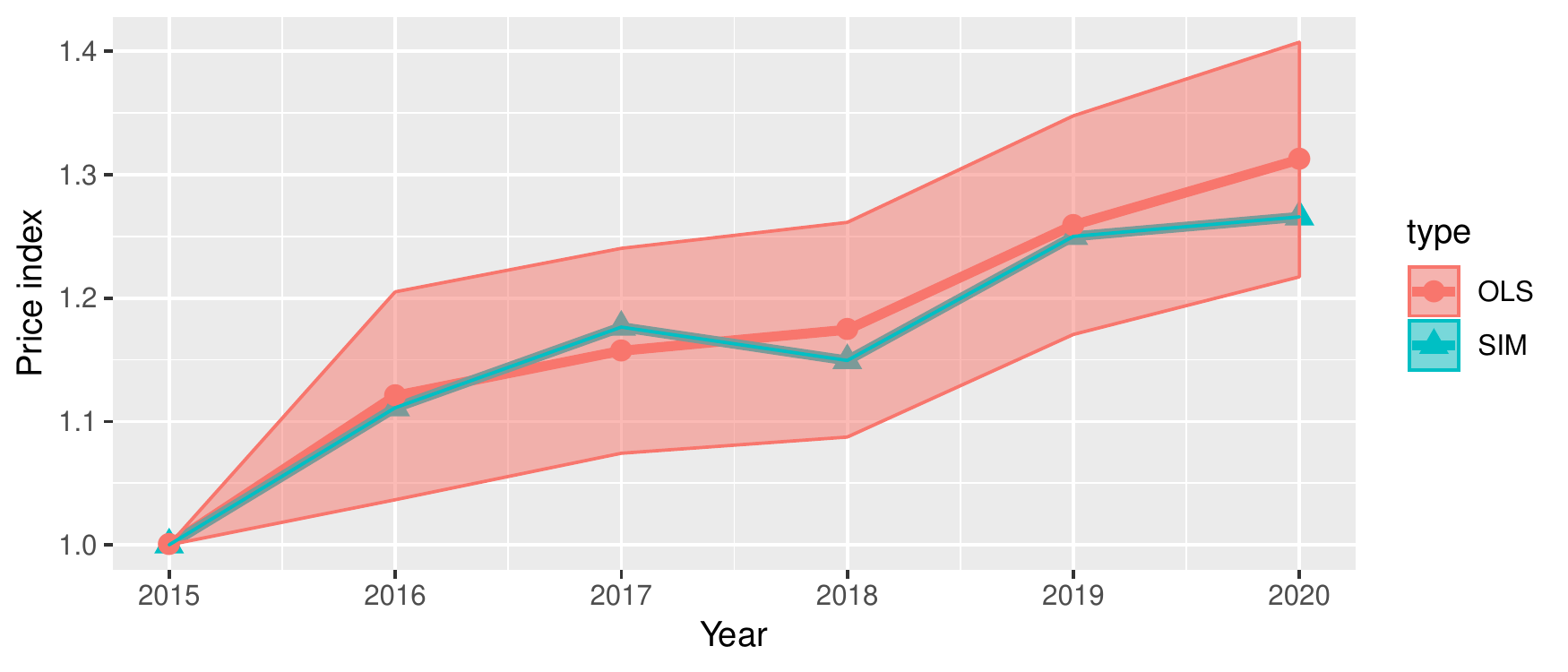} \includegraphics[scale=0.28]{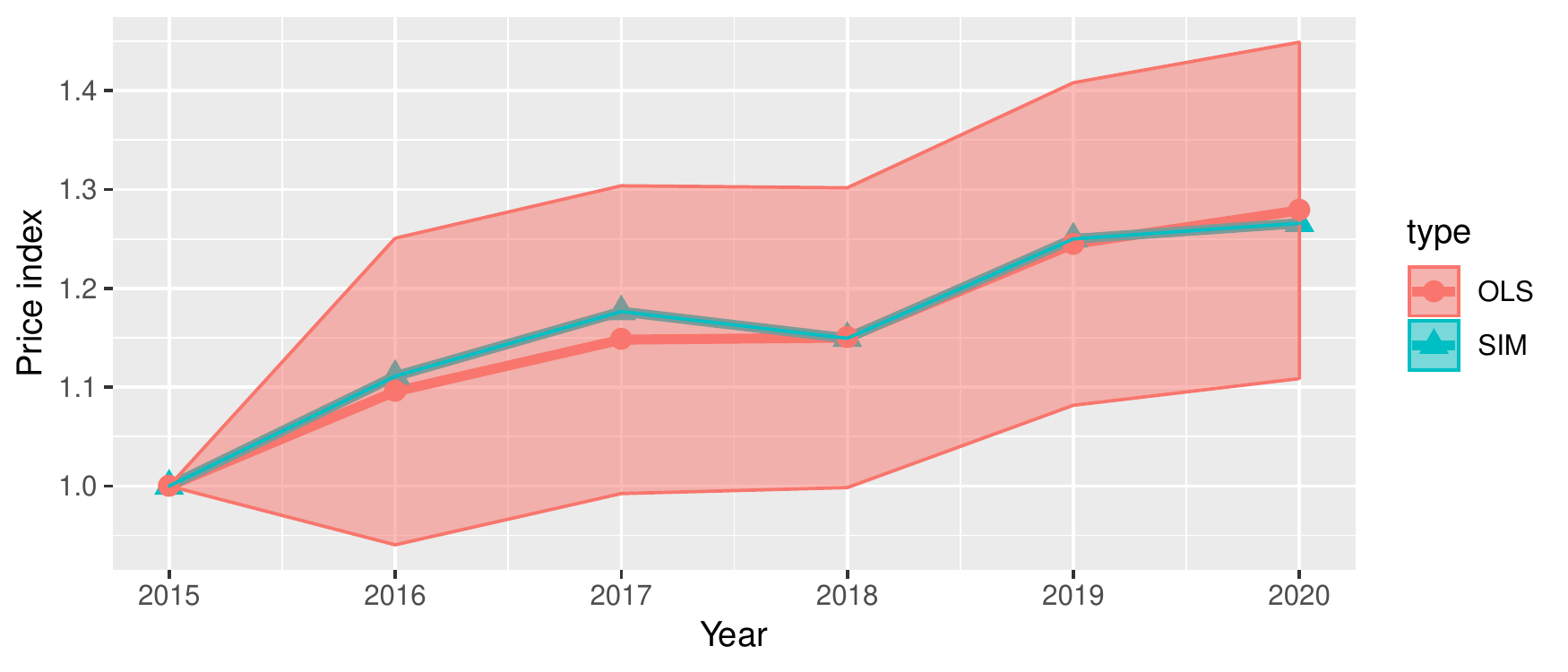} \includegraphics[scale=0.28]{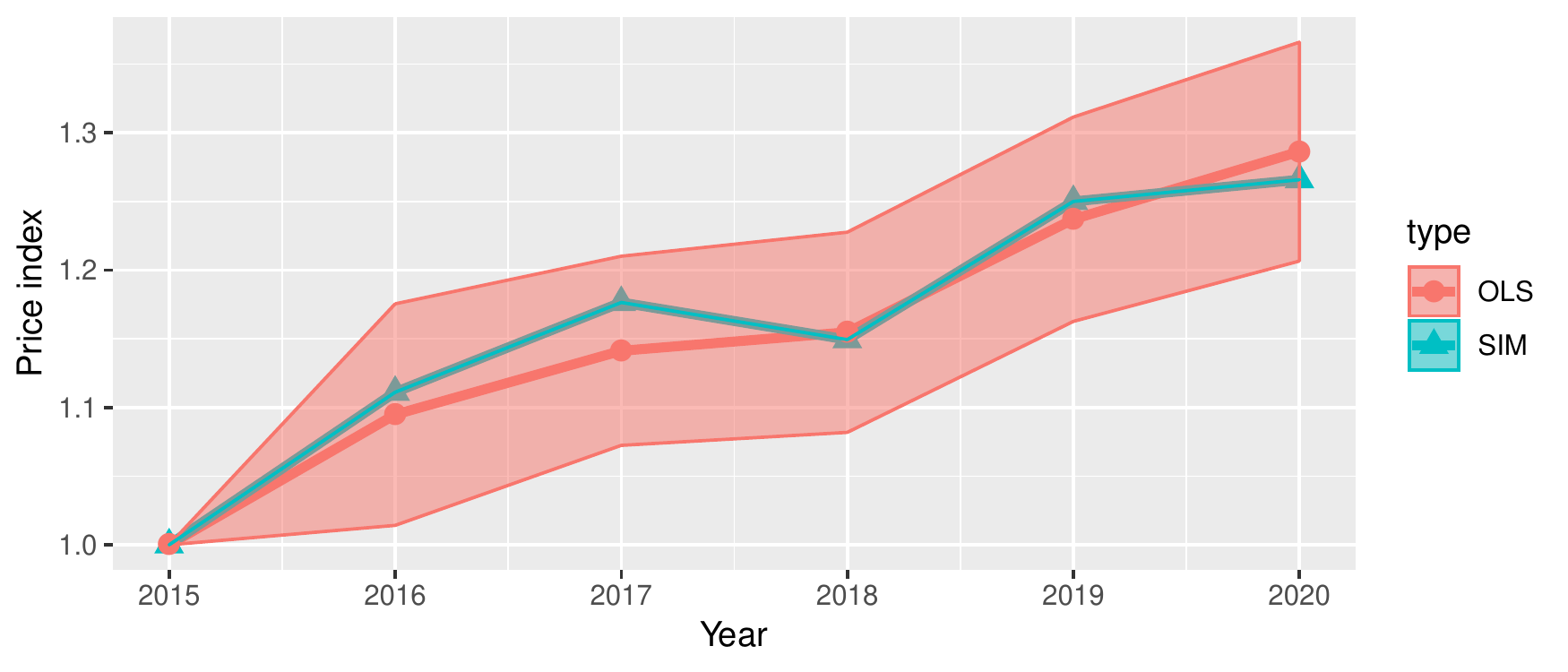}\\
\caption{Comparison of the 
GLS-s (first row panels), GLS-d (second row panels) and OLS (third row panels) versions of the MPL index with the "real" index $\boldsymbol{\lambda}$ assigned in the simulation (SIM), for the case of complete (left panels) and incomplete  price tableau with ``standard'' (central panels) and MPL (right panels) reference basket. Confidence bands are also plotted.
}
\label{fig:5sim1}
\end{center}
\end{figure}

\begin{figure}[htbp]
\begin{center}
\includegraphics[scale=0.42]{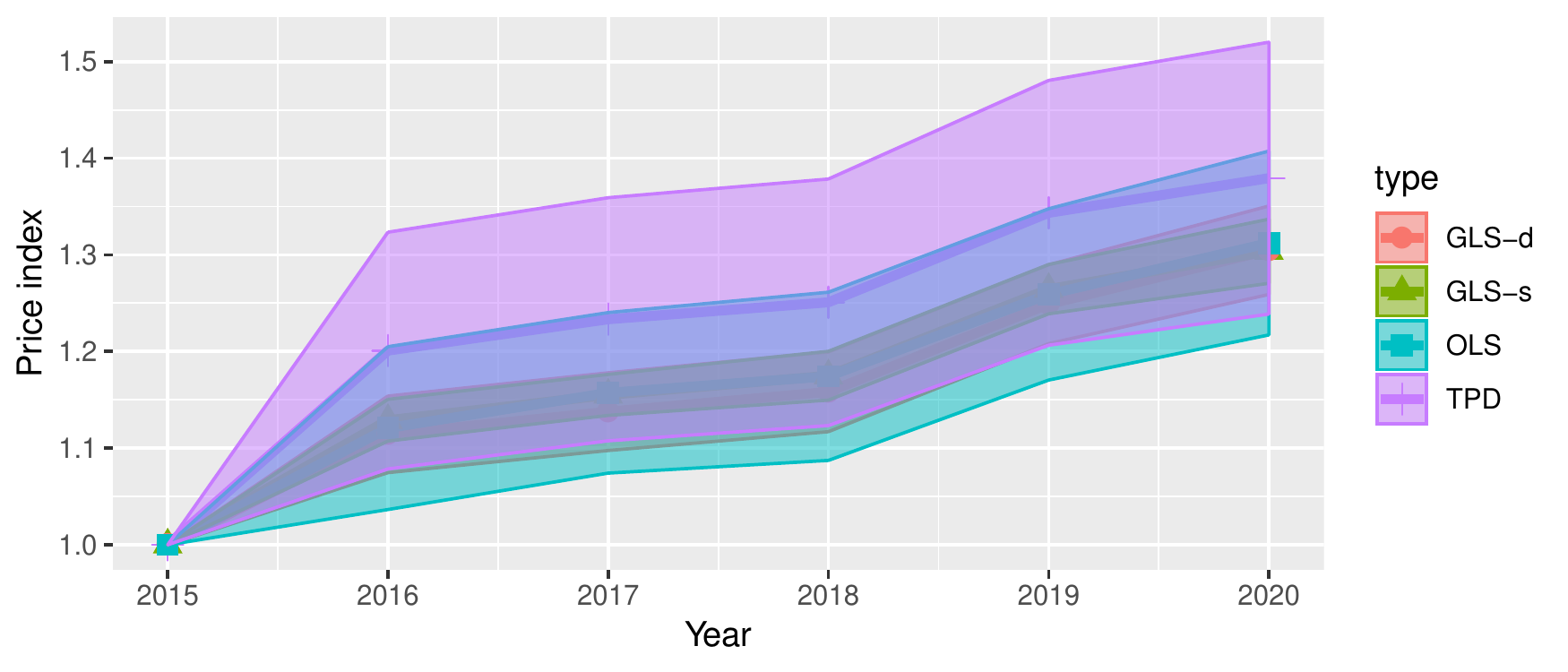} \includegraphics[scale=0.42]{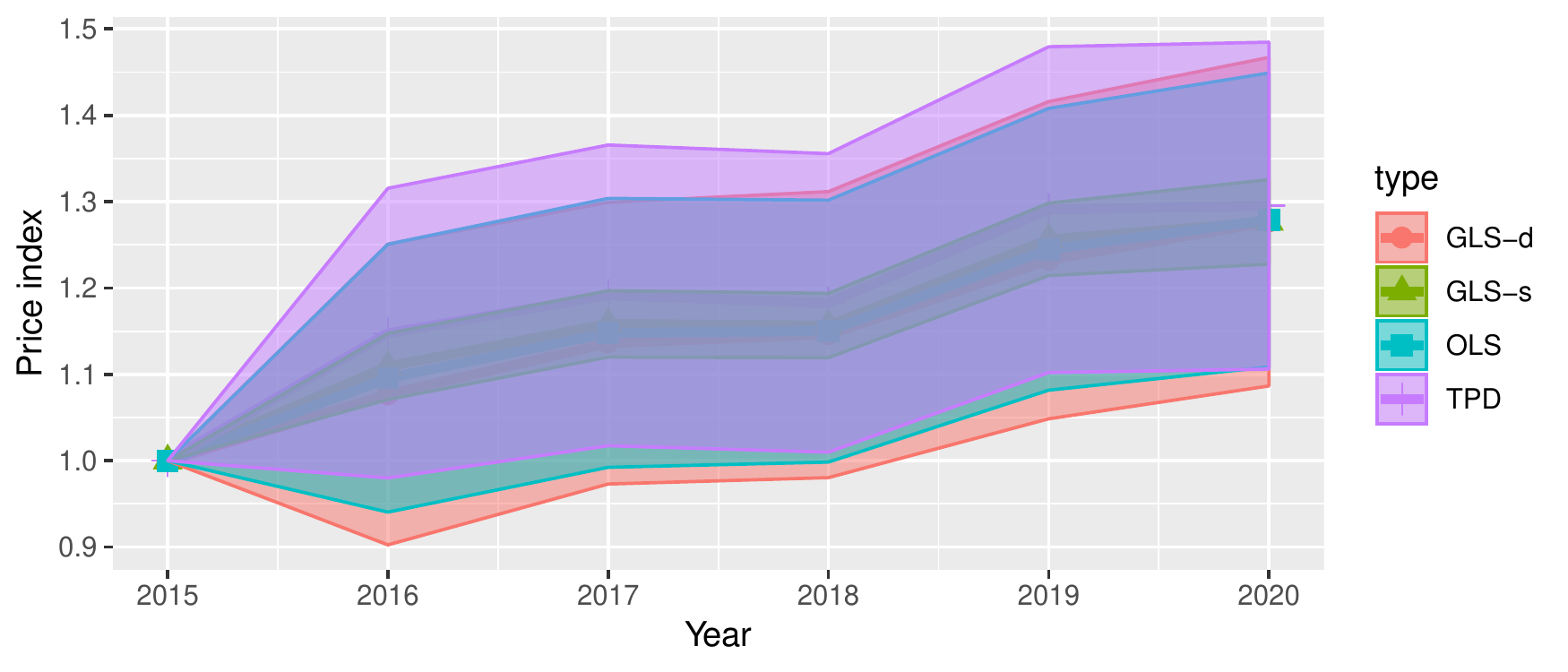}
\includegraphics[scale=0.42]{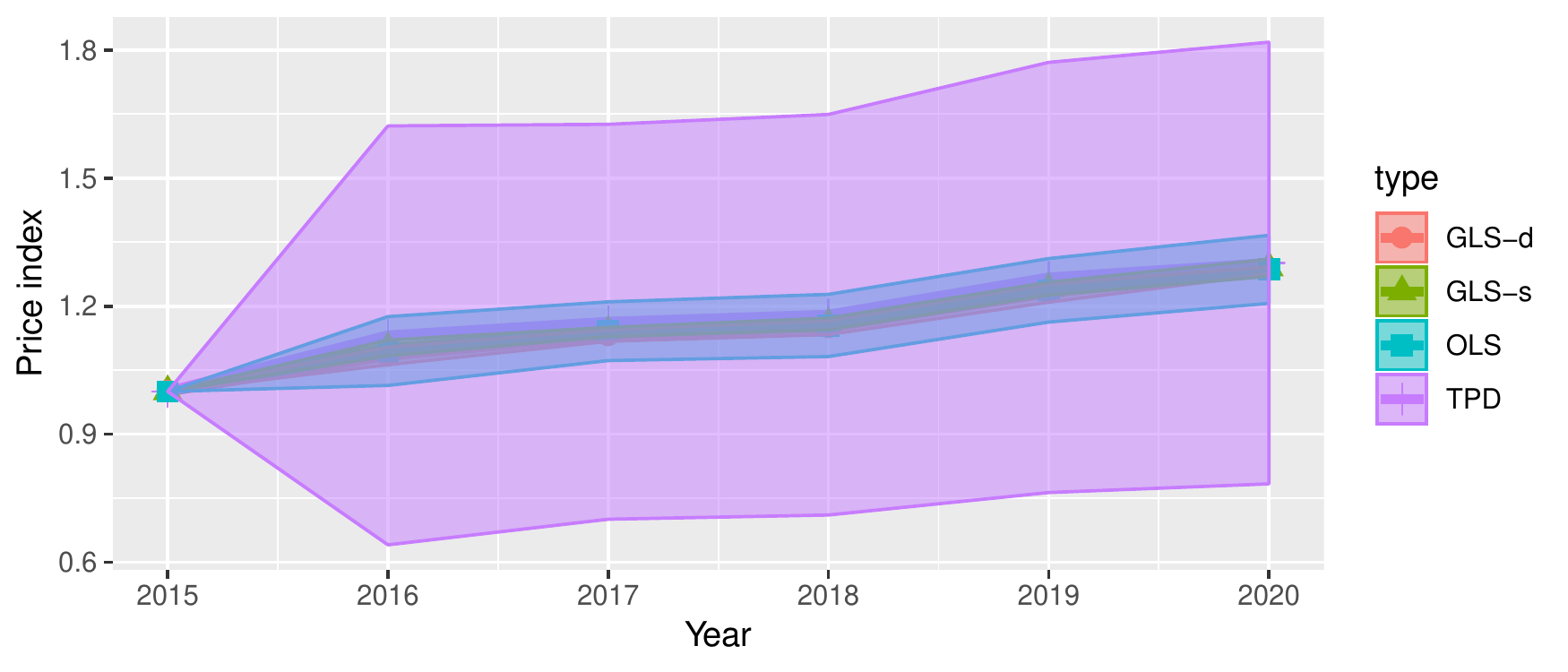}
\caption{Comparison of different versions of the  MPL index with the TPD one for each of the three cases considered in the simulation. Confidence bands are also plotted.  
}
\label{fig:5sim2}
\end{center}
\end{figure}

\begin{table}[h]
\center
\caption{Sum of squares of the differences between the estimated reference prices and the real ones.}
\label{tab:perform_ref}
\begin{tabular}{lrrrr}
\hline
\hline
Data                                                               &  & GLS-s   & GLS-d   & OLS      \\
\hline
\hline
\begin{tabular}[c]{@{}l@{}}1. Complete price tableau\\  (4 commodities)\end{tabular}  &       & 0.0152 & 0.0134 & 0.0165 \\
\begin{tabular}[c]{@{}l@{}}2. Incomplete price tableau\\  ``classical'' basket (2 commodities)\end{tabular} &       & 0.0016        & 0.0005   & 0.0008         \\
\begin{tabular}[c]{@{}l@{}}3. Incomplete price tableau\\  novel basket (4 commodities)\end{tabular} &       & 0.0083        & 0.0038   & 0.0055         \\
\hline
\hline
\end{tabular}
\end{table}

\begin{figure}[htbp]
\begin{center}
\includegraphics[scale=0.42]{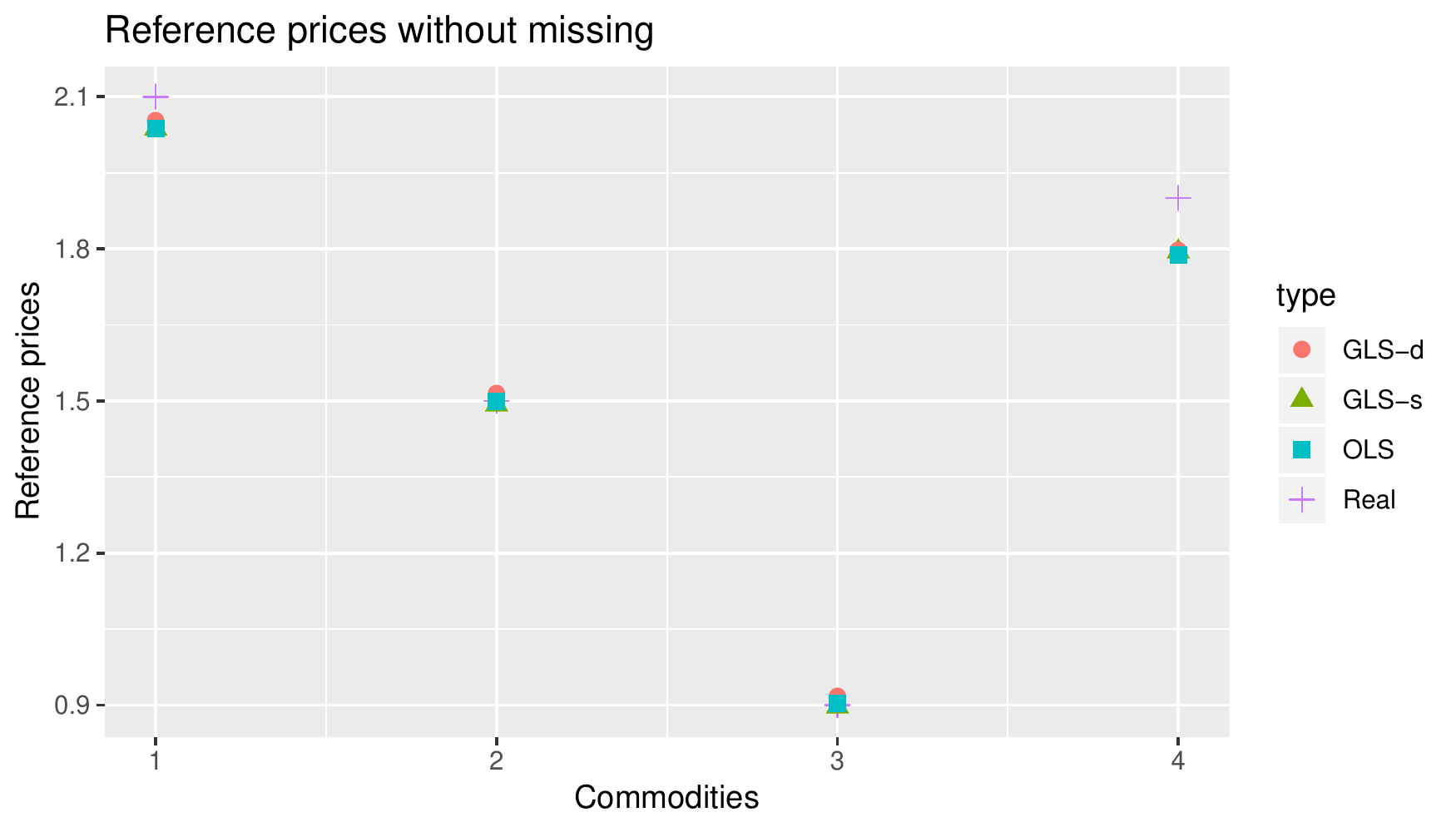} 
\includegraphics[scale=0.42]{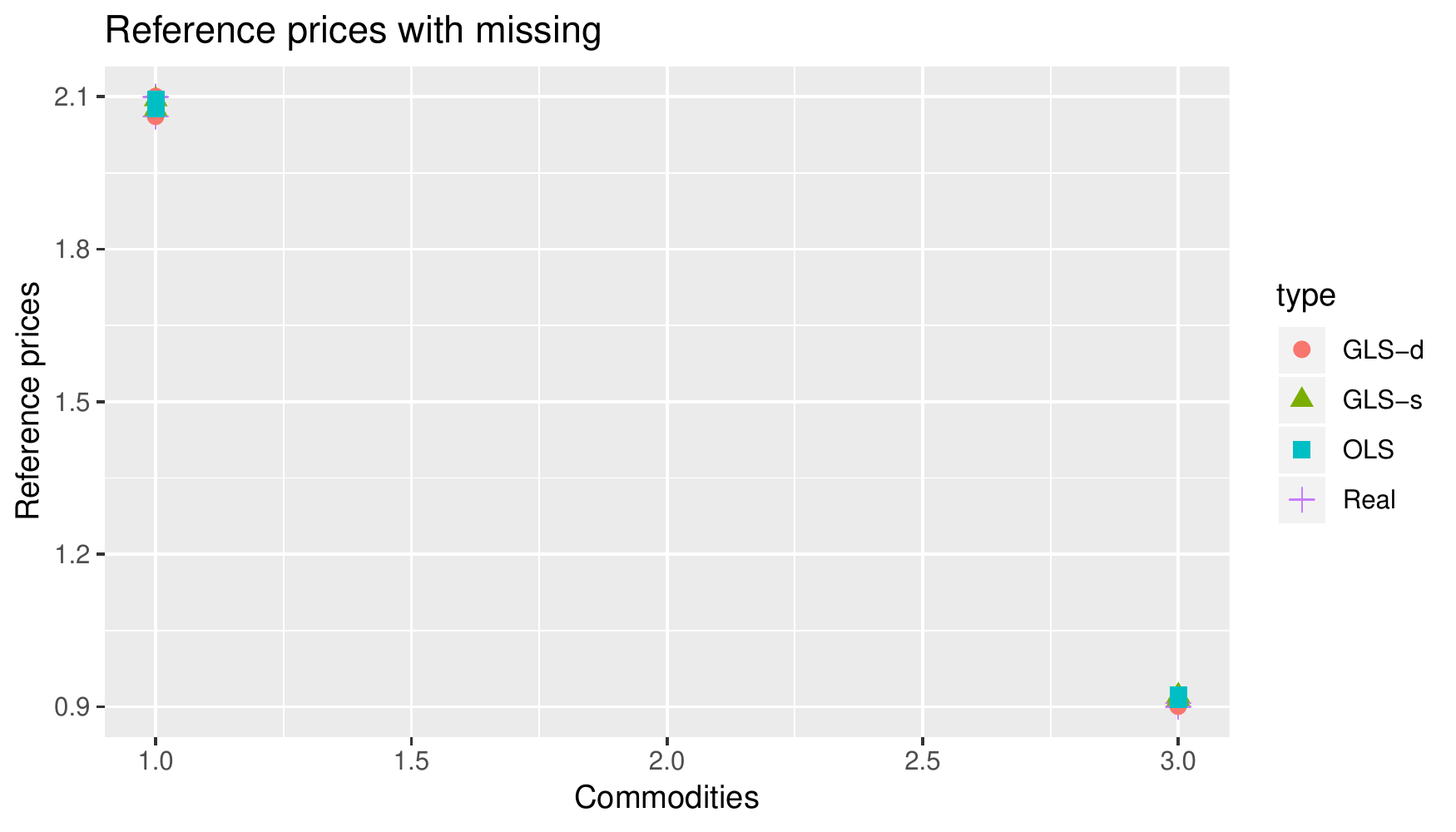}
\includegraphics[scale=0.42]{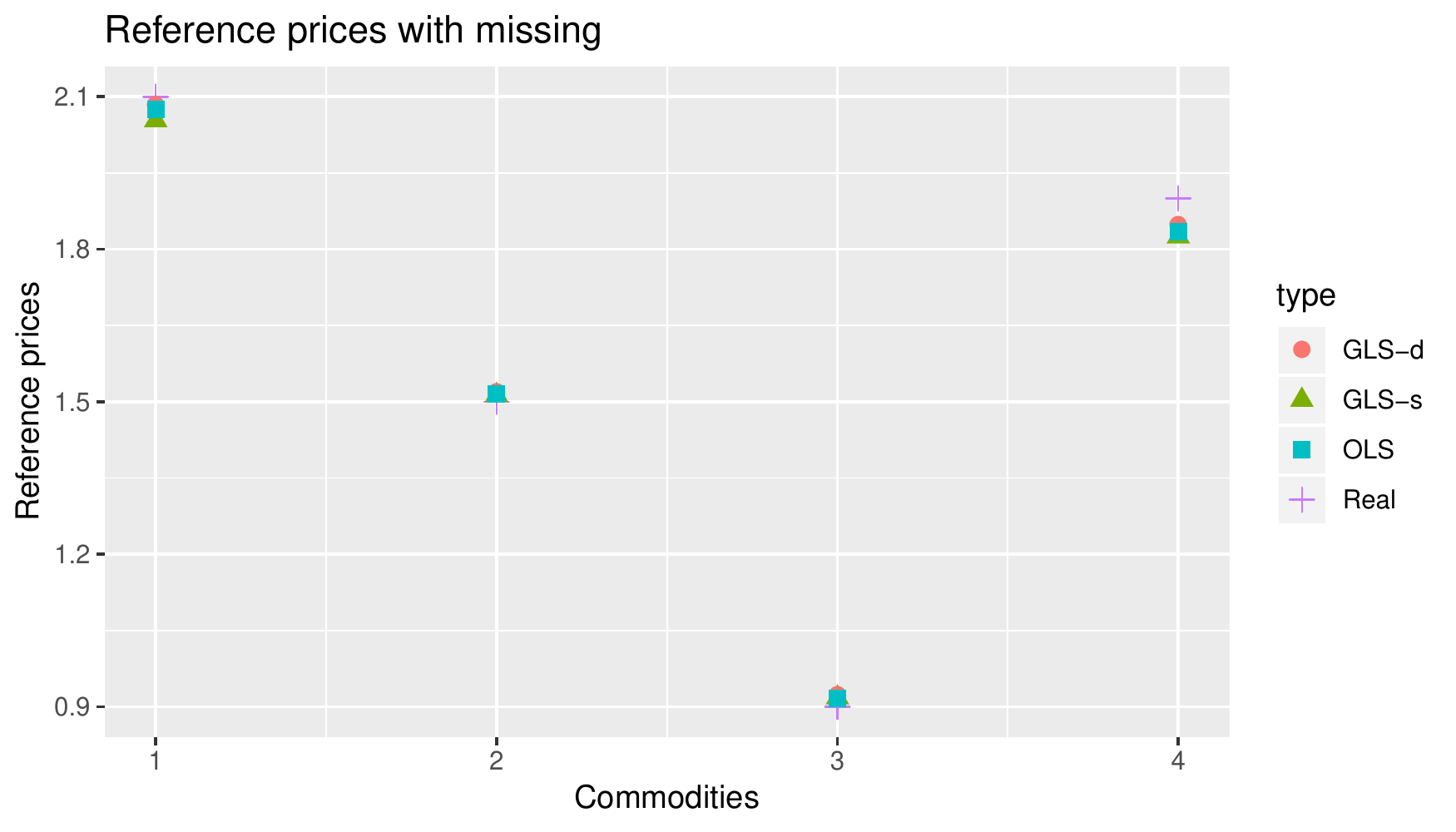}
\caption{Comparison between the reference prices estimated by the MPL index, under different assumptions for the error specification, and the real reference prices.  
The left panel refers to the case of complete price tableau, the right and central panels to the case of incomplete price tableau: with ``standard'' reference basket the former and with the MPL reference basket the latter. 
}
\label{fig:6sim2}
\end{center}
\end{figure}

\begin{figure}[htbp]
\begin{center}

\includegraphics[scale=0.42]{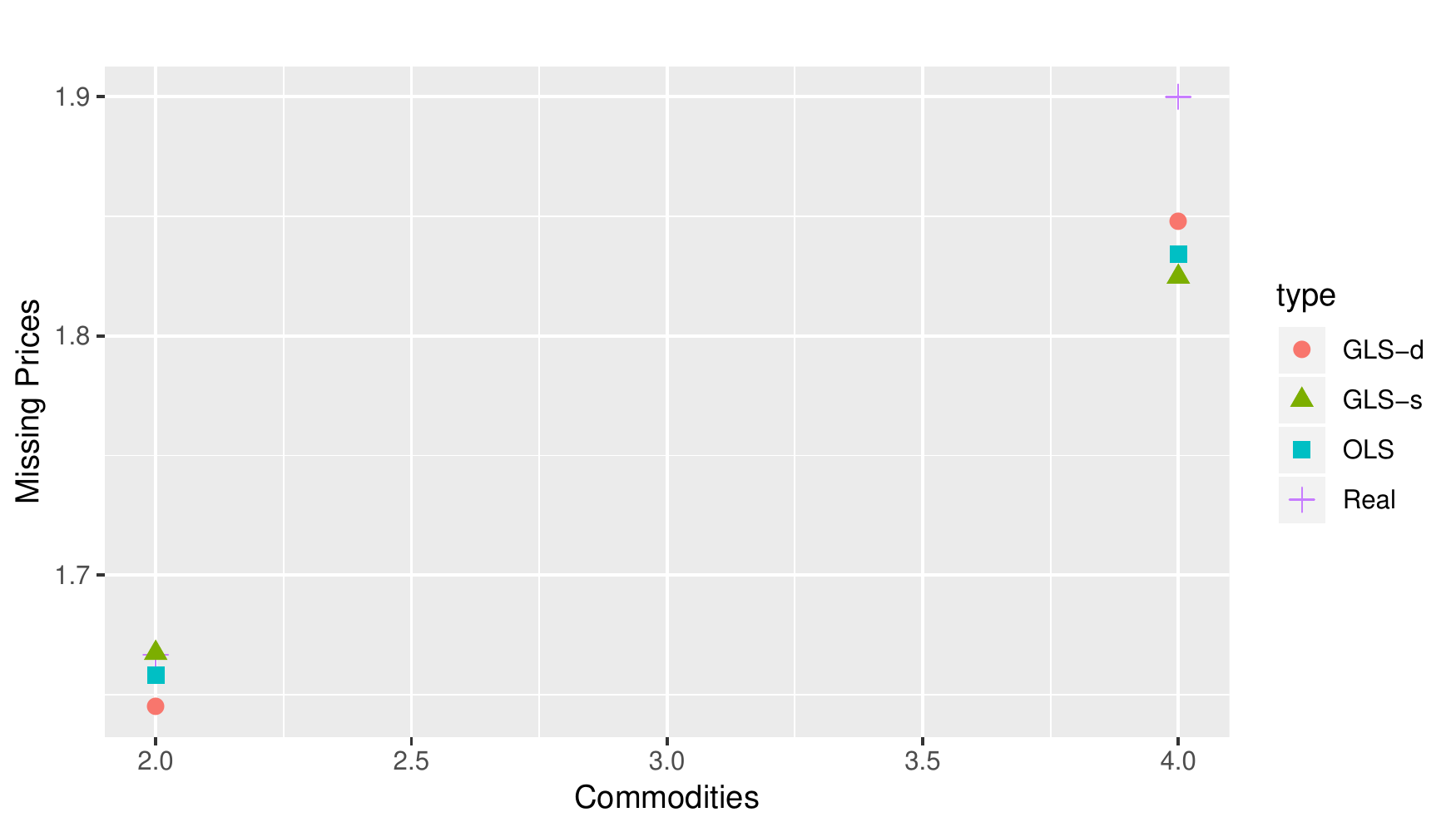}
\caption{ Comparison between the real and  estimated missing prices of the second and third commodities, obtained with the MPL index under different specifications of the error terms,  
}
\label{fig:6sim3}
\end{center}
\end{figure}

\begin{table}[htb]
\center
\caption{Sum of squares of the differences between the real prices and the estimated ones with the TPD and MPL index. For the latter different specifications of the error terms have been considered  }
\label{tab:perform_b}
\begin{tabular}{lrrrrr}
\hline
\hline
Data                                                               &  & GLS-s   & GLS-d   & OLS     & TPD     \\
\hline
\hline
\begin{tabular}[c]{@{}l@{}}1. Complete price \\ tableau\end{tabular}  &       & 0.1337 & 0.0940 & 0.1324 & 0.7602 \\
\begin{tabular}[c]{@{}l@{}}2. Incomplete price\\ tableau (``classical basket'')\end{tabular} &       & 0.0243        & 0.0291        & 0.0227        &   0.0661    \\
\begin{tabular}[c]{@{}l@{}}3. Incomplete price \\ tableau (novel basket)\end{tabular} &       & 0.0772        & 0.0660        & 0.0724        &   0.1471     \\
\hline
\hline
\end{tabular}
\end{table}

%
%
%

%
\end{document}